%% file: main.tex
\documentclass[letterpaper,twocolumn,10pt]{article}
\usepackage{usenix2020}

\usepackage[english]{babel}
\usepackage{amsthm}
\usepackage{balance}
\usepackage{graphicx}
\usepackage{multirow}
\usepackage{paralist}
\usepackage{caption}
\usepackage{subcaption}
\usepackage{xspace}
\usepackage{booktabs}
\usepackage{listings}
\usepackage{hyperref}
\usepackage{listings}
\usepackage{enumitem}
\usepackage{siunitx}
\usepackage{breqn}
\usepackage[breakable,theorems,skins]{tcolorbox}

\usepackage{algorithmic}
\usepackage{clipboard}

\xdefinecolor{attackscolor}{RGB}{139,0,0}
\newcommand*\attacklabel[1]{{\small \color{attackscolor} #1}}

\newcommand{\myitem}[1]{\vspace*{0.07in}\noindent\textbf{#1}}
\newcommand{\myitemit}[1]{\vspace*{0.07in}\noindent\textit{#1}}

\newcommand{\system}{\textsc{StarGlider}\xspace}
\newcommand{\approach}{\textsc{StarGlider}\xspace}

\newtheorem{property}{Property}
\newtheorem{corollary}{Corollary}
\newtheorem{mylemma}{Lemma}
\newtheorem{mytheorem}{Theorem}
\newtheorem{assumption}{Assumption}

\DeclareRobustCommand{\mygraybox}[2][gray!10]{%
	\begin{tcolorbox}[   
		breakable,
		left=0pt,
		right=0pt,
		top=0pt,
		bottom=0pt,
		before skip=5pt,
		after skip=5pt,
		colback=#1,
		colframe=#1,
		enlarge left by=0mm,
		boxsep=2pt,
		arc=0pt,outer arc=0pt,
		]
		#2
	\end{tcolorbox}
}

\setlist[itemize]{align=parleft,left=0pt..1em,topsep=0.05in,noitemsep}

\begin{document}

\date{}

\title{Reliable Low-Delay Routing In Space with Routing-Oblivious LEO Satellites}

\author{
Stefano Vissicchio, Mark Handley\\
University College London (UCL)
}

\maketitle

\begin{abstract}
Large networks of Low Earth Orbit (LEO) satellites are being built using
inter-satellite lasers.  These networks promise to offer low-latency
wide-area connectivity, but reliably routing such traffic 
is difficult, as satellites are very resource-constrained and paths
change constantly.

We present \system, a new routing system where path computation is delegated
to ground stations, while satellites are routing-oblivious and exchange no
information at runtime.
Yet, \system satellites effectively support reliability primitives:
they fast reroute packets over near-optimal paths when links fail,
and validate that packets sent by potentially malicious
ground stations follow reasonable paths.
\end{abstract}

\maketitle

\input{intro}

\input{background}

\input{problem}

\input{overview}

\input{theory-summary}

\input{internals}

\input{eval}

\input{discussion}

\input{relwork}

\section{Conclusions}

Source routing is generally considered a
no-go in the Internet, mainly for robustness and security reasons.
LEO constellations are bigger, much more dynamic, and generally less
reliable than Internet networks; satellites are also
much more resource-constrained than Internet routers.
Yet, our work shows that source routing is a promising option
to support premium services in LEO networks, \emph{robustly}
and \emph{securely}.

We indeed designed and evaluated \approach, a system addressing
reliability concerns in LEO mega-constellations with routing-oblivious satellites.
In \approach, our novel routing theory provides a theoretical basis to
\begin{inparaenum}[(i)]
\item define an expressive and compact path encoding; and
\item devise effective fast rerouting and path validation schemes,
runnable by routing-oblivious satellites.
\end{inparaenum}
Our simulations confirm that \approach uses very few satellite resources,
and is more effective than much more
resource-hungry techniques commonly used in the Internet.

\bibliographystyle{plain}
\bibliography{main}

\clearpage
\appendix
\input{theory-proofs}

\end{document}

%% file: intro.tex
\section{Introduction}

LEO constellations promise to support communication between ground stations (GSes) located potentially anywhere on Earth.
Several companies are planning to build such constellations, each with thousands of LEO satellites~\cite{fcc-spacex,fcc-oneweb,fcc-telesat,fcc-norway}.
Satellite launches have started~\cite{spacex-launches}.
SpaceX Starlink already offers Internet connectivity through its partially deployed constellation, although only for short-distance communication and with unreliable performance~\cite{obo-imc22,mellia-imc22}.

Inspired by prior research~\cite{eth-hotnets18,mark-hotnets18,mark-hotnets19}, we consider wide-area low-latency connectivity that fully deployed constellations can offer as a premium, high-revenue service.
Satellites currently being launched support laser inter-satellite links (ISLs).
This opens the possibility to route traffic over satellite paths with delay lower than can be achieved with optical fiber.

\begin{figure}[t]
	\centering
	\includegraphics[width=0.95\columnwidth]{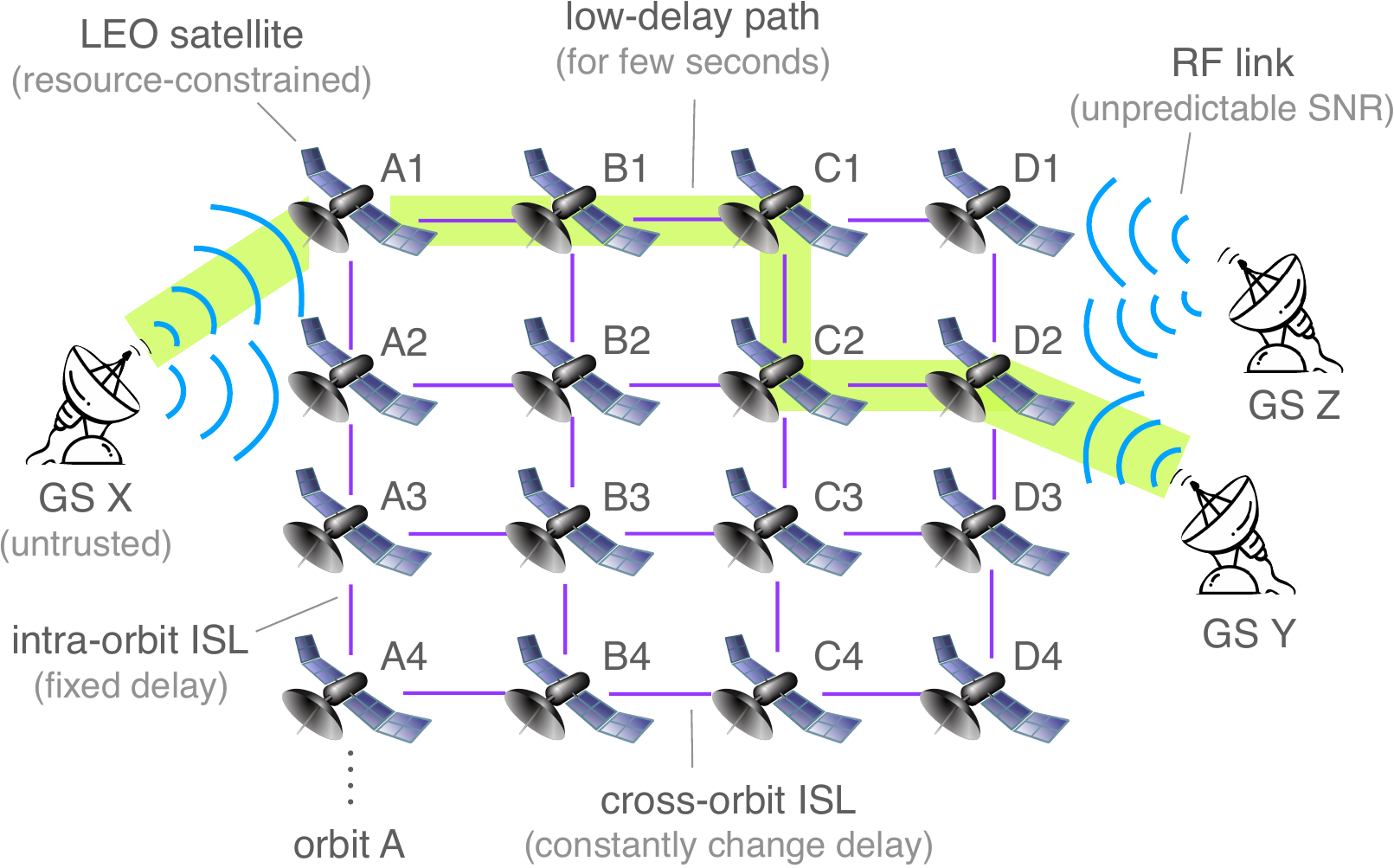}
	\caption{Schematization of a LEO satellite network supporting a wide-area premium service: path $(A1, B1, C1, C2, D2)$ provides lower delay than can be achieved with optical fibers.}
	\label{fig:basic-setup}
\end{figure}

Consider Figure~\ref{fig:basic-setup}.
GSes are connected to one or few satellites via RF (radio) links.
Assuming that GS Z has a wired connection to the Internet, current Starlink's service may allow GS Y to communicate with GS Z via satellite D2, and use GS Z as an Internet proxy.
We instead focus on long-range connectivity (e.g., between GS X and GS Y) over low-latency multi-satellite paths, such as the one highlighted in the figure.

Supporting wide-area connectivity in LEO constellations is hard.
In these networks, the lowest-latency paths change every few seconds.
As satellites orbit around the Earth, the satellites reachable from any given GS change continuously, and so do length and delay of ISLs as well as availability and signal quality of RF links.
Constellations contain many more satellites than existing routing protocols (e.g., OSPF) are designed for.
Finally, satellites are power-constrained: they are solar powered, but need to maintain many simultaneous RF spot beams to GSes, drive four or more ISLs over thousands of kilometers, and possibly perform in-network computations~\cite{oec-asplos20}, all while forwarding packets.
With power at a premium, satellites may simply not be able to perform heavy routing computations. 

In fact, SDN architectures have been deployed~\cite{loon-sigcomm22} for non-terrestrial networks with requirements similar to LEO constellations.
This design delegates path computation to controllers located in on-ground data centers.
Fundamentally, however, controllers cannot guarantee fast response to \emph{unexpected} failures and environmental changes, a concern exacerbated by the increasing popularity of mobile GSes.
For a premium low-latency service, source GSes cannot interrupt connectivity and wait for instructions from controllers whenever their paths are disrupted by an ISL failure or new appearing obstacles.

Delegating routing tasks to premium GSes seems the best choice.
GSes are in an optimal position to compute the shortest paths, given that each GS knows its own \emph{current} RF situation.
Given that satellite orbits are predictable, GSes can also be equipped with enough resources to frequently recompute best paths, using algorithms such as~\cite{mark-hotnets19}.
If each GS source-routes premium traffic it originates, the approach would also naturally scale with the number of customers.

Can we however implement a \emph{reliable} low-latency routing service if GSes source-route premium traffic?
There are two main concerns.
First, we still need to quickly react to unexpected failures that source GSes cannot locally detect.
Our simulations show that it takes up to $\approx$ 100 ms for GSes just to be notified about a failed ISL or satellite, because of GS-to-satellite path propagation delays: we simply cannot drop the premium packets source-routed over disrupted paths for that long.
Second, contrary to SDN controllers, GSes cannot be fully trusted:
source-routing GSes can be compromised and used for attacks such as topology discovery, privilege escalation at specific satellites, or link congestion.

In this paper, we consider an extreme source-routing architecture where constellations hold a \textbf{routing-free core}, meaning that satellites exchange \emph{no} routing information, keep \emph{no} routing table, and perform \emph{no} path computation.
While convenient to minimize satellite resources, this design rules out any possibility to re-use existing techniques (e.g.,~\cite{lfa-rfc5286,mplsfrr-rfc4090})
in order to mitigate the above reliability concerns, as elaborated in \S\ref{sec:problem}.

Yet, we show that constellations with a routing-free core can actually offer the envisioned premium service, reliably.
We present \approach, a system where routing-oblivious satellites both fast reroute failure-affected traffic over low-delay paths, and validate source-routed packets (\S\ref{sec:overview}).
\approach leverages a novel routing theory that characterizes the shape of low-delay paths in LEO constellations irrespectively of exact ISL delays (\S\ref{sec:theory}).
Based on our theory, we devise a compact but expressive path encoding used by GSes to specify source-routed paths, and by satellites to check, forward and fast reroute packets with minimal computation (\S\ref{sec:internals}).

Our simulations confirm \approach effectiveness and resource efficiency (\S\ref{sec:eval}).
Despite requiring virtually no resources on satellites, \system fast reroutes packets over near-optimal paths, typically shorter than state-of-the-art alternatives.
Also, \system's packet validation is both more accurate and 1000x faster than checking path delays.
Our experiments show how such validation prevents satellite-specific and link flooding attacks, while also increasing costs and detectability of DDoS ones.

%% file: background.tex
\section{Background: Satellite Networks}
\label{sec:background}

\begin{figure}[t]
	\centering
	\includegraphics[width=0.7\columnwidth]{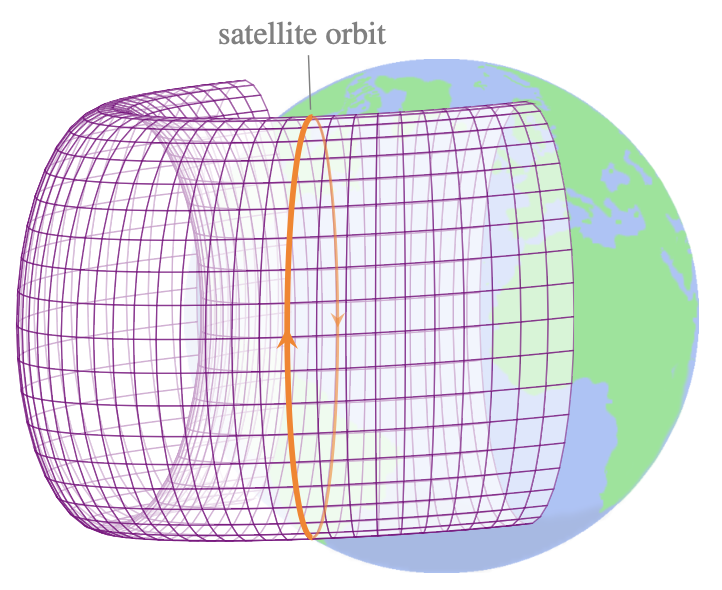}
	\caption{Illustration of a section of the torus formed by inter-satellite links (purple). Note that satellites' orbits are not perpendicular to the equator. Also, in reality, cross-orbit ISLs become longer when they get closer to the equator.}
	\label{fig:isl-torus}
\end{figure}

A LEO constellation includes thousands of satellites that constantly move in space along predefined \emph{orbits}.
Satellite orbits can be represented as paths on the surface of a torus -- see Figure~\ref{fig:isl-torus}.
Each satellite moves towards the north pole (i.e., north-east in Starlink) half of the time, and towards the south pole (i.e., south-east in Starlink) the other half.
In this paper, we use constellation and satellite network as synonyms.

All LEO satellites are at the same altitude and have the same speed
\footnote{Technically, constellations may include several shells, each with thousands of satellites.
Satellites in different shells are \emph{not} at the same altitude, do \emph{not} have the same speed, and maintain \emph{no} fixed relationship over time.
Likely, links between satellites in different shells are \emph{not} technically feasible.
We thus consider different shells as different satellite networks.}.
If two satellites are in the same orbital plane, the distance between them remains mostly constant, with only small variations deriving from operational constraints~\cite{leo-operation-mobicom23}.
If two satellites are in neighboring orbital planes, the distance between them predictably varies over time, being as much as \emph{twice longer} at the equator than when the satellites are at the most northerly or southerly part of their orbits.

Satellites maintain laser ISLs with each other.
We focus on constellations where each satellite is connected to
\begin{inparaenum}[(i)]
\item two neighboring satellites in the same orbital plane with two {\em intra-orbit ISLs}, and
\item one satellite in a neighboring orbital plane to the East and one to the West with two {\em cross-orbit ISLs}.
\end{inparaenum}
Because of the above orbital dynamics, the length of intra-orbit ISLs can be approximated as fixed,
while the length and direction of cross-orbit ISLs changes (predictably) over time.

We denote the arrangement of ISLs in a constellation as network \emph{topology}.
The topologies we focus on are consistent with the ones assumed by most prior work~\cite{eth-hotnets18,mark-hotnets18}.
We discuss extensions of \system to other topologies and relaxations of our assumptions in~\S\ref{sec:discussion}.

Ground stations, such as user terminals and SpaceX Internet proxies, are located on Earth.
They are the sources and destinations of users' packets, and are connected to satellites via \emph{RF links}.
At any time, a GS can maintain RF links, and hence exchange packets, with satellites in its communication range at that time.
GSes typically have a few antennas, so they can maintain a limited number of RF links simultaneously. 

\clearpage

%% file: problem.tex
\section{Problem Statement}
\label{sec:problem}

We aim at supporting wide-area low-delay connectivity through a satellite network as a premium service for selected customers.
The service is akin to pseudo-wires: \textbf{specific pairs of GSes} (e.g., GS X and Y in Figure~\ref{fig:basic-setup}) are allowed to exchange premium traffic over low-delay paths, at a contractualized, low rate -- e.g., in the order of 100Mbps.
We assume that GSes efficiently compute low-delay paths, and source-route premium packets.
Also, satellites prioritize the relatively few, low-rate premium flows over all the others, so \textbf{premium packets from legitimate GSes never cause congestion}\footnote{Our reliability problem covers \emph{malicious} attempts to cause congestion.}.

A premium service needs to be offered reliably, ideally ensuring that failures and malicious users do \emph{not} interrupt connectivity, cause packet loss, significantly inflate delay, etc.

Resources on satellites are extremely scarce, especially power.
Hence, \textbf{satellites cannot run existing routing protocols} commonly used in Internet Service Providers, such as IGPs (e.g., OSPF~\cite{ospf-rfc}, IS-IS~\cite{isis-rfc}) or MPLS.
Even on networks much less dynamic than LEO constellations, running such protocols requires significant memory and processing power at each node -- e.g., for storing the network graph, updating it, sharing updates with other nodes, computing and potentially signaling new paths.
In fact, current best practices (e.g., see~\cite{cisco-arch11}) suggest to run those protocols within no more than 50 routers, breaking down larger networks into areas of that size when needed.
LEO constellations include thousands of satellites, whose continuous movement would require constant updates of the network graph, and also prevents any meaningful sub-division of the network into areas.

Our goal is instead to ensure service reliability within a design where satellites do not compute paths, do not hold routing tables, and do not exchange dynamic information (e.g., on health and delay of ISLs).
Satellites can only keep static data (e.g., on the network topology and satellites' orbits).

\subsection{Reactivity to failures}

As any network, LEO constellations are subject to unpredictable failures.
For example, hardware malfunctioning may affect forwarding at satellites; ISLs may be temporarily disrupted when they closely align with the sun or if satellites slightly deviate from their orbits; RF links may be subject to signal interference caused by other transmitting stations, physical obstacles or suddenly changed weather conditions.
We have no data on how frequent these failures are in LEO constellations, partially because these networks are still being built.
As a reference, though, fast failure recovery is a problem still attracting vast research and industrial efforts in all Internet networks, from data centers to ISPs -- see, e.g.,~\cite{chiesa-surveyfrr-cmst21}.

Simply relying on path re-computation at GSes is not sufficient within a source-routed constellation.
Satellites next to a failure can quickly detect the failure, with techniques similar to~\cite{bfd-rfc5880}.
After detection, though, GSes need to be informed about failures and recompute paths.
In the meanwhile, all the packets sent over disrupted paths would be lost.
For how long?
Communicating ISL failures to GSes takes up to $\approx$ 100ms, just because of path propagation delays.
Computing new paths at GSes is not instantaneous either: state-of-the-art algorithms~\cite{mark-hotnets19} take a few milliseconds per destination in our simulations.
This is a long time for a premium service!
We therefore need satellites to locally reroute premium traffic.

\myitem{Prior work cannot be re-used.}
Resource constraints prevent existing fast-rerouting techniques from being used in LEO satellites.
Clearly, pre-installing backup paths at each satellite for any possible destination GS does not scale.
Also, fast rerouting (FRR) approaches widely used in the Internet, such as IGP Loop Free Alternate (LFA)~\cite{lfa-rfc5286} and MPLS FRR~\cite{mplsfrr-rfc4090}, rely on intra-domain routing protocols that cannot be run in satellite networks (see above).
The network size and satellites' resource constraints prevent static failover techniques (e.g.,~\cite{chiesa-staticfrr-ton17}) from being adopted in LEO constellations.
These techniques are based on decomposing the network graph into $k$ spanning trees per destination, where $k$ is the number of edge-disjoint paths between any pair of nodes in the graph.
LEO constellations include hundreds of edge-disjoint paths between any pair of nodes: satellites cannot compute, store and reroute packets over so many spanning trees.
Finally, approaches based on exploring new paths upon failures (e.g., by adding information to packet headers~\cite{shenker-failurecarryingpackets-sigcomm07}) provide no guarantees on the path stretch, which is a critical limitation for our low-latency connectivity service.

\myitem{Problem.}
We therefore focus on the following question:
\mygraybox{
  Can we enable satellites to quickly reroute packets on low-latency paths, without needing any routing information and significant computation?
}

\subsection{Security}

A second key facet of reliability is security.
Even if every GS runs a trusted computing environment, it is likely that GSes can be compromised.
In fact, vulnerabilities of current Starlink terminals have already been exposed~\cite{starlink-attack}.
Attackers can also source packets from their own devices, given that GSes communicate with satellites over radio links.
For example, in Figure~\ref{fig:basic-setup}, an attacker can compromise GS X or deploy her own device next to it: in both cases, she can send packets to satellites A1 and A2 impersonating the legitimate GS X.

We focus on threats peculiar to our routing-free core architecture.
Prominently, they do not include physical-layer attacks, such as jamming RF links, that are fully orthogonal to our design.
We also do not consider attempts to exceed the allowed premium traffic rate.
We indeed assume a centralized admission control system that easily detects and mitigates rate-based attacks.
For example, such a system can enforce that for each GS pair, only one ingress satellite accepts premium packets, up to the contractualized rate; source GSes can then schedule changes of ingress satellites for their premium traffic by interacting with the admission system.

\myitem{Attacker model.}
We consider an attacker who controls a botnet of GSes.
Each compromised GS can forge arbitrary packets, and send them to any reachable satellite.

The general goal of the attacker is to source-route packets over paths of her choice, in order to enable a variety of attacks including: 
\begin{inparaenum}
\item[\attacklabel{$\mathcal{A}$1.}]\label{item:attack1} gathering information about the deployed satellites (e.g., hardware, software and possible vulnerabilities), ISLs (e.g., network topology, and per-link health, bandwidth and delay), and RF links (e.g., their activation/deactivation over time to reverse engineer GSes' satellite selection algorithms); 
\item[\attacklabel{$\mathcal{A}$2.}]\label{item:attack2} vulnerability exploitation at specific satellites, for example, for privilege escalation;
\item[\attacklabel{$\mathcal{A}$3.}]\label{item:attack3} network service degradation, for instance by concentrating high volumes of traffic on target links in specific geographical regions;
\item[\attacklabel{$\mathcal{A}$4.}]\label{item:attack4} indirect attacks to legitimate users of the compromised GSes, to increase their costs, cause reputation damages or induce the satellite network to block legitimate traffic.
\end{inparaenum}

\myitem{Simple defenses are not enough.}
GS authentication alone is not sufficient, as it would not help at all against attackers who compromise GSes without imparing their authentication -- or those who break the authentication mechanism itself.
A more promising approach consists in complementing authentication with the validation of source-routed packets at satellites.

Simple packet validation strategies, however, do not work either.
Using a centralized validator (e.g., a network controller or a scrubber) would delay the forwarding of legitimate packets to be validated by hundreds of milliseconds, which is not viable for a low-delay connectivity service.
Computing valid paths outside satellites and caching them in satellites is impractical: for example, it is unclear
\begin{inparaenum}[(i)]
\item how many paths would any satellite have to cache per premium GS pair, e.g., towards how many possible destination satellites;
\item how quickly can the set of valid paths per GS pair be updated, for example, after failures;
\item how to protect against attackers requiring paths for more and more GS pairs to be cached on a single satellite.
\end{inparaenum}
Precomputing paths or ISL delays is hard to scale to huge, highly dynamic networks, and simply does not work in the presence of failures.
Finally, satellites cannot rely on routing protocols or local path computations to check the delay of source-routed paths (see discussion above).

\myitem{Problem.}
We therefore focus on the following question:
\mygraybox{
  Can we enable satellites to cheaply validate source-routed packets, on the fly, without needing any routing information and significant computation?
}

Note that packet validation should not be too strict, but permit the use of paths slightly longer than the theoretical shortest ones -- e.g., to let GSes deal with failures, optimize across multiple destinations, or perform multi-path routing.

\newpage

%% file: overview.tex
\section{\approach Overview}
\label{sec:overview}

\begin{figure}[t]
	\centering
	\includegraphics[width=0.95\columnwidth]{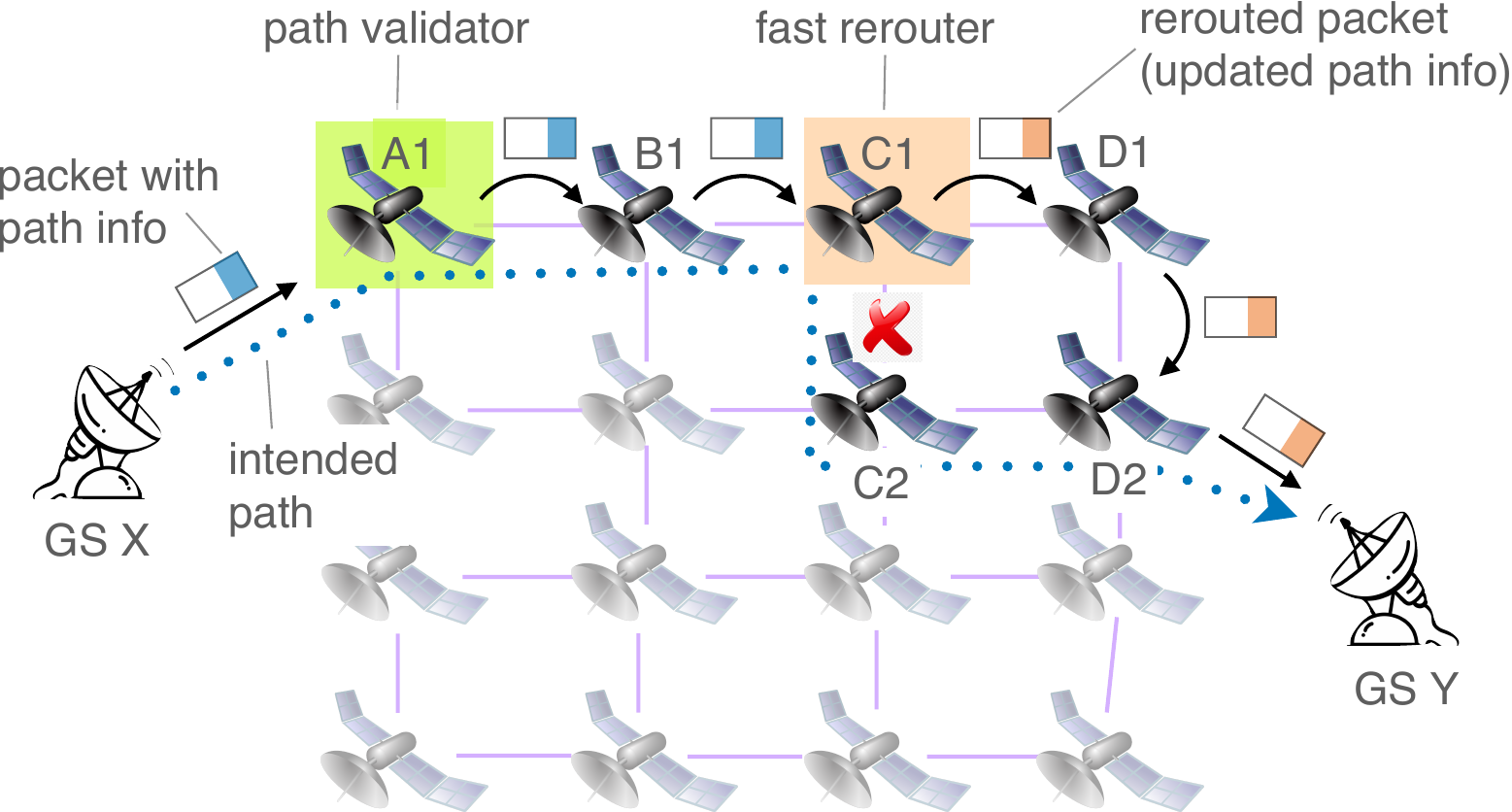}
	\caption{\approach overview. The ingress satellite A1 validates the packet from GS X. Being next to the failed link (C1, C2), C1 reroutes the packet using path information in it.}
	\label{fig:overview-example}
\end{figure}

\approach enables routing-oblivious satellites to validate and fast reroute source-routed packets, possibly at line rate.

Figure~\ref{fig:overview-example} showcases \approach in action.
GS X has been authorized to send premium traffic to GS Y.
Consistently with Figure~\ref{fig:basic-setup}, X source-routes packets for Y over the low-delay path (A1,B1,C1,C2,D2).
To do so, X encodes this path in the header of the packets it sends to satellite A1.
Being the first hop, A1 validates X's packets: each packet is forwarded if its validation succeeds (as in Figure~\ref{fig:basic-setup}), and discarded otherwise. 
Packets are validated only at the ingress to minimize work across satellites as well as ISL bandwidth consumption.

Validated packets eventually reach C1.
Suppose that the link (C1,C2) has just failed.
Upon detection, C1 notifies all GSes about the failure.
While X recomputes its best paths, C1 reroutes packets around the failure by reordering
the intra- and inter-orbit links in the source-routed path.
So, instead of crossing the intra-orbit link (C1,C2) and then the inter-orbit one (C2,D2), C1 updates the encoded path to cross an inter-orbit link and then an intra-orbit one.
Fast-rerouted packets are therefore forwarded over the sub-path (C1,D1,D2).

For this approach to work, the path encoding must be:
\begin{itemize}
\item expressive enough for GSes to encode all low-delay paths in the satellite network;
\item informative enough for routing-oblivious satellites to fast reroute and validate packets with minimal processing;
\item cheap and easy to process by satellites, ideally in hardware;
\item compact, in order to minimize packet overhead.
\end{itemize}

To satisfy all these requirements, \approach path encoding relies on a new routing theory.
Our theory captures constellations' geometrical properties relevant to routing, so that the shape of the shortest path between any pair of satellites at any time can be determined by only knowing the orbits and topological positions of source and destination satellites.

We present our routing theory in \S\ref{sec:theory}.
We then detail \approach path encoding as well as fast-rerouting and validation algorithms built upon such encoding in \S\ref{sec:internals}.

%% file: theory-summary.tex
\section{\approach Theoretical Foundations}
\label{sec:theory}

Consistently with \S\ref{sec:background}, our theory applies to constellations where each satellite has two links with neighbors in its orbit and two cross-links with neighbors in adjacent orbits.
We only make two \textit{assumptions}:
\begin{inparaenum} [(i)]
	\item \label{ass:length-cross-links} the length of any cross-orbit link decreases with the distance of its midpoint from the equator; and
	\item \label{ass:no-mixed-links} paths with either only intra-orbit links or only cross-orbit links are always shorter than paths mixing intra- and cross-orbit links.
\end{inparaenum}
These assumptions hold in real constellations, as we confirm empirically (\S\ref{sec:eval}).   
We make no assumption about many other parameters such as orbital inclinations, the orientation of ISLs with respect to the equator, or the function by which cross-orbit links change length over time. 

Formal definitions and proofs are in Appendix~\ref{app:theory}.

\myitem{Grids.}
Let $s$ and $d$ be source and destination satellite respectively.
The two satellites divide the network graph into four sub-graphs that we call \emph{grids}.
Intuitively, each grid resembles a parallelogram moving on the surface of the torus formed by ISLs (see Figure~\ref{fig:isl-torus}), with $s$ and $d$ occupying opposite corners of the parallelogram.
Figure~\ref{fig:spacegrid-base} provides an illustration.

\begin{figure}[t]
	\centering
	\includegraphics[width=0.95\columnwidth]{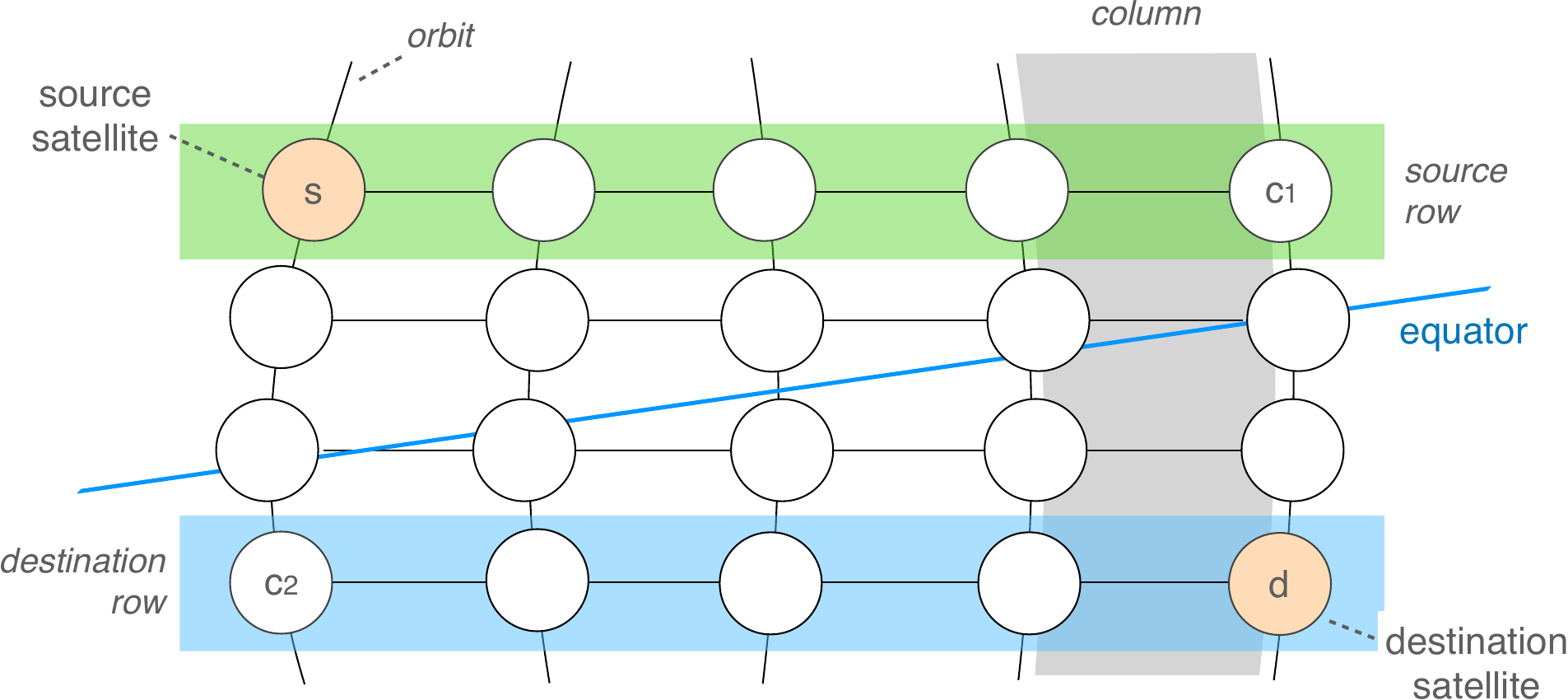}
	\caption{Schematization of a type-A $s$-$d$ grid. Note that cross-orbit links become longer as they get closer to the equator. For simplicity, the following figures visualize cross-orbit links as if they had the same length.
          }
	\label{fig:spacegrid-base}
\end{figure}

Given a grid $G$, we define a \textit{row} of $G$ as a sequence of cross-orbit links that constitute a path between two satellites on opposite sides of $G$.
A \textit{column} of $G$ is also a set of cross-orbit links, but they do not connect to each other.
Consider two neighboring sequences $S_1$ and $S_2$ of intra-orbit links that connect the source row to the destination row; a column is the set of cross-orbit links between the satellites in $S_1$ and $S_2$ (see Figure~\ref{fig:spacegrid-base}).
Finally, the \textit{border} of $G$ is the set of satellites in $G$ connected to either the grid's source or its destination via only intra-orbit
links, or via only cross-orbit links.

The following theorem proves a basic property of grids.

\Copy{theo:ingrid}{
\begin{mytheorem}\label{theo:ingrid-vs-outgrid}
	The shortest path from a source to a destination satellite stays within one of their grids whenever there are no link failures between satellites in the grid borders.
\end{mytheorem}
}

\begin{figure*}[th]
	\centering
	\begin{subfigure}[b]{0.9\columnwidth}%
		\centering
		\includegraphics[width=\columnwidth]{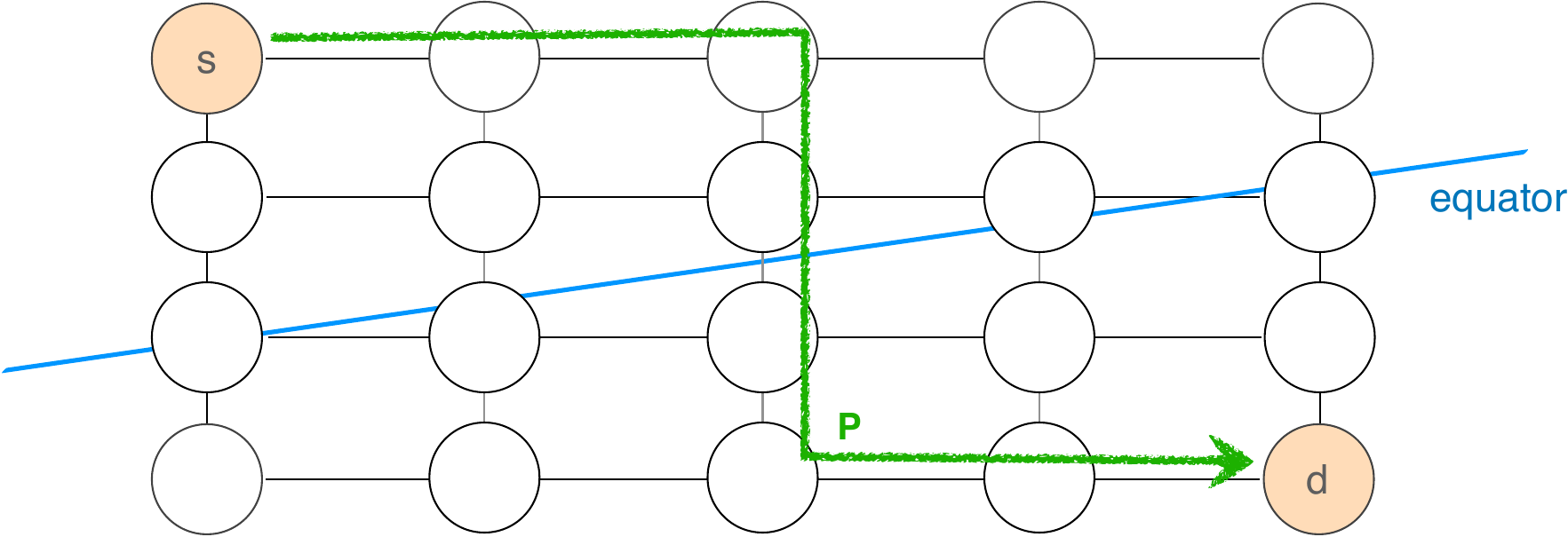}
		\caption{Type-A grids.}
		\label{fig:spacegrid-equator-typeA-sp}
	\end{subfigure}
	\qquad
	\begin{subfigure}[b]{0.9\columnwidth}%
		\centering
		\includegraphics[width=\columnwidth]{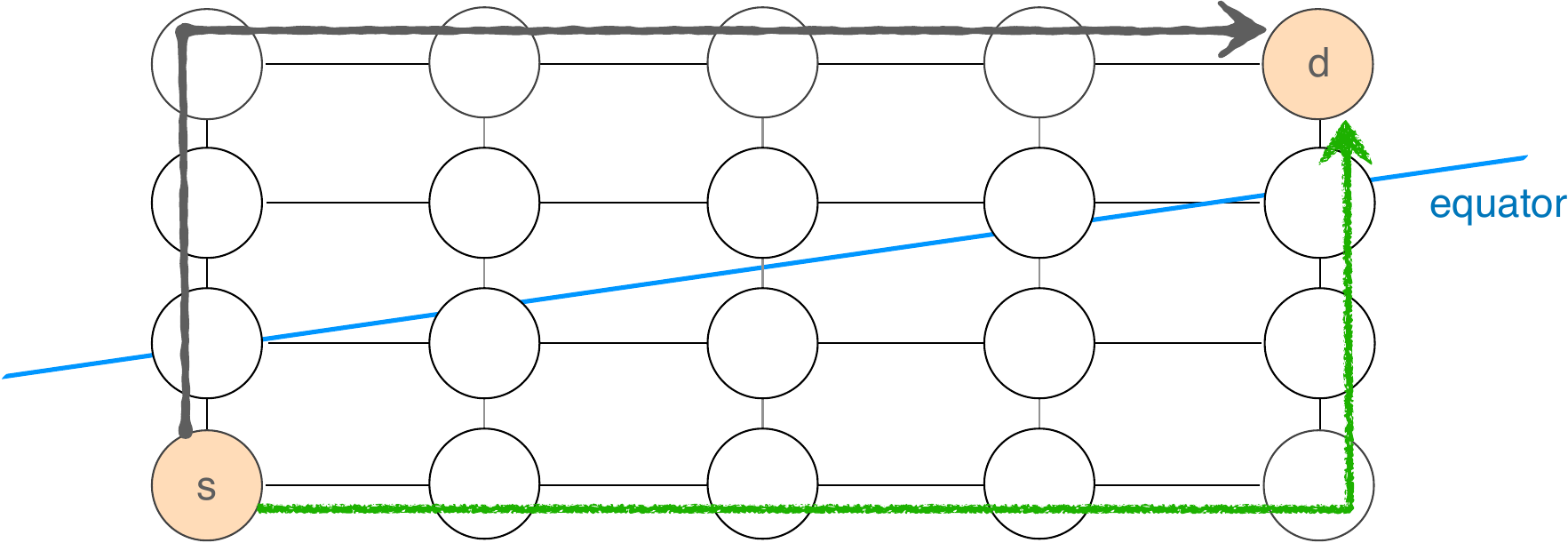}
		\caption{Type-B grids.}
		\label{fig:spacegrid-equator-typeB-sp}
	\end{subfigure}
	\caption{Candidate shortest paths in grids moving in the same direction.}
\end{figure*}

\begin{figure*}[t]
	\centering
	\begin{subfigure}[b]{0.9\columnwidth}%
		\centering
		\includegraphics[width=\columnwidth]{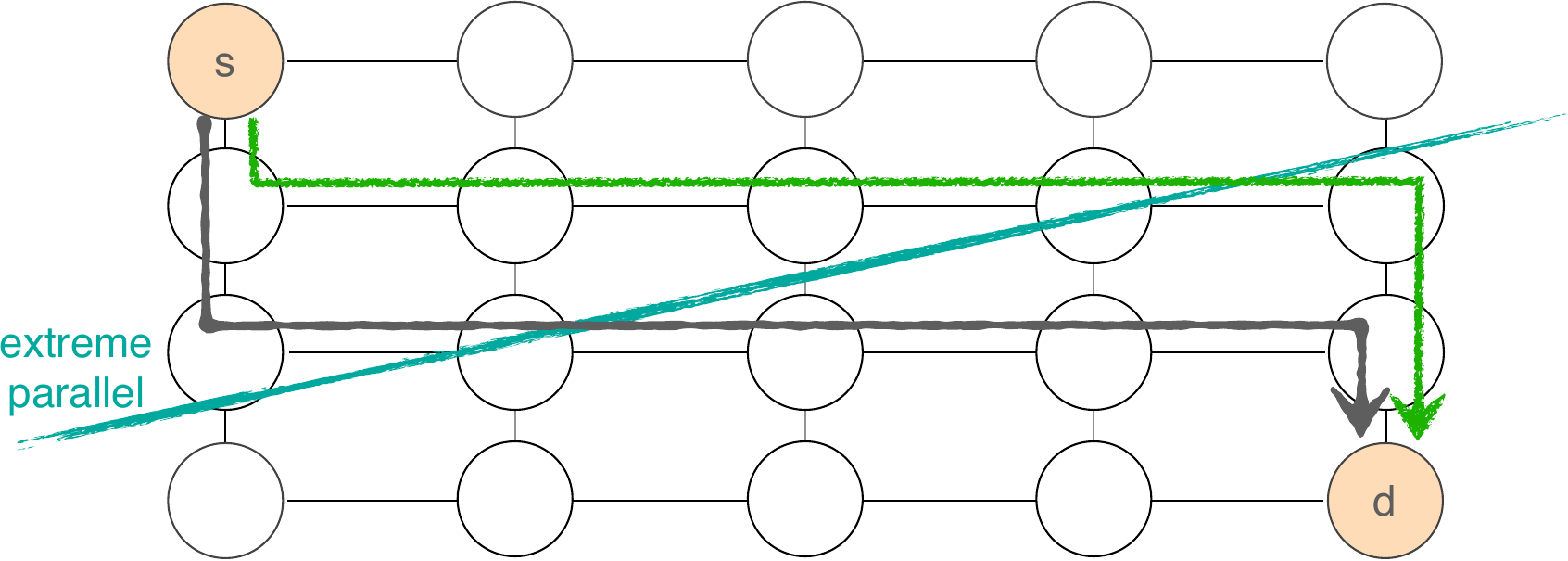}
		\caption{Type-A grids.}
		\label{fig:spacegrid-pole-typeA-sp}
	\end{subfigure}
	\qquad
	\begin{subfigure}[b]{0.9\columnwidth}%
		\centering
		\includegraphics[width=\columnwidth]{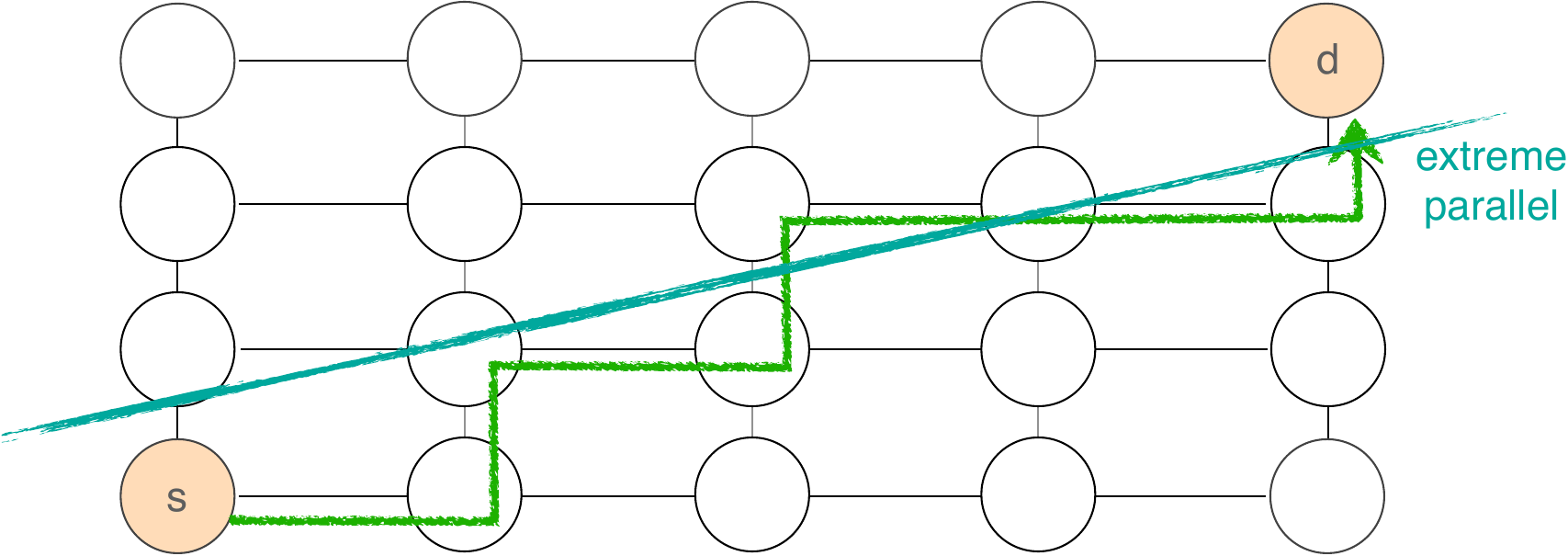}
		\caption{Type-B grids.}
		\label{fig:spacegrid-pole-typeB-sp}
	\end{subfigure}
	\caption{Candidate shortest paths in grids moving in different directions and close to a single extreme parallel.}
	\label{fig:spacegrid-pole}
\end{figure*}

\myitem{Grids' characteristics relevant to routing.}
Within any $s$-$d$ grid, the shape of the shortest path from $s$ to $d$ is determined by the grid's type and position.

We define {\em type-A} and {\em type-B} grids as follows.
Let $c_1$ be the satellite at the other end of the source row from $s$.
As the grid orbits, it will reach a point where some of its links cross the equator and either $s$ or $c_1$ is on the equator.
If $c_1$ is on the equator at that time, the grid is a type-A one, otherwise it is a type-B grid.
Figure~\ref{fig:spacegrid-base} shows a type-A grid: the equator traverses several links in the grid when $c_1$ is on it; conversely, the equator crosses no links in the grid when $s$ is on it.
If $c_1$ were the source and $c_2$ the destination (or vice versa), the grid would have been a type-B one.

The type of a grid does not change as it orbits, but the grid {\em position} does.
We say that the grid \textit{moves in the same direction} if all satellites move in one direction, and that it \textit{moves in different directions} otherwise.  
For example, Starlink satellites move north-east, reach their northmost parallel and then head south-east.
If some satellites in a grid have passed the northmost parallel and some are still approaching it, then the grid moves in different directions.

\myitem{Shortest paths in grids moving in the same direction.}
One among the following Theorems~\ref{theo:sp-descending-crossing} and~\ref{theo:sp-ascending-crossing} applies.

\Copy{theo:samedir-descending-sp}{
\begin{mytheorem}\label{theo:sp-descending-crossing}
	For any type-A grid moving in the same direction, the shortest path from $s$ to $d$ comprises the cross-orbit link farthest from the equator in each column plus the minimal set of intra-orbit links to connect the cross-orbit links.
\end{mytheorem}
}

Figure~\ref{fig:spacegrid-equator-typeA-sp} shows the shortest path P in a type-A grid moving the same direction.
As the shortest link in each column is in either the source or destination row, P follows the source row for some distance, then traverses intra-orbit links to the destination row, and ends with destination-row links.

\Copy{theo:samedir-ascending-sp}{
	\begin{mytheorem}\label{theo:sp-ascending-crossing}
	For any type-B grid moving in the same direction, the shortest path from $s$ to $d$:
	\begin{inparaenum}[(i)]
		\item includes exactly one sequence of cross-orbit links,
		\item cannot have the first cross-orbit link closer to the equator than the cross-orbit link adjacent to $s$, and
		\item cannot have the last cross-orbit link closer to the equator than the cross-orbit link adjacent to $d$.
	\end{inparaenum}
\end{mytheorem}
}

Figure~\ref{fig:spacegrid-equator-typeB-sp} shows the candidate shortest paths resulting from Theorem~\ref{theo:sp-ascending-crossing}.
The shortest cross-orbit links within the source column are in the destination row, and those within the destination column are in the source row.
Thus, it is shorter to cross a single row than to mix links from different rows.

\myitem{Shortest paths in grids moving in different directions.}
We denote the northmost and southmost parallels ever reached by satellites in a constellation as \textit{extreme parallels}.

We say that a grid moving in different directions is \textit{close to a single extreme parallel} if it crosses only one extreme parallel and both its source and destination rows are closer to the equator than to the other extreme parallel.
For any such grid, one of the following two theorems holds.

\Copy{theo:diffdir-descending}{
\begin{mytheorem}\label{theo:shortest-path-descending-grid-wrapping}
	For any type-A grid moving in different directions and close to a single extreme parallel, the shortest path from $s$ to $d$:
	\begin{inparaenum}[(i)]
		\item includes exactly one sequence of cross-orbit links, and
		\item traverses the row of the cross-orbit link closest to the extreme parallel in the source column, the row of the cross-orbit link closest to the extreme parallel in the destination column, or a row located between those two.
	\end{inparaenum}
\end{mytheorem}
}

\Copy{theo:diffdir-ascending}{
\begin{mytheorem}\label{theo:shortest-path-ascending-grid-wrapping}
	For any type-B grid moving in different directions and close to a single extreme parallel, the shortest path from $s$ to $d$ at any time $t$ includes all and only the cross-orbit links that are the closest to the extreme parallel in their respective columns plus a minimal number of intra-orbit links connecting these cross-orbit links.
\end{mytheorem}
}

Figure~\ref{fig:spacegrid-pole} shows the candidate shortest paths identified by Theorems~\ref{theo:shortest-path-descending-grid-wrapping} and~\ref{theo:shortest-path-ascending-grid-wrapping}.
The links closest to the extreme parallel are the shortest in their columns; geometrical properties depending on the grids' type then determine whether it is better to include all per-column shortest links or to cross entire rows.

Not all the grids moving in different directions are close to a single extreme parallel. 
We denote the remaining grids moving in different directions as \textit{close to both extreme parallels}.
The following theorem applies.

\Copy{theo:diffdir-both-poles}{
\begin{mytheorem}\label{theo:diffdir-both-poles}
	Consider a grid close to both extreme parallels and any partitioning of it into two sub-grids each close to a single extreme parallel.
	The shortest path from $s$ to $d$ includes only cross-orbit links traversed by the candidate shortest paths in the two sub-grids.
\end{mytheorem}
}

Theorem~\ref{theo:diffdir-both-poles} constrains shortest paths less than previous theorems. 
This is not a big limitations in practice because the shortest path almost always lies in one of the other three grids, as we confirmed in our simulations (\S\ref{sec:eval}).

\newpage

%% file: internals.tex
\section{\approach Internals}
\label{sec:internals}

\approach path encoding must minimize satellites' work to forward, fast reroute and validate packets, in huge networks.
These requirements rule out existing path encodings, such as per-path~\cite{mpls-label-rfc3032}, IGP-based~\cite{SegmentRoutingArchitecture-rfc8402}, and per-router-interface~\cite{wilfong-pathswitching-conext15} ones.
Based on our theory, we co-design a new path encoding scheme (\S\ref{ssec:internals-forwarding}), a forwarding procedure (\S\ref{ssec:internals-forwarding}), and heuristics for fast rerouting (\S\ref{ssec:internals-frr}) and path validation (\S\ref{ssec:internals-validation}).

\subsection{Path Encoding and Packet Forwarding}
\label{ssec:internals-forwarding}

Source GSes enrich the header of every source-routed packet with the ID of the destination GS, and a list of \textit{path tags}.
Each tag represents a sequence of links that are either all intra-orbit or all cross-orbit, specifying the number of steps and direction of travel.
\system satellites forward packets according to the carried tags, using the procedure shown in Figure~\ref{alg:sat-logic}.

\input{pseudocode-sat}

Figure~\ref{fig:pathtags-encoding} displays an example. The
intended path consists of three sequences of same-type links,
and it is thus encoded by three tags.  The top tag E8 (rightmost
in the packet's representation) indicates that the packet has first to
be forwarded east for 8 steps. The top part of the figure
illustrates the implementation of this tag.  After traversing
8 cross-orbit links, the packet is forwarded south\footnote{We use
``south'' for ease of explanation.
In practice, however, links along an orbital plane travel north
for half the orbit and south for the other.  Tags therefore encode
if a sequence of intra-orbit links is prograde or retrograde.} for 7 steps (tag S7),
and so on.

Our tag-based path representation is typically concise.
Theorems~\ref{theo:sp-descending-crossing}-\ref{theo:shortest-path-ascending-grid-wrapping} indeed indicate that shortest paths often concatenate a few sequences of links of the same type.
Experimental results discussed in \S\ref{sec:eval} provide a confirmation.

Tag-based forwarding at satellites is straightforward in the absence
of failures. An index in the packet indicates
which tag is the current one.  A forwarding satellite decrements the current
tag, and checks if it contains zero steps: if so, it consumes the current
tag by incrementing the packet's tag index (line 2
in Figure~\ref{alg:sat-logic}).
Keeping consumed tags in the packet simplifies fast rerouting, as
explained below.
The satellite then forwards the packet:
if all tags are consumed, it delivers the packet to the destination GS (lines 3-5), otherwise it forwards the packet as specified in the current tag (lines 6-7).

Let's consider again Figure~\ref{fig:pathtags-encoding}.
Every satellite in the path from $s$ and $x$ simply checks the top tag, decrements the number of steps in it, and forwards the packet east. When $x$ receives the packet with top tag E1, it decrements the number of steps, and consumes the tag by incrementing the tag index. Then, $x$ sends the packet south, following the new current tag S7. The packet is forwarded south for 7 steps, until $y$ is reached. The last tag E9 ensures that the packet is delivered to $d$.

\subsection{Fast rerouting}
\label{ssec:internals-frr}

Fast rerouting is triggered at satellites whenever the direction in the current tag cannot be followed.
There are two cases to consider, as shown in Figure~\ref{fig:frr}.

\begin{figure}[t]
	\centering
	\includegraphics[width=.9\columnwidth]{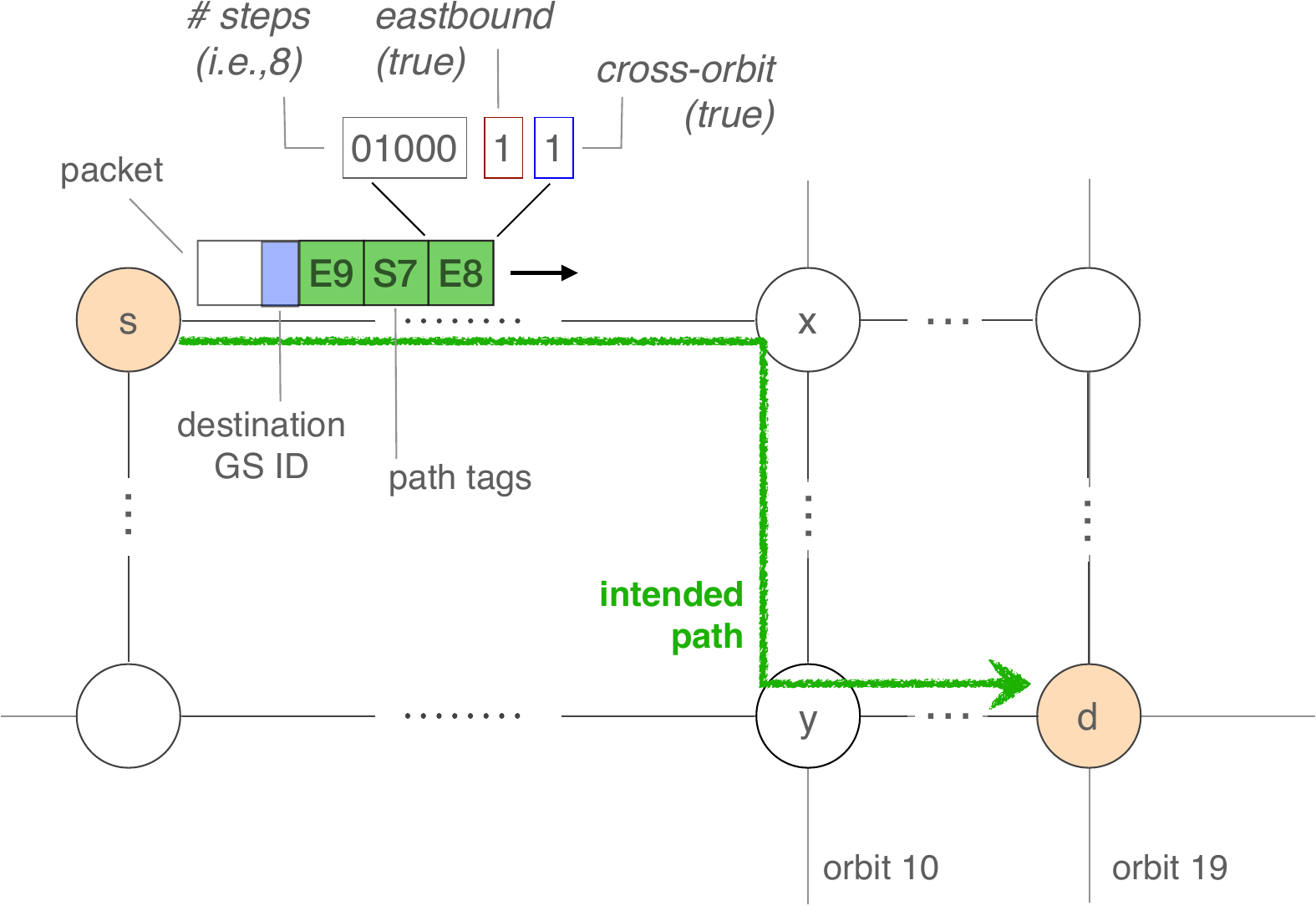}
	\caption{Example of \system path encoding.}
	\label{fig:pathtags-encoding}
\end{figure}

\begin{figure*}[t]
	\centering
	\begin{subfigure}[b]{0.9\columnwidth}%
		\includegraphics[width=.9\columnwidth]{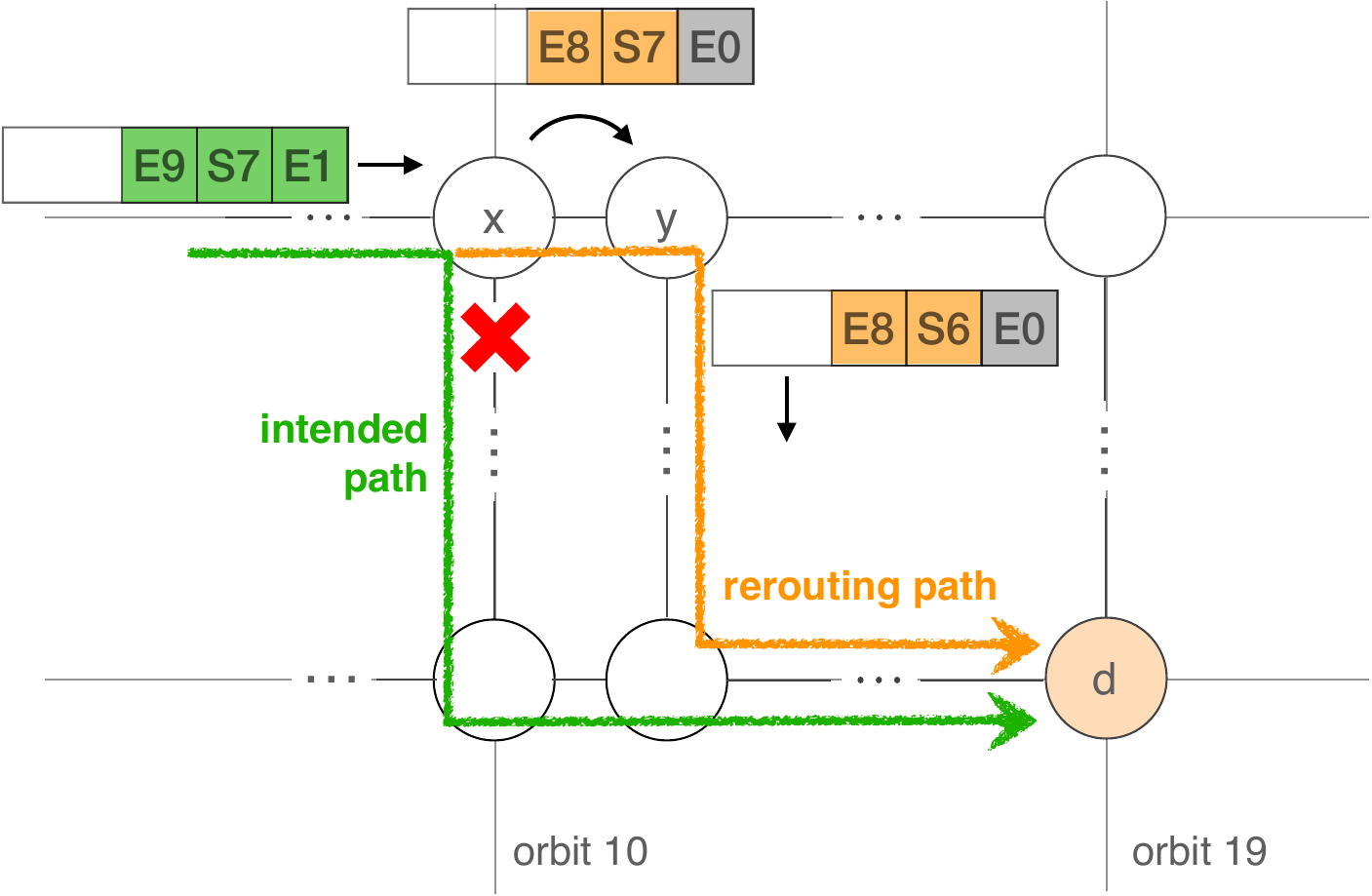}
		\caption{Case 1.}
		\label{fig:frr-base}
	\end{subfigure}
	\qquad
	\begin{subfigure}[b]{0.9\columnwidth}%
		\centering
		\includegraphics[width=.8\columnwidth]{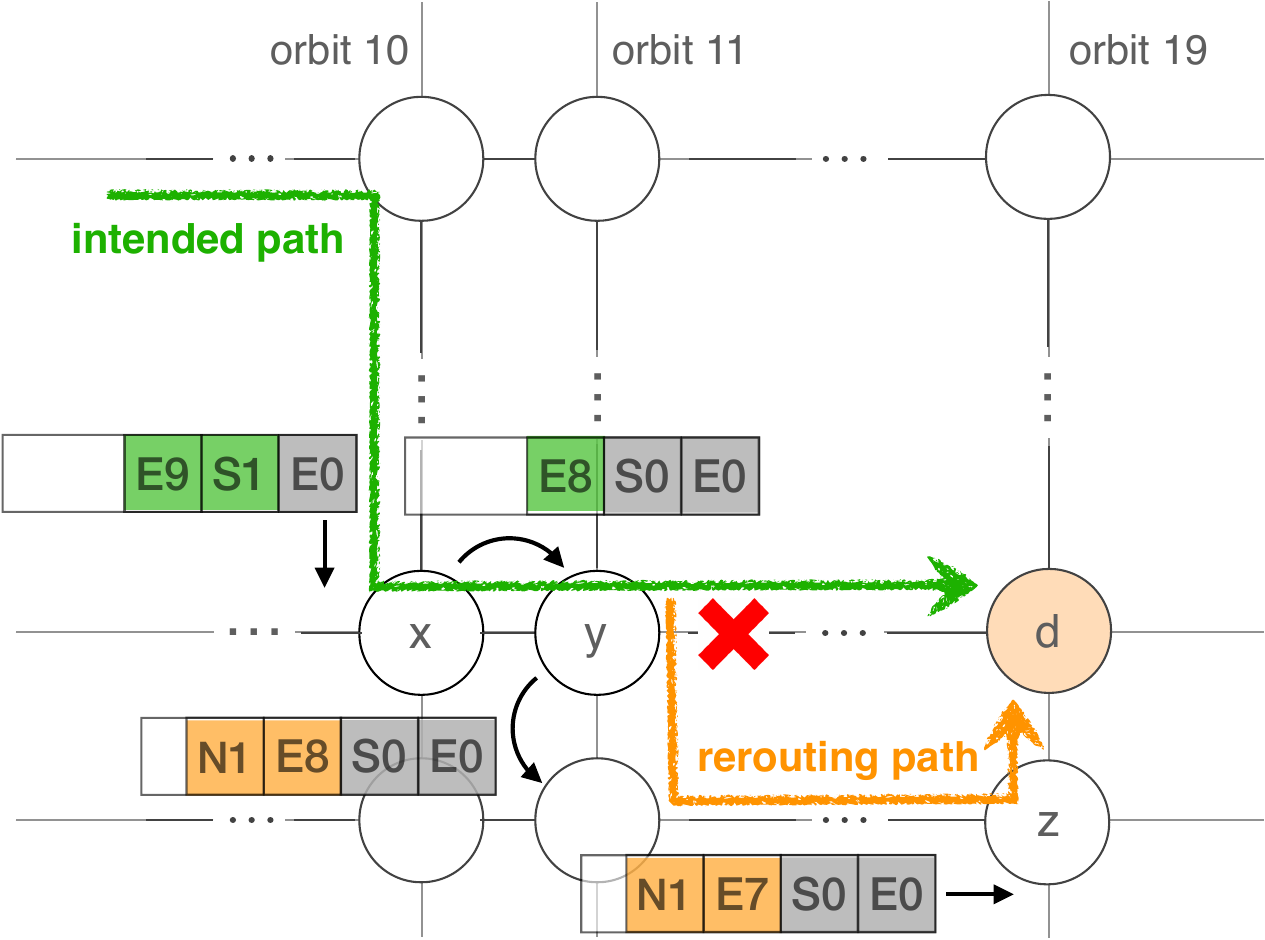}
		\caption{Case 2.}
		\label{fig:frr-special}
	\end{subfigure}
	\caption{Illustration of \system fast rerouting scheme.}
	\label{fig:frr}
\end{figure*}

The first case is the focus of lines 9-11 in Figure~\ref{alg:sat-logic}
and exemplified in Figure~\ref{fig:frr-base}:
packets include at least two unconsumed tags when they reach
the failed link (i.e., S7 and E9 in Figure~\ref{fig:frr-base}).
Any such packet is forwarded according to the second-top unconsumed tag
(i.e., E9) for one hop; that tag is also decremented (i.e., becoming
E8 after satellite $x$).
The rerouted packet then follows a sequence of links parallel to the one encoded in the current tag. 

In the second case, visualized in Figure~\ref{fig:frr-special}, packets to be rerouted include only one unconsumed tag when they reach the failure.
These packets are still rerouted over a sequence of links parallel to the current tag.
Since we do not have other unconsumed tags, however, the rerouting satellite chooses the next hop according to the direction of the last consumed tag (line 13 in Figure~\ref{alg:sat-logic}); it then appends a new tag pointing in the opposite direction, to guarantee that the packets eventually reach their destination (line 14 in Figure~\ref{alg:sat-logic}).
Packets without any consumed tag (i.e., all links in the original path are intra-orbit or all are cross-links) are rerouted on an arbitrary direction orthogonal to the one of the current tag.

In both cases, if the rerouting link is also failed, \system satellites simply drop packets.
This is intended to limit delay inflation for rerouted packets while also minimizing computations at rerouting satellites.
Packets are also dropped if the RF link between the destination satellite and the destination GS is failed.
Following our routing-free core philosophy, source-destination GSes are indeed responsible to agree on the RF links to use at any time.

\myitem{Single-link and single-node failures.}
\system fast-rerouting scheme is mainly designed for single-ISL and single-satellite failures.
We indeed expect that multiple link failures requiring fast rerouting be infrequent, as GSes quickly recompute paths after being notified of each failure.

For single-link or single-node failures, \approach provides the following guarantees, as proved in Appendix~\ref{app:frr}.

\Copy{theo:frr}{
	\begin{mytheorem}\label{theo:frr}
		If any $s$-$d$ grid includes one failed link or node $f \neq d$, our fast rerouting scheme guarantees that every packet from $s$ reaches $d$ with a maximum hop stretch of two links, and a delay increase equal to either the length of two links or the difference between two adjacent rows in the grid.
	\end{mytheorem}
}

\myitem{Multi-link failures.}
Although expected to be rare, we still want to provide basic guarantees for multiple link or node failures, especially regarding forwarding loops.
Loops are indeed worse than dropping packets because they also waste network resources (i.e., link bandwidth, satellites' power, etc.).

As is, the above algorithm may cause loops if multiple satellites reroute packets on paths with failed links.
For example, if the last link of the rerouting path in Figure~\ref{fig:frr-special} is failed, $z$ will reroute the packet on a new sub-path (E1,N1,W1).
If the last link of this latter sub-path is also failed, the packet will be rerouted again, taking one more step north before being sent west.
With enough carefully placed failures, the packet can eventually reach $y$ again, and hence be trapped in a loop.

To avoid loops, we allow at most two reroutings per packet\footnote{Two is the maximum number of reroutings always guaranteeing the absence of loops in our topologies.}.
Each packet has a 2-bit loop flag, set to zero by the source GS.
When rerouting, any satellite checks if the loop flag has a value lower than two: if so, it increments the flag and reroutes the packet, otherwise it drops the packet\footnote{Source GSes are \emph{not} notified by possible packet drops, exactly as it happens when routers drop packets in the current Internet}.

We note that \system often delivers packets to their destinations, over low-delay paths, even in the presence of several on-path failures, as we experimentally show in \S\ref{sec:eval}.

\subsection{Packet Validation}
\label{ssec:internals-validation}

\system satellites validate premium packets received over RF links by running a few checks on them: they drop packets failing any check.
GSes are informed of these checks, so they can always generate traffic passing the validation.

Building upon our theory, \system checks are designed to be simple, cheap and quick to run, minimizing both resource consumption and processing delay.
In fact, they are amenable to be implemented in hardware.
Our checks are also not overly constraining, therefore allowing legitimate discrepancies from theoretical shortest paths (e.g., upon failures) and customized routing algorithms at source GSes.

\input{table-pathval}

\myitem{Validation checks.}
Consistently with the service definition in \S\ref{sec:problem}, a preliminary check assesses if the source GS is allowed to send premium traffic to the destination GS specified in the packet, and if the traffic rate is not excessive.

The other checks, summarized in Table~\ref{tab:pathval-checks}, evaluate properties of encoded paths: if source-routed paths end at a reasonable satellite (check 1), zig-zag or loop (check 2), and cross a reasonable number of intra- and inter-orbit links (check 3).

Check 1 determines if the specified destination GS is actually reachable from the destination satellite.
To do so, it translates packet tags into the topological position (relative to the satellite running the check) of the last satellite in the path.
The geographical position of such destination satellite is then derived from the pre-loaded satellite orbits.

Check 2 computes how many times and for how many links the input path inverts direction -- for example, starting east and going west at a later point. 
Our theorems prove that shortest paths in a working constellation never invert direction.
Paths around failures are expected to invert directions rarely.
Since they require similar information, \system runs check 1 and 2 with a single pass on the tag list. 

Check 3 builds upon the intuition that shortest paths are contained in a single grid (see Theorem~\ref{theo:ingrid-vs-outgrid}), often the smallest in size.
This tends to remain true even if a few ISLs fail.
Hence, check 3 compares the number of cross-orbit links in the input path with the number of columns of the minimal grid, and the number intra-orbit links in the path with the number of rows in that grid.

%% file: pseudocode-sat.tex
\begin{figure}
	\small
	\hrule
	\medskip
	\begin{algorithmic}[1]
		\STATE \textbf{process\_tags} (tag\_list, curr\_index):
		\STATE curr\_index $\leftarrow$ update\_tags(tag\_list)
		\IF{curr\_index == len(tag\_list)}
		\RETURN None
		\ENDIF
		\STATE intended\_fwd\_dir $\leftarrow$ get\_dir(tag\_list[curr\_index])
		\STATE actual\_fwd\_dir $\leftarrow$ intended\_fwd\_dir
		\IF {is\_fast\_rerouting\_needed(intended\_fwd\_dir)}
		\IF {curr\_index < len(tag\_list) - 1}
		\STATE actual\_fwd\_dir $\leftarrow$ get\_dir(tag\_list[curr\_index + 1])
		\STATE decrement\_steps(tag\_list[curr\_index + 1])
		\ELSE
		\STATE actual\_fwd\_dir $\leftarrow$ get\_prev\_dir(tag\_list)
		\STATE tag\_list.append\_tag(get\_opposite\_dir(actual\_fwd\_dir))
		\ENDIF
		\ENDIF
		\RETURN (actual\_fwd\_dir, tag\_list,curr\_index)
	\end{algorithmic}
	\medskip
	\hrule
	\medskip
	\caption{Pseudo-code of satellites' forwarding logic.}
	\label{alg:sat-logic}
\end{figure}

%% file: table-pathval.tex
\begin{table}[]
	\centering
	\small
	\def\arraystretch{1.2}
	\begin{tabular}{@{}ll@{}} 
		\toprule
		test & description \\
		\midrule
		check1 & $P$'s last satellite is reachable by the destination GS. \\
		check2 & $P$ inverts direction no more than once; the inversion \\
           & (if any) has $\leq$ 3 links, unless $G_{min}$ is a single-path grid. \\
		check3 & $P$ has $\leq$ 2 links among intra-orbit links exceeding $r_{min}$ \\
           & and cross-orbit links exceeding $c_{min}$.\\
		\bottomrule
		\multicolumn{2}{l}{$P$ is the input path; $G_{min}$ is the minimal-size grid; $r_{min}$ and $c_{min}$} \\
		\multicolumn{2}{l}{respectively are the number of rows and columns in $G_{min}$.} \\
	\end{tabular}
	\caption{\system packet validation checks.}
	\label{tab:pathval-checks}
\end{table}

%% file: eval.tex
\begin{figure*}[t]
	\centering
	\begin{subfigure}[b]{0.9\columnwidth}%
		\includegraphics[width=\columnwidth]{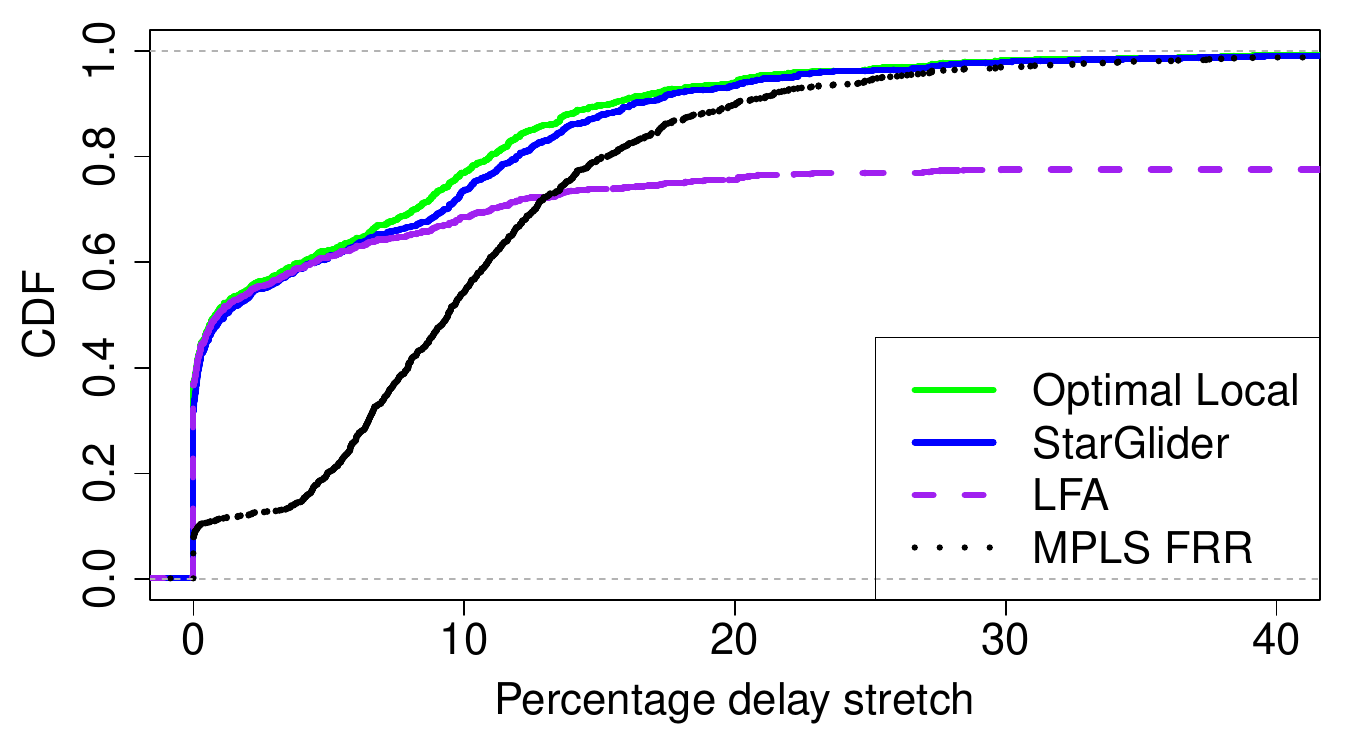}
		\caption{Single-link failures}
		\label{fig:eval-frr-1link-10frames-stretch}
	\end{subfigure}
	\qquad
	\begin{subfigure}[b]{0.9\columnwidth}%
		\centering
		\includegraphics[width=\columnwidth]{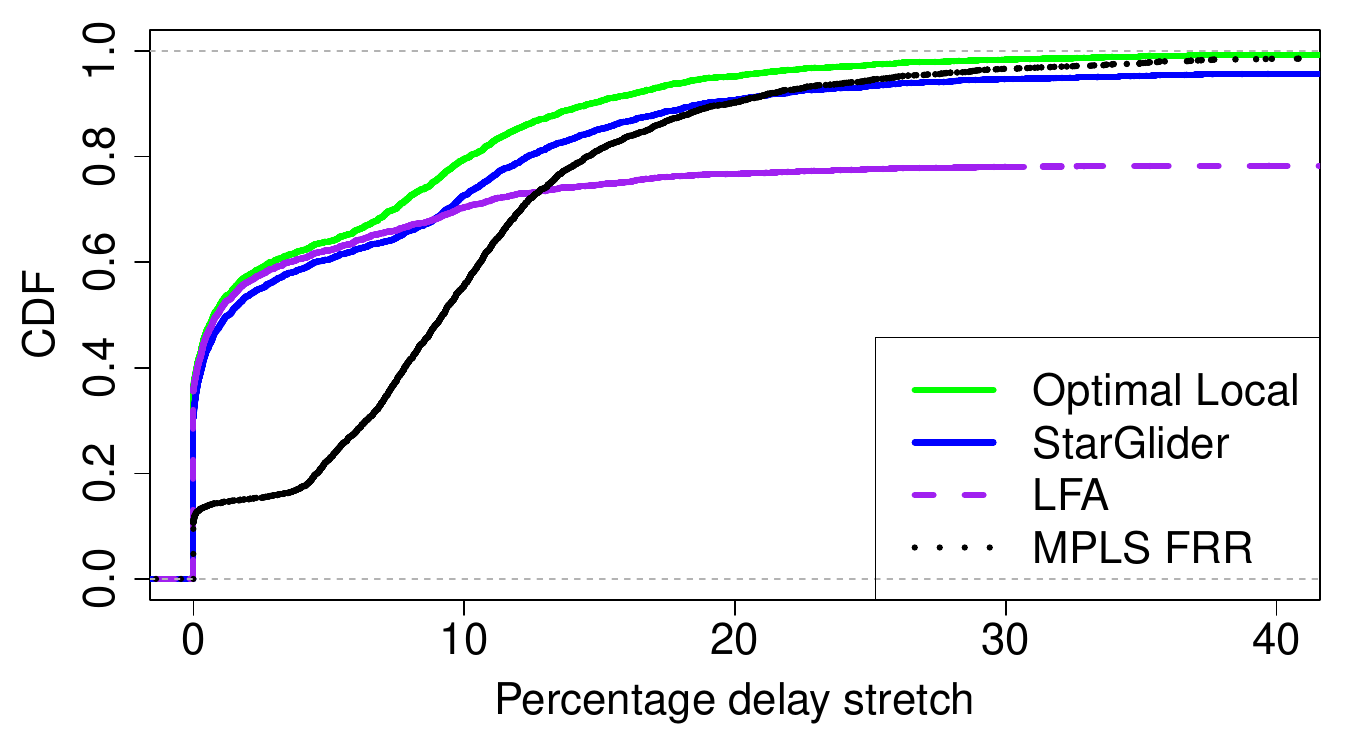}
		\caption{Three-link failures}
		\label{fig:eval-frr-3link-10frames-stretch}
	\end{subfigure}
	\caption{Fast rerouting delay stretch. In addition to requiring no computation on satellites, \system fast-reroutes packets on paths shorter than state-of-the-art alternatives. We remark that in each of our experiments, \system fast-rerouting paths always include $<$ 3.6\% ISLs, and results would remain the same if \emph{all} the other $\geq 96.4\%$ links were failed. This suggests that \system is robust to \emph{many} more failure scenarios than the simulated ones.}
	\label{fig:eval-frr}
\end{figure*}

\section{Evaluation}
\label{sec:eval}

We use Unity~\cite{unity} to simulate the fully deployed Starlink constellation as
described in~\cite{mark-hotnets18}.
We run all our experiments on a commodity laptop (2.9 GHz CPU, 16 GB RAM).

In our simulations, source GSes compute shortest paths using the Dijkstra's algorithm,
which takes a few milliseconds.
In the following, we evaluate \system fast rerouting and validation schemes built
on top of these path computations.

\vspace{-0.15cm}
\subsection{Forwarding and Fast Rerouting}

We compare \system against forwarding and FRR solutions commonly adopted in the Internet.
State-of-the-art FRR schemes can be classified into two families.
Approaches in the first family~\cite{lfa-rfc5286,uturn-draft06,chiesa-purr19} enable devices next to a failure to redirect \textit{unmodified} packets to unaffected neighbors.
Approaches in the second family~\cite{mplsfrr-rfc4090,rlfa-rfc7490,d2r-sosr21}
\textit{tunnel} traffic to remote nodes.
These two families achieve inherently different tradeoffs in terms of effectiveness, overhead and scalability, whereas techniques in each family mainly differ in how they select backup paths.
We thus experiment with the most popular representative for each family: LFA~\cite{lfa-rfc5286} and MPLS FRR~\cite{mplsfrr-rfc4090}, respectively.

Consistently with our low-delay routing goal, we implement LFA and MPLS FRR so that satellites select minimal-delay paths from among the backup paths allowed by each scheme.
Since ISL lengths change continuously, we adapt LFA and MPLS FRR to recompute backup paths periodically.

We run 1,000 experiments with randomly selected source and destination satellites.
In each experiment, we first compute LFA and MPLS FRR backup paths for the current shortest path between the selected satellites.
Within the first second of simulation, we then fail one link in the shortest path.
Hence, LFA and MPLS FRR reroute on paths less than a second old, which is a favourable case for these schemes:
real implementations would likely compute backup paths less often than once per second.
\system entails no pre-computation.

\newpage

\noindent
\textbf{Resource consumption.}
We separately evaluate resources consumed by \approach path encoding and forwarding.

\myitemit{Packet overhead.}
\system packet overhead depends on the tags needed to encode paths and possibly to fast reroute.
In our experiments, the median (resp., max) number of tags is 2 (resp., 9) before failures, and 3 (resp., 10) after fast rerouting.
Each tag is 9 bit long: 2 bits encode the tag direction, and 7 bits the number of steps per direction (bounded by the number of satellites per orbit which is 66).
Hence, in the absence of failures, the packet overhead is 18 bits in median, and 81 bits at most.
Triggering rerouting brings the overhead to 27 bits in median, and 90 bits in the worst case.

As a comparison, LFA has no packet overhead, but MPLS FRR requires one 32-bit-long MPLS label per packet~\cite{mpls-label-rfc3032}, an overhead comparable to \system.

\myitemit{Satellite resources.}
By design, \system requires virtually no resources other than packet forwarding hardware:
we envision that \system satellites forward and fast-reroute traffic by processing the packet tags in the data plane.

In contrast, for each hop in every path to protect, LFA requires one shortest path computation per neighbor, and MPLS FRR one path computation.
For a satellite, this equates to $\approx$ 5 ms for LFA and 1-2 ms for MPLS FRR per path to protect in our experiments.
These computation times are not negligible, especially when many backup paths have to be re-computed.
Even worse, such computations consume the scarcest satellite resource: power.
LFA and MPLS FRR also need resources to run an IGP; MPLS FRR also needs to maintain tunnels for backup paths.
Finally, both LFA and MPLS FRR pre-install one backup forwarding entry per path to protect on every satellite, except  
when LFA cannot find a safe neighbor~\cite{lfa-applicability-rfc6571}.

\myitem{Effectiveness.}
Figure~\ref{fig:eval-frr-1link-10frames-stretch} shows a CDF of the delay stretch of backup paths enforced by \system, LFA and MPLS FRR.
The reported delay stretch is the percentage increase of rerouted path latency over that of the optimal post-failure path.
When a scheme is unable to fast reroute traffic, we assign an infinite value to the delay stretch.
The plot also includes the delay of the best possible backup paths from the satellite next to the failure: such delay is denoted as Optimal Local.

\system always computes a working backup path with limited delay stretch, close to the Optimal Local one.
It reroutes on the overall optimal path for $\approx$ 30\% of the simulated failures, and on a path $\leq$ 5\% longer than the optimal one in $\approx$ 60\% of the experiments.
Its delay stretch is >20\% for only 7\% of the failures -- and so is the Optimal Local one.

In contrast, LFA delay stretch is as good as \system for only $\approx$ 60\% of the failures.
In the remaining cases, LFA either reroutes on paths with higher delay than \system, or, for more than 20\% of failures, cannot reroute at all.
MPLS FRR redirects packets on paths with delays that are always higher than \system, often significantly so.

\myitem{Additional experiments.}
We also run experiments with two and three link failures, either consecutive or simultaneous.
In these experiments, we always fail links in the shortest path before each failure.
Figure~\ref{fig:eval-frr-3link-10frames-stretch} shows the result of 1,000 simulations with three link failures.
Although designed for single-link failures, \system still provides low-delay backup paths, typically outperforming LFA and MPLS FRR.
It drops packets only after less than 4\% of the failures.

\subsection{Packet Validation}
\label{ssec:eval-validation}

We now evaluate \system validation against a Dijkstra-based approach that compares the length of the path to be validated against the shortest one.

In each experiment, we randomly select a pair of traffic source and destination.
Shortly after starting a simulation, we generate packets to be validated.
Some of these packets are legit, and are source-routed over the shortest paths (i) when all ISLs work, (ii) when one ISL fails in the current shortest path, and (iii) when one ISL fails in the shortest path and another ISL fails in the optimal backup path.
Other packets are malicious.
Consistently with the attacker model in \S\ref{sec:problem}, malicious sources try to steer traffic through target links or satellites outside the source-destination shortest paths.

For malicious packets, we implement an \emph{attacker strategy tailored against \system validation}.
Malicious sources generate packets for destination satellites reachable by allowed destination GSes, so that they always pass \system pre-check and check 1. 
To increase chances to pass \system checks 2 and 3, malicious paths concatenate the shortest path from the source satellite to the attack target with the shortest path from the target to the destination satellite.
Our theory (\S\ref{sec:theory}) indeed shows that shortest paths minimize the metrics tracked by checks 2 and 3 -- i.e., number and size of inversions of directions and traversed links.

\input{table-pathval-comptime}

\myitem{Resource consumption.}
Table~\ref{tab:pathval-time} shows the efficiency of \system checks when performed one after the other, across 1,000 experiments.
Checks 1 and 2 have a common tag process time because a single pass on the tags is sufficient for both checks.
The extract metrics time for check 2 is zero because processing tags directly provide check 2 metrics.

All \system checks take only a few microseconds, with no statistical difference between legit and malicious paths.
This is because by design (see \S\ref{sec:internals}), they only process information contained in packets' tags, and access static, pre-stored data.
In comparison, the Dijkstra-based approach is more than 1000x slower.
Computing the shortest path on an updated network graph already takes a few milliseconds.
Calculating the current delay of all links in such a path and summing them together requires a few more milliseconds.

\system checks also have lower memory requirements than the Dijkstra-based approach, as they don't require to track the delays of all ISLs in the network.

\input{table-pathval-accuracy}

\begin{figure*}[t]
	\centering
	\begin{subfigure}[b]{0.99\columnwidth}%
		\includegraphics[width=.95\columnwidth]{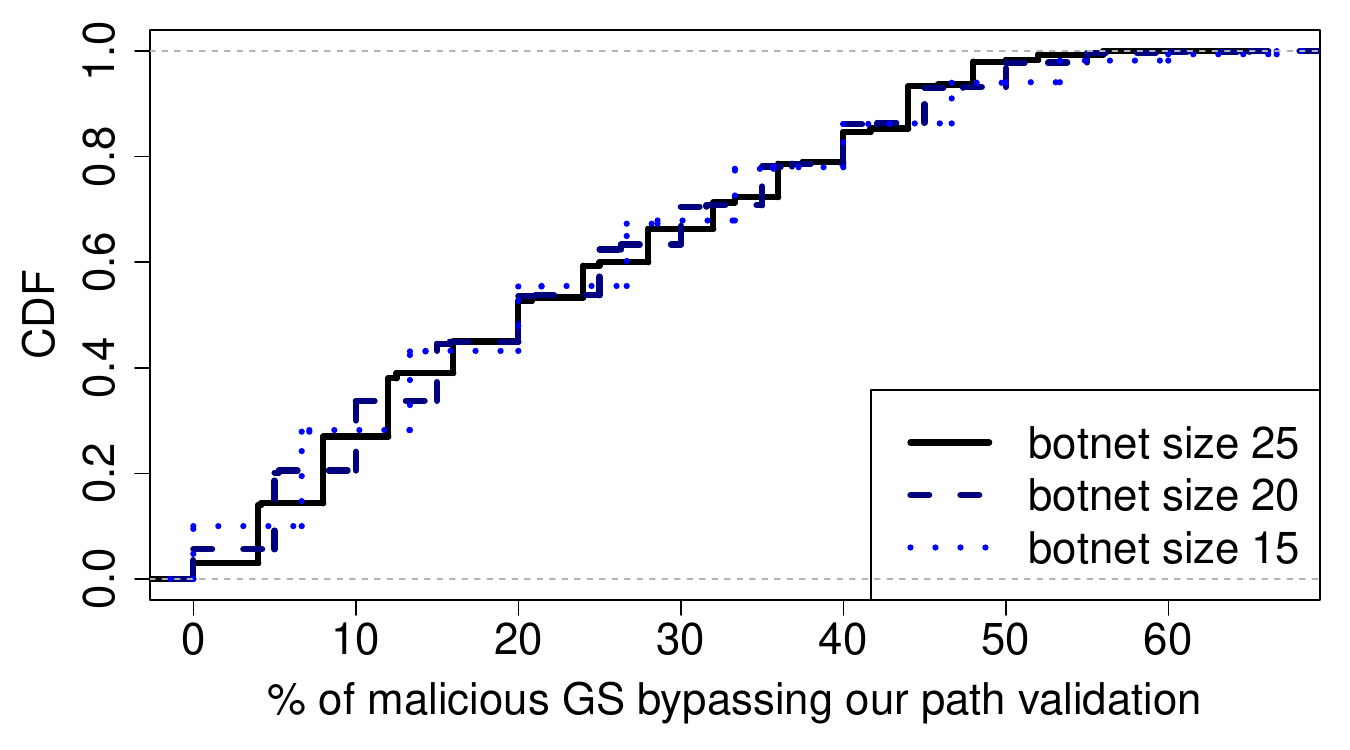}
		\caption{}
		\label{fig:linkflood-successful-gspairs}
	\end{subfigure}
	\hfill
	\begin{subfigure}[b]{0.99\columnwidth}%
		\centering
		\includegraphics[width=.95\columnwidth]{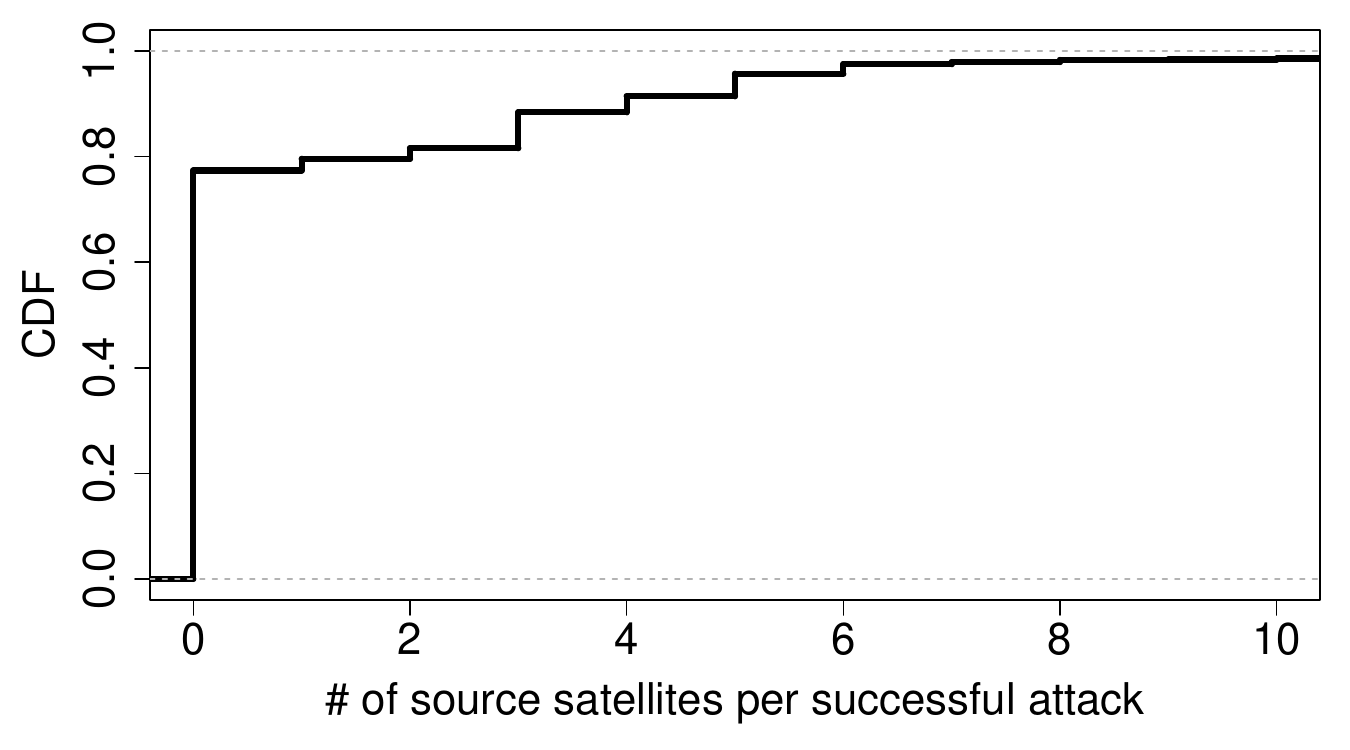}
		\caption{}
		\label{fig:linkflood-critical-sats}
	\end{subfigure}
	\caption{Effectiveness of \system validation for multi-GS attacks (e.g., DDoS).}
	\label{fig:linkflood}
\end{figure*}

\myitem{Effectiveness.}
We simulate attack cases consistent with \S\ref{sec:problem}. 

\myitemit{Per-satellite-pair attacks (e.g.,~\attacklabel{$\mathcal{A}$1}-\attacklabel{$\mathcal{A}$2}).}
Suppose that an attacker can source-route packets on any path from satellite $s$ to satellite $d$, and has a target outside the shortest $s$-$d$ path.

We analyze the success rate of \system validation with respect to the distance of the target to the shortest path from $s$ to $d$.
We do so because malicious paths targeting farther-away links generally have to be longer and more circuitous, so easier to identify. 
We additionally check the impact on legit GSes, by computing the fraction of optimal post-failure backup paths forbidden by path validation.

Table~\ref{tab:pathval-accuracy} summarizes the result of 1,000 experiments with random sources and destinations.
Our main findings follow.
\begin{itemize}
\item \system validation is effective:
the attacker can steer traffic through a target which is one hop away from $s$-$d$ shortest paths in roughly 19\% of our experiments, and through more distant targets in even fewer cases.

\item \system validation rarely forbids optimal backup paths.
It does so in $\approx$ 2\% of the single-failure cases, and $\approx$ 7\% of the two-failure ones.
In such cases, \system anyway allows slightly longer backup paths.
For example, for every single-link failure, simple heuristics enable us to build \system-validated backup paths having delay stretch within 1\% of the optimum in most cases.

\item \system validation works better than relying on path lengths.
Suppose for example that the Dijkstra-based approach admits only paths with delay $\leq$10\% higher than the shortest one.
This approach fails for $\approx$ 40\% of malicious paths with a one-hop away target, and 20\% with two-hop away targets.
Even worse, it forbids optimal backup paths, and hence \emph{any possible backup path}, in $\approx$ 5\% (resp., $\approx$ 10\%) of single-link (resp., two-link) failure cases.
Picking a different delay threshold does not improve performance.
Indeed, allowing higher-delay paths reduces incorrectly classified legit paths, but also allows even more malicious paths -- as exemplified in the last column of Table~\ref{tab:pathval-accuracy}.
\end{itemize}

\myitemit{Multi-GS attacks (e.g.,~\attacklabel{$\mathcal{A}$1}-\attacklabel{$\mathcal{A}$3}).}
Consider now attackers who control botnets of GSes and try to use any reachable satellite.

Figure~\ref{fig:linkflood} displays the results of $\approx$2,500 settings, with randomly selected botnets and targets.
Any botnet of size X consists of X pairs of GSes extracted among 55 of the largest cities in the world.
We make the following observations.
\begin{itemize}
\item \system validation renders botnets largely unusable.
As shown in Figure~\ref{fig:linkflood-successful-gspairs}, packets bypassing \system validation can be originated from only 20\% of the botnets in median, and $\leq$ 40-50\% of them in almost all experiments.
Those numbers do not change with the botnet size.

\item \system validation makes successful attacks fragile and easier to detect.
Figure~\ref{fig:linkflood-critical-sats} shows that even successful attacks rely on a handful of source satellites, typically less than four.
The few satellites critical to successful attacks may however be unavailable at times (e.g., when RF links are disrupted by physical obstacles).
Having to rely on few satellites also makes attacks easier to detect.
For example, if attackers need to change ingress satellites for specific premium demands, both legitimate users and network controllers can notice such unexpected changes.

\end{itemize}

\noindent
Overall, \system validation significantly increases the complexity and cost of multi-GS attacks.
Suppose that an attacker aims at sending 10Gbps through a target ISL.
If compromised GSes are allowed to source-route an average of 100 Mbps per premium destination, the attacker must generate traffic from 10 Gbps / 100 Mbps = 100 premium GS pairs.
According to Figure~\ref{fig:linkflood-successful-gspairs}, this typically requires a botnet of size at least 500.
Likely, however, many more GSes plus the ability to rapidly switch source GSes and ingress satellites are needed
in order to compensate for the few ingress satellites usable by each source GS for the attack (see Figure~\ref{fig:linkflood-critical-sats}).

\myitemit{Other attacks (e.g.,~\attacklabel{$\mathcal{A}$4}).}
Packets are checked independently, so attackers cannot affect validation results of legit packets.

\system satellites are also robust to malicious increases of their validation workload.
Premium traffic exceeding the per-GS allowed rate is immediately discarded (see \S\ref{sec:problem}).
The remaining packets are validated through checks that are cheap and fast: Table~\ref{tab:pathval-time} indicates that our under-optimized software prototype can validate millions of packets per second; hardware implementations would validate packets at line rate.
Hence, attackers need to control many close GSes to stress-test an ingress satellite.
For example, our software prototype can validate minimal-size packets from tens of GSes sending $\approx$ 100 Mbps of premium traffic.

Attackers can deplete or disrupt physical resources -- e.g., congesting or jamming RF links.
These attacks are orthogonal to \system, and outside its scope (see \S\ref{sec:problem}).
Note however that they are easy to detect from both victim GSes and ingress satellites, and can be mitigated with semi-automated procedures -- e.g., changing GS authentication methods.

\myitem{Additional experiments.}
We also experiment with alternative malicious paths, built by following the shortest path up to a random node, and then adding a detour through a target located one, two or three hops away.
Results and trends are generally consistent with the ones discussed in this section.

%% file: table-pathval-comptime.tex
\begin{table}[]
	\centering
	\small
	\def\arraystretch{1.2}
	\begin{tabular}{@{}l|c|c|c|c@{}} 
		\toprule
		& \multicolumn{2}{c|}{process tags} & \multicolumn{2}{c}{extract metrics} \\
		check ID & legit & non-legit & legit & non-legit \\
		\midrule
		Dijkstra-based & 27-144 $\mu$s  & 34-157 $\mu$s  & 2-7 ms & 2-7 ms  \\
 	    check1 & \multirow{2}{*}{4-8 $\mu$s}  & \multirow{2}{*}{4-8 $\mu$s}  & 2-6 $\mu$s & 3-7 $\mu$s  \\
		check2 &  &  & 0  & 0 \\
		check3 & 2-4 $\mu$s & 2-4 $\mu$s & 1-3 $\mu$s  & 1-3 $\mu$s  \\
		\bottomrule
	\end{tabular}
	\caption{Time efficiency of \system validation against the Dijkstra-based baseline. Cells show 5th-95th percentiles.}
	\label{tab:pathval-time}
\end{table}

%% file: table-pathval-accuracy.tex
\begin{table}[]
	\centering
	\small
	\def\arraystretch{1.2}
	\begin{tabular}{@{}l|c|c|c@{}}
		\toprule
		& & \multicolumn{2}{c}{Dijkstra-based} \\
		scenario & \system & 10\% stretch & 20\% stretch \\
		\midrule
		attack targets & & & \\
   		- 1 hop away & 80\% & 59.5\% & 25.7\% \\
		- 2 hops away & 83.5\% & 77.9\% & 53.1\% \\
		optimal backup & & & \\
		- 1 fail & 98.4\% & 95.7\% & 97.6\% \\
		- 2 fails & 91.3\% & 89.4\% & 95.1\% \\
		\bottomrule
	\end{tabular}
	\caption{Success rate (i.e., correct validation) of \system against the Dijkstra-based baseline. Results are very similar for Dijkstra-based with absolute delay thresholds; trends are consistent for experiments with 3-hop away targets and 3 fails.}
	\label{tab:pathval-accuracy}
\end{table}

%% file: discussion.tex
\section{Limitations and Extensions}
\label{sec:discussion}

We now discuss \system limitations, and sketch possible extensions whose full investigation is left for future work.

\myitem{Violating assumption of our theory.}
\system's theory based on two assumptions (see \S\ref{sec:theory}) which stem from the network model detailed in Appendix~\ref{app:model}.

We empirically confirm that these assumptions hold even if constellations do not fully comply with our model.
For example, intra-orbit ISLs do not have \emph{exactly} the same length throughout our simulations, because of geometrical properties of actual satellites' orbits.
This suggests that our theory is robust to minor deviation from our model, possibly including those due to collision avoidance maneuvers, automorphic decays, and orbit reconfigurations reported in~\cite{leo-operation-mobicom23}.

What happens however if the assumptions of our theory do not hold?
The paths characterized by our theory would still be low-delay ones, although they may be not optimal.
This implies that:
\begin{inparaenum}[(i)]
  \item \system fast rerouting scheme will keep preserving connectivity, as it relies on our theory only for its performance guarantees (see \S\ref{ssec:internals-frr}); and
  \item \system validation checks will still be useful, given that they do not only allow the paths characterized by our theory (see \S\ref{ssec:eval-validation}).
\end{inparaenum}

\myitem{Supporting other topologies.}
\system works on regular mesh topologies which are often assumed by previous work~\cite{eth-hotnets18,mark-hotnets18}.
Other topologies are however possible~\cite{ankit-conext19}.

We believe that this work can be a starting point to handle a variety of topologies.
First, our theory can be extended for specific variants of our assumed topologies, or used as-is to design heuristics.
Figure~\ref{fig:extension-example} shows an example.
Second, \system effectiveness may motivate others to explore theory-based approaches tailored to different topologies.

\begin{figure}[t]
	\centering
	\includegraphics[width=\columnwidth]{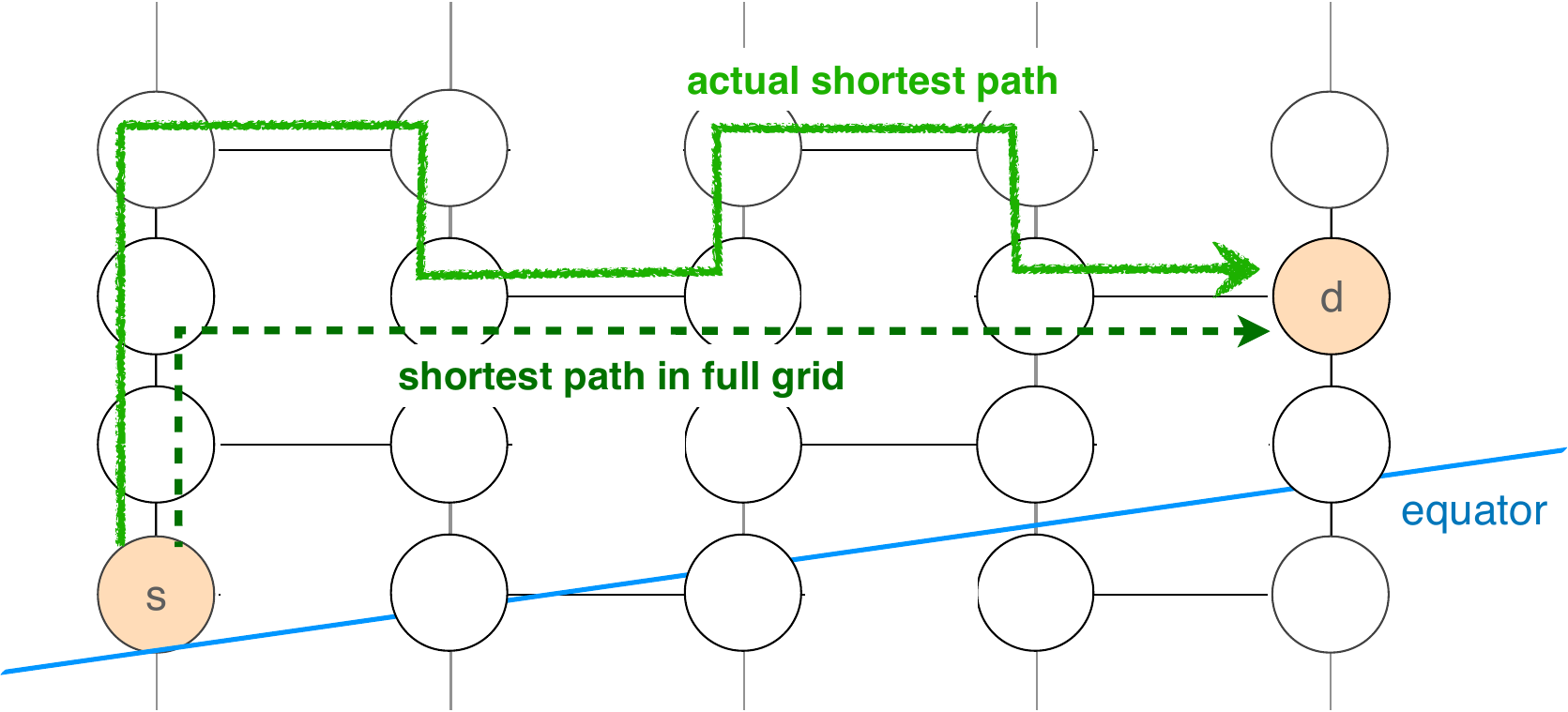}
	\caption{Example of extension of our theory to topologies where satellites have three ISLs.}
        \vspace{-0.1in}
	\label{fig:extension-example}
        \vspace{-0.1in}
\end{figure}

\system does not deal with congestion because constellations are expected to be well over-provisioned with respect to premium traffic.
This may not be true in very sparse topologies.
Even for these topologies, though, our work provides an interesting starting design.
For example, our fast rerouting technique may be effective against short-lived ISL congestion: when meant to cross a congested link, packets can indeed be fast rerouted with \system fast rerouting scheme rather than being dropped.

\myitem{Exploring intermediate designs.}
\system achieves reliability within a routing-free core.
Yet, better designs may be achieved if satellites perform some (limited) computations.

\system is amenable to extensions allowed by such computations.
For example, if they share information on active RF links, \system satellites may (i) deal with some RF link failures, and (ii) check if paths to be validated end at satellites having active RF links with destination GSes.

\myitem{Using \system techniques on the ground.}
\system is specifically designed for LEO constellations.

We however believe that \system's principles and techniques can be successfully ported to other networks too.
For example, resource-constrained devices in current IoT networks often run a minimalistic routing protocol (e.g.,~\cite{rpl-rfc6550}), which is both relatively inflexible (e.g., unable to dynamically reroute around low-battery devices) and energy-hungry~\cite{sdwsn-survey-access17}.
Adaptations of \system would offer alternative approaches to reduce energy consumption of IoT devices while also improving network reliability and routing flexibility.
Generally, \system can inspire approaches to reduce in-network power consumption, e.g., for sustainable networking.

%% file: relwork.tex
\section{Related Work}
\label{sec:relwork}

Pioneering work~\cite{ankit-hotnets18,mark-hotnets18,eth-hotnets18} shows the potential of LEO constellations to offer very low-latency paths.
In this work, we describe a system and technical solutions to implement a low-delay connectivity service, reliably.

A follow-up paper~\cite{mark-hotnets19} focuses on forwarding packets between satellites using ground stations as relays.
The paper also describes a derivative Dijkstra's algorithm to smooth the transition from one shortest path to the following.
In \system, we assume that GSes can use this algorithm, or variants of it, to source-route premium packets.
Our experiments (\S\ref{sec:eval}) confirm the feasibility of this approach for GSes.

A few works (e.g.,~\cite{tang-tmc19,hu-tnsm22}) propose routing algorithms that GSes or centralized controllers can use for bandwidth optimization.
We focus on satellite primitives aimed to ensure reliability.
As such, \system is complementary to the above algorithms, and can be combined with them.

Most of the literature on routing in satellite networks instead focuses on designs of distributed protocols, including adaptations of intra-domain protocols~\cite{leosurvey-arxiv23} and of geographical routing techniques~\cite{satnet-georouting-tnsm22}.
Motivated by the limited capabilities of solar-powered satellites, we explore a different design point where satellites are completely routing-oblivious.

A couple of previous contributions~\cite{destructresistantLEOrouting-fskd15,survivableLEOrouting-access20,starcure-infocom23} consider reliability in LEO mega-constellations.
Once again, they focus on the design of optimized mechanisms (e.g., adapted path discovery, local repair strategies, limited scope flooding) to deal with failures during reconvergence of satellite-run distributed routing protocols.
In \system, satellites do \emph{not} run any routing protocol in order to preserve the very scarce power they have access to.
Yet, \system satellites provide a fast rerouting scheme which effectively preserves low-delay connectivity from the very moment a failure is detected.

A different body of work focuses on topology and parameter optimization for satellite networks.
Different topologies are explored in~\cite{ankit-conext19}.
Optimal placement of ground stations is studied in~\cite{delportillo-actaastro19}.
Older contributions (e.g.,~\cite{topodesign-leoico95}) consider similar aspects in networks~\cite{iridium-IMSC93,globalstar-IMSC93,odyssey-IMSC93} smaller and inherently different from the constellations being deployed nowadays.
We focus on routing and reliability.
\approach assumes regular mesh topologies, consistently with most related work, but may be extended to others, as discussed in \S\ref{sec:discussion}.

More generally, routing has been studied in many contexts including mobile~\cite{fischer-icnp08} and large networks, even with some focus on low-delay paths recently~\cite{B4,ldr-sigcomm18}.
These contexts however do not entail networks with thousands of constantly moving, solar-powered satellites connected in mesh topologies.
LEO constellations are so unique that prior solutions are not a good fit: just running a routing protocol is expected to be impractical~\cite{cisco-arch11}, let alone supporting premium services.

%% file: theory-proofs.tex
\section{Routing Theory}
\label{app:theory}
\setcounter{mytheorem}{0}

This appendix includes formal proofs of the theorems presented in \S\ref{sec:theory}.

Unless explicitly specified, we consider networks with no failure, and say that a link is working if it is not failed.
Since we focus on the lowest delay paths between satellites, we interchangeably refer to links' delay and length, as well as to lowest delay path and shortest path.
When multiple shortest paths exist between a pair of satellites, we denote any of them as the shortest path.
We also use the terms node and satellite as synonyms.

\input{theory-model}

\input{theory-general}

\subsection{Grids moving in the same direction}
\label{app:samedir}

We say that the grid moves in the same direction if all the satellites in the grid move in the same direction -- e.g., either all north-east or all south-east in the SpaceX constellation's plan.
We start by proving properties that hold for all grids moving in the same direction, irrespectively of their type.

First of all, we note that for any grid moving in the same direction, links in any column can be sorted from the southmost to the northmost.
This is because the equator is not parallel to any link in any grid by Assumption~\ref{assume:no-perpendicular-eq}.
Additionally, since $s$ and $d$ are opposite corners in the grid, either $row(s)$ is the northmost row and $row(d)$ is the southmost one, or vice versa.
The following lemma then holds.

\begin{mylemma}\label{lem:optimal-in-column}
	For any grid moving in the same direction and every column $j$ in it, the shortest cross-orbit link in $j$ at any time cannot be a middle link.
\end{mylemma}
\begin{proof}

	Consider any middle link $l_m$ in column $j$.
	Since the equator is not parallel to any link in any grid by Assumption~\ref{assume:no-perpendicular-eq}, $l_m$ cannot be the northmost link nor the southmost one in its column.
	Let then $l_a \neq l_m$ be the northmost-row link in $j$, and $l_b$ be the southmost-row link in $j$.

	At any time $t$, one of the following cases must hold.
	\begin{itemize}
		\item If the midpoint of $l_m$ is between the equator and the northmost row, then $l_m$ is closer to the equator than $l_a$: formally, $\delta_t(l_a) > \delta_t(l_m)$. By Property~\ref{prop:hlink-length}, $l_a$ is then shorter than $l_m$.

		\item Similarly, if the midpoint of $l_m$ is between the equator and the southmost row, then $\delta_t(l_b) > \delta_t(l_m)$, and hence $l_b$ is then shorter than $l_m$ by Property~\ref{prop:hlink-length}.

		\item Otherwise, the equator must cross the midpoint of $l_m$. In this case, both $l_a$ and $l_b$ are farther away from the equator than $l_m$, and hence are both shorter than $l_m$.
	\end{itemize}

	In all the cases, at least one link among $l_a$ and $l_b$ is shorter than any middle link $l_m$ at any time $t$, which directly proves the statement. 
\end{proof}

Lemma~\ref{lem:optimal-in-column} implies that for each column, the source-row link or the destination-row one or both are always the shortest cross-orbit links in the column.

\begin{figure}[t]
	\centering
	\includegraphics[width=.7\columnwidth]{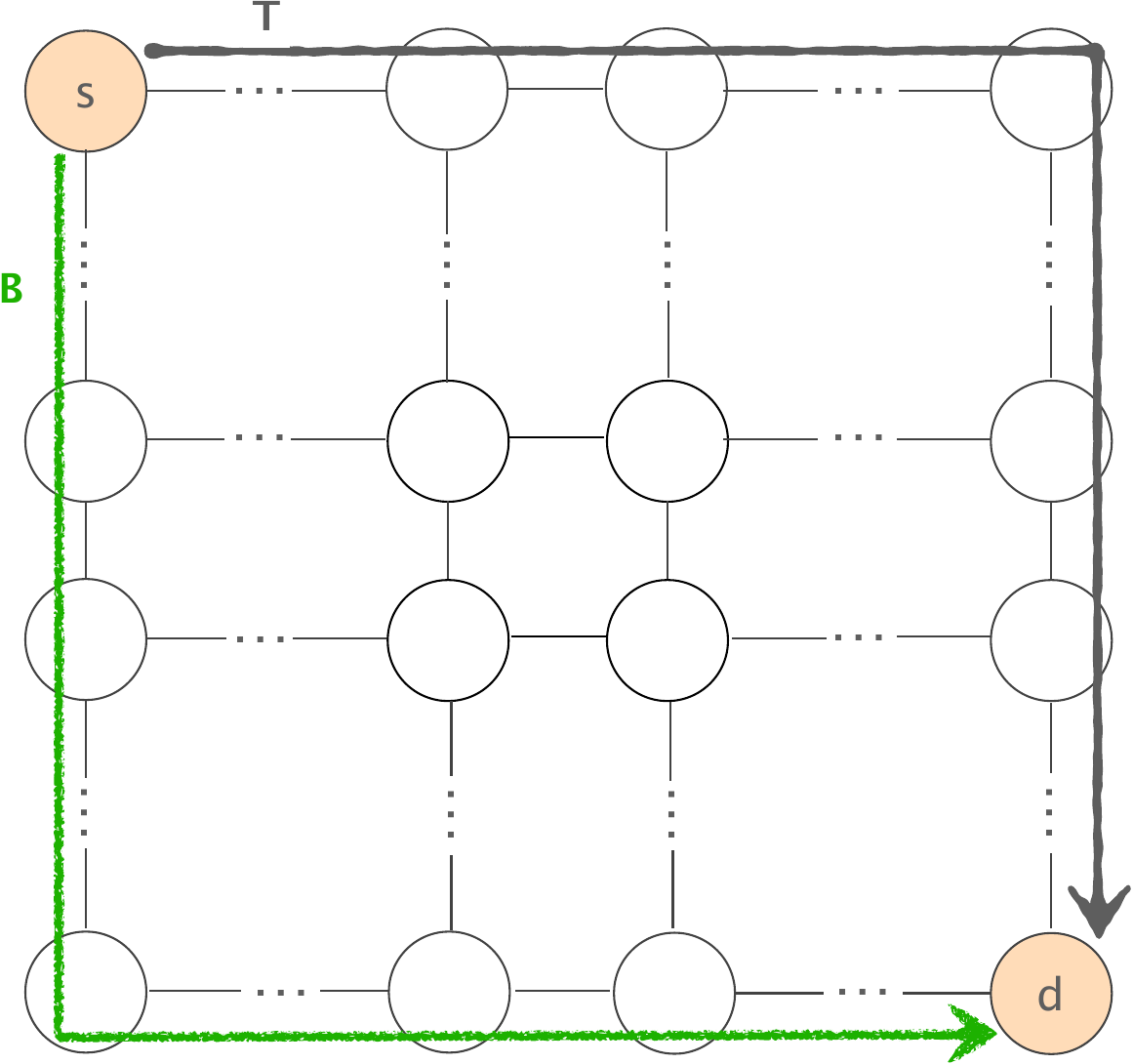}
	\caption{Illustration of the two notable paths $T$ and $B$ on a type-A grid, such as the one in Figure~\ref{fig:spacegrid-base}.}
	\label{fig:t-and-b}
\end{figure}

Building upon the reasoning of Lemma~\ref{lem:optimal-in-column}, we can derive the shape of the shortest paths in grids moving in the same direction when the equator does not cross any link in the grid.
To this end, we define two notable paths, $T$ and $B$, which are visualized in Figure~\ref{fig:t-and-b}.
\begin{itemize}
	\item Path $T$ is the $s$-$d$ path including all the cross-orbit links in the northmost row and a concatenation of intra-orbit links to reach the southmost row.
	\item Path $B$ is the $s$-$d$ path including all the cross-orbit links in the southmost row and a concatenation of intra-orbit links to reach the northmost row.
\end{itemize}

The following theorem proves that $T$ or $B$ is the shortest path when the equator does not cross any link in the grid.

\begin{corollary}\label{theo:sp-samedir-eq-outside}
	For any grid moving in the same direction, $T$ is the shortest path when all the nodes in the grid are north of the equator, and $B$ is the shortest path when all the nodes in the grid are south of the equator.
\end{corollary}
\begin{proof}
	When all the nodes in the grid are north of the equator, any cross-orbit link $l$ in the northmost row is farther away from the equator than the link in the same column of $l$ and in the southmost row. 
	By Lemma~\ref{lem:optimal-in-column}, $l$ is therefore the shortest link in its column. This implies that path $T$ includes a minimal number of cross-orbit links, each of which is the shortest in its column, and a minimal number of fixed-length intra-orbit links. Hence, no other path can be shorter than $T$.

	The statement follows by applying a symmetric reasoning to grids whose nodes are all south of the equator.
\end{proof}

Given Corollary~\ref{theo:sp-samedir-eq-outside}, in the following we focus on scenarios where the equator crosses at least one link in the grid.
We however note that the Theorems~\ref{theo:sp-descending-crossing} and~\ref{theo:sp-ascending-crossing} proved in the following sections are fully compatible with the above statement.

\smallskip

We additionally derive properties from the definition of grids' types.

For grids moving in the same direction, the following property characterizes \textbf{type-A} grids:
when the equator crosses a corner in the source row and at least one link in the grid, the source is farther away from the equator than any other node in the source row -- i.e., if we walk the source row starting from $s$, the equator becomes closer and closer.
Since $s$ and $d$ are opposite corners in the grid, a similar property holds for the destination row: when the equator crosses a corner in the destination row and at least one link in the grid, the destination satellite is the farthest away from the equator among all the nodes in the destination row.

Taking the perspective of cross-orbit links, we can reformulate this property of type-A grids as follows.
When the equator crosses a corner in the source row and at least one link in the grid, the cross-orbit link $(s,u)$ is the farthest away from the equator among all the links in $row(s)$.
Since the equator is not parallel to any link by Assumption~\ref{assume:no-perpendicular-eq}, the cross-orbit link in $row(s)$ and $col(d)$ 
is the closest to the equator, and hence (by Property~\ref{prop:hlink-length}) the longest, among all the source-row links.

Since intra-orbit links have all the same length, we can generalize this last observation for any cross-orbit link that is not located between the equator and the destination row.
\begin{property}\label{prop:descending-grids-hlink-sourcecol}
	In a type-A grid, when a cross-orbit link $l_1$ is not between the equator and $row(d)$, $l_1$ is closer to the equator than any other link $l_s$ such that $row(l_1) = row(l_s)$ and $col(l_s)$ is closer to the source column than $col(l_1)$.
\end{property}

Applying the definition of type-A grid to cross-orbit links in the destination row, we have a symmetric property for cross-orbit links that are not between the equator and the source row.
\begin{property}\label{prop:descending-grids-hlink-destcol}
	In a type-A grid, when a cross-orbit link $l_1$ is not between the equator and $row(s)$, $l_1$ is closer to the equator than any other link $l_d$ such that $row(l_1) = row(l_d)$ and $col(l_d)$ is closer to the destination column than $col(l_1)$.
\end{property}

In \textbf{type-B} grids, the source satellite is the source-row node closest to the equator when the equator crosses a corner in the source row and at least one link in the grid.
Since $s$ and $d$ are opposite corners, when the equator crosses a corner in the destination row and a link in the grid, the destination satellite is the closest node to the equator in its row.
As an example, Figure~\ref{fig:spacegrid-base} would represent a type-B grid if $c_1$ was the source and $c_2$ the destination.

Type-B grids have the following properties that are symmetric to Property~\ref{prop:descending-grids-hlink-sourcecol} and~\ref{prop:descending-grids-hlink-destcol}.
\begin{property}\label{prop:ascending-grids-hlink-destcol}
	In a type-B grid, when a cross-orbit link $l_1$ is not between the equator and $row(d)$,  $l_1$ is closer to the equator than any other link $l_d$ such that $row(l_1) = row(l_d)$ and $col(l_d)$ is closer to the destination column than $col(l_1)$.
\end{property}
\begin{property}\label{prop:ascending-grids-hlink-sourcecol}
	In a type-B grid, when a cross-orbit link $l_1$ is not between the equator and $row(s)$,  $l_1$ is closer to the equator than any other link $l_s$ such that $row(l_1) = row(l_s)$ and $col(l_s)$ is closer to the source column than $col(l_1)$.
\end{property}

\input{theory-samedir-descending}

\input{theory-samedir-ascending}

\subsection{Grids moving in different directions}
\label{sapp:diffdir}

We refer to grids including nodes moving in different directions as \textit{grids moving in different directions}.
We also indicate the two parallels that are the closest to respectively the north and south poles among those reached by any constellation's satellite as \textit{extreme parallels}.
By definition, grids moving in different directions cross at least one extreme parallel.

In the following subsections, we first consider grids close to only one extreme parallel, and then focus on grids with some satellites close to one extreme parallel and some close to the other extreme parallel.

\subsection*{\ref{sapp:diffdir}.1 Grids close to a single extreme parallel}

We say that a grid moving in different directions is \textit{close to a single extreme parallel} if
\begin{inparaenum}[(i)]
	\item  it includes satellites closer to one extreme parallel than to the equator, and
	\item both its source and destination rows are closer to the equator than to the other extreme parallel.
\end{inparaenum}
Intuitively, this means that some links in the grid are crossed by one extreme parallel and the entire grid is distant from the other extreme parallel.

In this section, we use the term extreme parallel to denote the one crossing the considered grid.

In grids close to a single extreme parallel, the closest cross-orbit links to the extreme parallel are the shortest links in their respective columns, as they are the farthest away from the equator.
We therefore reformulate Property~\ref{prop:hlink-length} as follows.

\begin{property}\label{prop:hlink-length-pole}
	If the midpoint of a cross-orbit link h1 is closer to the extreme parallel than the midpoint of another cross-orbit link h2, then h1 is shorter than h2.
\end{property}

Since the equator is parallel to the extreme parallel, we also reformulate Properties~\ref{prop:descending-grids-hlink-sourcecol},~\ref{prop:descending-grids-hlink-destcol},~\ref{prop:ascending-grids-hlink-destcol} and~\ref{prop:ascending-grids-hlink-sourcecol} as follows.

\begin{property}\label{prop:descending-grids-pole-hlink-sourcecol}
	In a type-A grid, when a cross-orbit link $l_1$ is not between the extreme parallel and $row(d)$, then $l_1$ is closer to the extreme parallel than any other link $l_s$ such that $row(l_1) = row(l_s)$ and $col(l_s)$ is closer to $col(s)$ than $col(l_1)$.
\end{property}

\begin{property}\label{prop:descending-grids-pole-hlink-destcol}
	In a type-A grid, when a cross-orbit link $l_1$ is not between the extreme parallel and $row(s)$, then $l_1$ is closer to the extreme parallel than any other link $l_d$ such that $row(l_1) = row(l_d)$ and $col(l_d)$ is closer to $col(d)$ than $col(l_1)$.
\end{property}

\begin{property}\label{prop:ascending-grids-pole-hlink-destcol}
	In a type-B grid, when a cross-orbit link $l_1$ is not between the extreme parallel and $row(d)$, $l_1$ is closer to the extreme parallel than any other link $l_d$ such that $row(l_1) = row(l_d)$ and $col(l_d)$ is closer to $col(d)$ than $col(l_1)$.
\end{property}

\begin{property}\label{prop:ascending-grids-pole-hlink-sourcecol}
	In a type-B grid, when a cross-orbit link $l_1$ is not between the extreme parallel and $row(s)$, $l_1$ is closer to the extreme parallel than any other link $l_s$ such that $row(l_1) = row(l_s)$ and $col(l_s)$ is closer to $col(s)$ than $col(l_1)$.
\end{property}

\input{theory-pole-descending}

\input{theory-pole-ascending}

\subsection*{\ref{sapp:diffdir}.2 Grids close to both extreme parallels}

Not all the grids moving in different directions are close to a single extreme parallel.
For example, a grid can be crossed by both extreme parallels, at different links.
Even if the grid is crossed by a single extreme parallel, its source may be close to the north pole and its destination to the south pole.
In both cases, such a grid does not match the definition of grid close to a single extreme parallel.
We denote grids moving in different directions but not close to a single extreme parallel as grids close to both extreme parallels: they are the focus of this section.

Clearly, Properties~\ref{prop:descending-grids-pole-hlink-sourcecol},~\ref{prop:descending-grids-pole-hlink-destcol},~\ref{prop:ascending-grids-pole-hlink-destcol} and~\ref{prop:ascending-grids-pole-hlink-sourcecol} as well as the above lemmas and theorems built on top of those properties do not hold for entire grids close to both extreme parallels.

Although not directly applicable, our theory provides us with a possibility to compute paths in any grid close to both extreme parallels, by dividing it in two.
We can indeed partition any such grid into two components, such that together they cover the entire original grid and each component is close to one extreme parallel.

The following theorem shows that the shortest path in the original grid can then be computed starting from the candidate shortest paths in the two sub-grids.

\Paste{theo:diffdir-both-poles}
\begin{proof}
	Suppose that the $s$-$d$ grid $G$ close to both extreme parallels is partitioned in two sub-grids $G_1$ and $G_2$, each moving in different directions and close to a single extreme parallel. 
	Since they cover all the links in $G$, $G_1$ and $G_2$ must have a common border.

	Consider the shortest path $P$ within $G$ at any time $t$.
	Since $s$ and $d$ are opposite corners in any $s$-$d$ grid, $P$ must cross at least one node $x$ in the common border of $G_1$ and $G_2$.

	By a property of shortest paths, $P$ must be the concatenation of $P_1$ and $P_2$, where $P_1$ is the shortest path from $s$ to $x$, and $P_2$ is the shortest path from $x$ to $d$. 
	Since $P$ is contained in $G$, $P_1$ is a path internal to $G_1$, and $P_2$ is internal to $G_2$. 

	Assume by contradiction that $P_1$ crosses at least one cross-orbit link which is not in the set $\mathcal{C}_t$ of cross-orbit links considered for the shortest path computation in $G_1$ at time $t$.
	We now prove that this assumption leads to a contradiction.

	Let $(u,v)$ be the first cross-orbit link in $P_1$ which is not in $\mathcal{C}_t$.
	We can then write $P_1$ as the concatenation of three sub-paths: $C_1 = (s \dots u)$, $(u,v)$, and $R_1 = (v \dots d)$, where all the cross-orbit links in $C_1$ (if any) are elements of $\mathcal{C}_t$.

	We now show that there exists at least one link $(w,z)$ in the column of $(u,v)$, such that $(w,z)$ can be used to build a path $Q_1$ shorter than $P_1$.
	We have the following two cases.

	\begin{itemize}
		\item the row of $(u,v)$ is closer to the source row than any other link in $\mathcal{C}_t$ belonging to its column. In this case, we define $(w,z)$ as the link in $\mathcal{C}_t$ whose column is $col(u,v)$ and whose row is the closest to $row(s)$ among all the links in $\mathcal{C}_t$ and in $col(u,v)$. We also define $a$ as the first node in $row(w,z)$ crossed by $P_1$: $a$ must exist because $(w,z)$ is closer to the destination row than $(u,v)$.

		We then build $Q_1$ starting from $P_1$, and replacing the sub-path $P_a = (u\ v \dots a)$ with the path $Q_a$ that concatenates a minimal sequence of intra-orbit links from $u$ to $w$ with a sequence of cross-orbit links $(w\ z \dots a)$.

		Let's now compare the length of $P_a$ and $Q_a$ at $t$.
		By definition of candidate shortest paths in any grids moving in different directions and close to a single extreme parallel, $(w,z)$ is closer to the extreme parallel, and hence shorter, than $(u,v)$ at $t$. Indeed, by construction of the candidate shortest paths, $(w,z)$ is between $(u,v)$ and the extreme parallel if $G_1$ is a type-A grid, and $(w,z)$ is the closest link to the extreme parallel in its column if $G_1$ is a type-B grid. The same property holds for every pair of cross-orbit links $h_p$ and $h_q$ such that $h_p \in P_a$, $h_q \in Q_a$, and $h_p$ and $h_q$ belong to the same column. Since the number of intra-orbit links in $Q_a$ is minimal, we conclude that $Q_a$ is shorter than $P_a$.

		Consequently, since $P_1$ and $Q_1$ differ only in the sub-path from $u$ to $a$, $Q_1$ is shorter than $P_1$.

		\item the row of $(u,v)$ is closer to the destination row than any other link in $\mathcal{C}_t$. We then define $(w,z)$ as the link in $\mathcal{C}_t$ whose column is $col(u,v)$ and whose row is the closest to $row(d)$ among all the links in $\mathcal{C}_t$. 

		Since $(u,v)$ is defined as the first link in $P_1$ not belonging to $\mathcal{C}_t$ and shortest paths are guaranteed to always cross columns and rows closer and closer to the destination by Lemma~\ref{lem:sp-no-step-down} and~\ref{lem:sp-no-step-aside}, $P_1$ must cross $w$, and must include a sub-path $P_w = (w \dots u\ v)$ formed by a sequence of intra-orbit links followed by $(u,v)$.

		We build $Q_1$ starting from $P_1$, and replacing the sub-path $P_w$ with the path $Q_w$ that concatenates $(w,z)$ with a minimal sequence of intra-orbit links from $z$ to $v$.

		Let's now compare the length of $P_w$ and $Q_w$ at $t$.
		By definition of candidate shortest paths in any grid moving in different directions and close to a single extreme parallel, $(w,z)$ is closer to the extreme parallel, and hence shorter, than $(u,v)$ at $t$.
		Indeed, by construction of the candidate shortest paths, $(w,z)$ is between $(u,v)$ and the extreme parallel if the grid is a type-A grid, and $(w,z)$ is the closest link to the extreme parallel in its column if the grid is a type-B grid. Since $P_w$ and $Q_w$ include the same number of fixed-size intra-orbit links, then $Q_w$ is shorter than $P_w$.

		Consequently, since $P_1$ and $Q_1$ differ only in the sub-path from $u$ to $w$, $Q_1$ must be shorter than $P_1$.
	\end{itemize}

	Note that there is no other case because the candidate shortest paths include cross-orbit links in consecutive rows for type-A grids (see Theorem~\ref{theo:shortest-path-descending-grid-wrapping}) and only the cross-orbit links closest to the extreme parallel in their column for type-B grids (see Theorem~\ref{theo:shortest-path-ascending-grid-wrapping}).

	In both the above cases, we have built a path $Q_1$ from $s$ to $x$, shorter than $P_1$.	
	If $Q_1$ includes loops, Lemma~\ref{lem:shortcut-loops} ensures that we can build an even shorter path $Q_2$.
	The existence of a path shorter than $P_1$ contradicts the hypothesis that $P_1$ is the shortest path from $s$ to $x$.

	The statement then follows by noting that a similar argument holds if we assume by contradiction that $P_2$ includes at least one cross-orbit link not crossed by candidate shortest paths within $G_2$.
\end{proof}

\input{theory-frr}

%% file: theory-model.tex
\subsection{Model}
\label{app:model}

We start by introducing some notation and describing our network model.
This model captures the characteristics of satellite networks that are relevant from a path computation viewpoint, and
includes the definition of grid and grid's type that are central to our routing theory.

\myitem{Links' length.}
The delay over any cross-orbit link depends on the distance between the link's endpoints, which in turn changes as the satellites move along their orbits.
The delay is instead fixed for intra-orbit links as the distance between consecutive satellites in the same orbit is constant over time.

More precisely, the following two properties hold for the satellite networks we consider.

\begin{property}\label{prop:vlink-length}
	All intra-orbit links have the same fixed length.
\end{property}

\begin{property}\label{prop:hlink-length}
	If the midpoint of a cross-orbit link $h_1$ is farther away from the equator than the midpoint of another cross-orbit link $h_2$, then $h_1$ is shorter than $h_2$.
\end{property}

For any link $l$, we denote the distance of its midpoint to the equator at time $t$ as $\delta_t(l)$, and its length at $t$ as $len_t(l)$.

According to this notation, we can formalize the above properties.
Property~\ref{prop:vlink-length} states that for any pair of intra-orbit links $h_1$ and $h_2$, we have $len_t(h_1) = len_t(h_2)$ at any time $t$.
Property~\ref{prop:hlink-length} instead states that for any time $t$ and any two cross-orbit links $h_1$ and $h_2$, the following logical implication holds: $\delta_t(h_1) > \delta_t(h_2) \Rightarrow len_t(h_1) < len_t(h_2)$.

We also denote the length of any path $Q$ at time $t$ as $len_t(Q)$. 
We assume the following property about the lengths of paths including only intra-orbit or only cross-orbit links.

\begin{assumption}\label{ass:ingrid-vs-outgrid-simple}
	If satellites $s$ and $d$ are the endpoints of a sequence of intra-orbit links or of a sequence of cross-orbit links, no path including a combination of intra- and cross-orbit links can be the shortest path from $s$ to $d$.
\end{assumption}

For example, if we consider the portion of satellite network displayed in Figure~\ref{fig:spacegrid-base}, Assumption~\ref{ass:ingrid-vs-outgrid-simple} implies that the shortest path from $s$ to $c_1$ includes only cross-orbit links, while the shortest path from $s$ to $c_2$ contains only intra-orbit links.

Assumption~\ref{ass:ingrid-vs-outgrid-simple} would follow from geometric properties of polygons if links in the satellite networks were segments.
Since they actually are arcs on a spherical surface, the above property is not guaranteed in all possible constellation designs.
We however checked empirically that Assumption~\ref{ass:ingrid-vs-outgrid-simple} holds for topologies considered in~\cite{mark-hotnets18}, which are compatible with the documents filed by SpaceX.

\myitem{Sub-graphs.}
Consider any sub-graph $N$ of the satellite network such that $N$ has a single connected component.

We define the \textbf{border} of $N$ as the set of nodes in $N$ that are connected to at least one node outside $N$. Formally, the border of $N$ is the set $\{x | x \in N \ \land\  \exists (x,y), y \not \in N\}$.
We refer to any node in the border of $N$ as \textbf{border node}, while we use the term \textbf{internal node} to indicate a node in $N$ which does not belong to its border.

We additionally call \textbf{corner} of $N$ any border node $c$ such that:
\begin{inparaenum}[(i)]
	\item $c$ is the extreme of exactly one intra-orbit link in $N$, and
	\item $c$ is the extreme of exactly one cross-orbit link in $N$.
\end{inparaenum}
Two nodes $c_1$ and $c_2$ are \textbf{opposite corners} if they are both corners of $N$, and there exists no path $(c_1 \dots c_2)$ that includes only intra-orbit links or only cross-orbit ones.

\myitem{Source-destination grids.}
Given a source and a destination satellite, we use the term source-destination grid, or simply \textbf{grid}, to denote a sub-graph of the network such that (i) the source and the destination are opposite corners in it, and (ii) the union of links between consecutive nodes in its border form a cycle including exactly two sequences of cross-orbit links and two sequences of intra-orbit links.

Intuitively, a grid is a sub-graph with a rectangular shape, where the source and destination do not share any side.
Unless stated otherwise, for any grid, we always denote the source satellite as $s$, and the destination satellite as $d$.

\begin{figure}[t]
	\centering
	\includegraphics[width=.9\columnwidth]{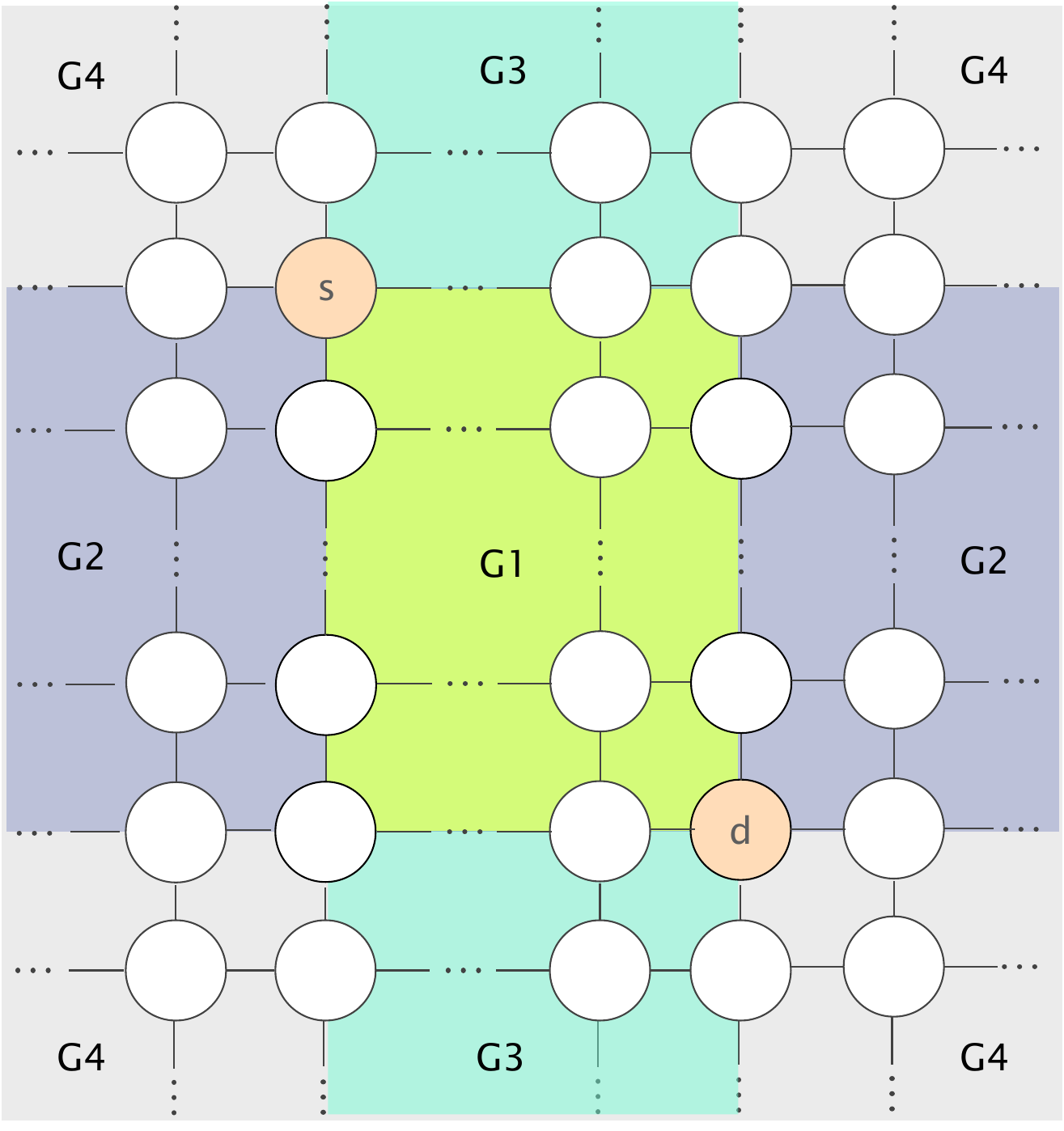}
	\caption{Illustration of the four grids, named G1, G2, G3 and G4, defined by any pair of satellites $s$ and $d$ in constellations where sequences of cross-orbit links start and end with the same node.}
	\label{fig:spacegrid-base-four-grids}
\end{figure}

Since each satellite follows a closed orbit and always maintains network links with the same four neighboring satellites, every $s$-$d$ pair defines exactly four grids, whose nodes and links do not change over time.
An example of source-destination grid is displayed in Figure~\ref{fig:spacegrid-base}. Using the same graphical notation, Figure~\ref{fig:spacegrid-base-four-grids} visualizes the four grids defined by any pair of satellites in constellations where sequences of cross-orbit links start and end with the same node.

For any grid $G$, we define rows and columns as follows.
A \textbf{row} is a set of cross-orbit links that form a path between two nodes $u$ and $v$ both in the border of $G$.
We call \textit{source row} the row including the cross-orbit link having $s$ as one extreme -- remember that there exists only one cross-orbit in $G$ that has $s$ as extreme by definition of corner. Similarly, we use the term \textit{destination row} to indicate the row with the cross-orbit link adjacent to $d$.

A \textbf{column} is a set of cross-orbit links such that the corresponding nodes form exactly two sequences of intra-orbit links, each connecting one node in the source row with one node in the destination row.
We refer to the column that includes the cross-orbit link $(s,u)$ with $u \in G$ as \textit{source column}, and to the column containing the cross-orbit link $(d,v)$ with $v \in G$ as \textit{destination column}.

We assume that the equator is not parallel to any cross-orbit link nor to any inter-orbit link in the constellation: this is in-line with deployment plans for actual satellite constellations~\cite{mark-hotnets18}.
From the perspective of satellite grids, the above assumption can be expressed as follows.
\begin{assumption}\label{assume:no-perpendicular-eq}
	Given any grid $G$, the equator crosses the midpoint of at most one link per column in $G$, and it is not parallel to any link in $G$.
\end{assumption}

For any cross-orbit link $l \in G$, we denote the row to which $l$ belongs as $row(l)$; similarly, we denote $l$'s column as $col(l)$.
Slightly abusing notation, we additionally indicate the source and destination rows as $row(s)$ and $row(d)$ respectively.
Similarly, we denote the source and destination columns as $col(s)$ and $col(d)$ respectively.
We use the term \textbf{middle link} to indicate any cross-orbit link $m$ which does not belong to either the source nor the destination row: that is, $m \not \in row(s) \land m \not \in row(d)$.

We also say that a row $r_1$ is closer to row $r_2$ than row $r_3$ if links in $r_1$ are topologically closer (i.e., in terms of number of hops) to links in $r_2$ than to links in $r_3$.
More precisely, consider any three nodes $n_1, n_2, n_3$ such that $n_1$ is an extreme of links in $r_1$, $n_2$ is an extreme of links in $r_2$,  $n_3$ is an extreme of links in $r_3$, and sequences of intra-orbit links exist from $n_1$ to $n_2$ and from $n_2$ and $n_3$.  If and only if the sequence of intra-orbit links from $n_1$ to $n_2$ has less hops than the sequence of intra-orbit links from $n_1$ to $n_3$, then $r_1$ is closer to row $r_2$ than row $r_3$.
We use a similar terminology for columns.

\myitem{Grid types.}
To capture the features that influence shortest paths from $s$ to $d$ in any grid, we define two types of grids: type-A and type-B grids.

Let $c_1$ be the satellite at the other end of the source row from $s$. As the grid orbits, it will reach a point where some of its links cross the equator and either $s$ or $c_1$ is on the equator. If in this moment it is $s$ that is on the equator, then the grid is a Type-B grid, otherwise it is a Type-A grid.

Figure~\ref{fig:spacegrid-base} shows a type-A grid. In fact, it is $c_1$ to be on the equator, when the equator crosses at least one intra-orbit link in the grid.

A degenerated case is represented by grids consisting of a single sequence of cross-orbit links, or of a single sequence of intra-orbit links: we denote them as \textbf{single-path grids}.

%% file: theory-general.tex
\subsection{General properties of shortest paths}

We now prove basic properties upon which our theory is built.
We start with a lemma that we use in multiple proofs.

\begin{mylemma}~\label{lem:shortcut-loops}
	Given a satellite network, let $P = (u \dots v)$ be a path in it such that $P$ includes at least one loop.
	A path $P_1$ from $u$ to $v$ exists in the network that has no loops, crosses a subset of links in $P$, and is shorter than $P$.
\end{mylemma}
\begin{proof}
	Since $P$ has at least one loop by hypothesis, it must be possible to write $P$ as the concatenation of four sub-paths $(u \dots x)$, $(x \dots y)$, $(y \dots x)$ and $(x \dots v)$.

	To prove the statement, we build a new path by simply removing the sub-paths $(x \dots y)$ and $(y \dots x)$ that constitute the loop itself.
	This operation indeed leaves us with a path $P_1 = (u \dots x) (x \dots v)$ which is still a path from $u$ to $v$, crosses a subset of the links in $P$, and is shorter than $P$ because all the links have positive weights (i.e., delays) in satellite networks.

	Note that $P_1$ can still have loops, because $P$ may have other loops in addition to the one between $x$ and $y$.
	The statement follows by iterating the above reasoning until we obtain a path with no loops.
\end{proof}

\myitem{Shortest paths are internal to a grid.}
We now show that the shortest path between any pair of satellites is always internal to one of the grids they form.

If a single-path grid exists between $s$ and $d$, Assumption~\ref{ass:ingrid-vs-outgrid-simple} directly implies the following lemma.

\begin{mylemma}\label{lem:ingrid-vs-outgrid-singlepath}
	If $s$ and $d$ form a single-path grid with all working links, the shortest path is internal to a single-path grid.
\end{mylemma}
\begin{proof}
	By definition of single-path grid, $s$ and $d$ must be the endpoints of either a sequence of cross-orbit links, or a sequence of intra-orbit links.	

	Suppose for now that $s$ and $d$ are connected by a sequence of cross-orbit links.
	This implies that there exists at least one single-path $s$-$d$ grids $G$ containing only cross-orbit links.
	Assumption~\ref{ass:ingrid-vs-outgrid-simple} then guarantees that the shortest path is either the only path in $G$ or a path in any other single-path grid that $s$ and $d$ may form.

	A symmetric argument holds if $s$ and $d$ are extremes of a sequence of intra-orbit links, yielding the statement.
\end{proof}

We now consider pairs of satellites that do not form any single-path grid.

\begin{mylemma}\label{lem:ingrid-vs-outgrid-nonsinglepath}
	If $s$ and $d$ do not form a single-path grid, the shortest path from $s$ to $d$ is entirely contained in one of their source- destination grids whenever all the links between satellites in the border of their grids are working.
\end{mylemma}
\begin{proof}
	Assume by contradiction that the shortest path from $s$ to $d$ is a path $Q$ includes some nodes internal to a grid and other external to it.

	Since by definition the four grids formed by $s$ and $d$ cover the entire network, $Q$ must cross nodes that are internal to two different grids.
	Let $G_1$ and $G_2$ be the last two grids that $Q$ crosses.

	We can then write $Q$ as the concatenation of  sub-paths $A\ (i_1,b_1)\ B\ (b_2,i_2)\ I$, where: 
	$A$ is any non-empty path from $s$ to $i_1$; 
	$i_1$ is a node internal to $G_1$, and hence $i_1 \neq s$; 
	$b_1$ and $b_2$ are nodes in the common border of $G_1$ and $G_2$, with possibly $b_2 = b_1$; 
	$B$ is a possibly empty path from $b_1$ to $b_2$ including only nodes in the common border of $G_1$ and $G_2$; 
	$i_2$ is a node internal to $G_2$, and hence $i_2 \neq d$; 
	and $I$ is a non-empty path in $G_2$ that terminates in $d$.
	We refer to Figure~\ref{fig:spacegrid-general-sp-within-one-grid} for a visualization of the nodes in $Q$ just defined.

	By definition, the border of any grid which is not a single-path one includes two sequences of cross-orbit links and two sequences of intra-orbit links. We thus have two cases.
	\begin{itemize}
		\item There exists a path from $b_1$ to $d$ composed only of nodes in the border of $G_2$ and including solely cross-orbit links or only intra-orbit ones. Note that $B\ (b_2,i_2)\ I$ is also a path from $b_1$ to $d$ that however includes a mix of cross-orbit and intra-orbit links, given that $i_2$ is an internal node of $G_2$. In this case, $B\ (b_2,i_2)\ I$ cannot be the shortest path from $b_1$ to $d$ by Assumption~\ref{ass:ingrid-vs-outgrid-simple}.
		\item There exists no path from $b_1$ to $d$, composed only of nodes in the border of $G_2$ and including solely cross-orbit links or only intra-orbit ones. Since $s$ and $d$ are opposite corners of every grid they form, it must then exist a path from $b_1$ to $s$ composed only of nodes in the border of $G_2$ and including solely cross-orbit links or only intra-orbit ones. Consider now $A\ (i_1,b_1)$: this path includes both cross-orbit and intra-orbit links because $i_1$ is an internal node of $G_1$. Hence, $A\ (b_1,i_1)$ cannot be the shortest path from $s$ to $b_1$ by Assumption~\ref{ass:ingrid-vs-outgrid-simple}.
	\end{itemize}

	In both cases, we end up with a sub-path of $Q$, starting at a node $x$ and ending at another node $y$, which is not the shortest path from $x$ to $y$.
	We can therefore replace such a sub-path with the actual shortest path from $x$ to $y$, and obtain a path $Q_1$ shorter than $Q$.
	If $Q_1$ contains loops, Lemma~\ref{lem:shortcut-loops} ensures that we can remove possible loops in $Q_1$, obtaining an even shorter path $Q_2$.

	The possibility to build a path shorter than $Q$ contradicts the hypothesis that $Q$ is the shortest path from $s$ to $d$, and therefore yields the statement.
\end{proof}

\begin{figure}[t]
	\centering
	\includegraphics[width=.9\columnwidth]{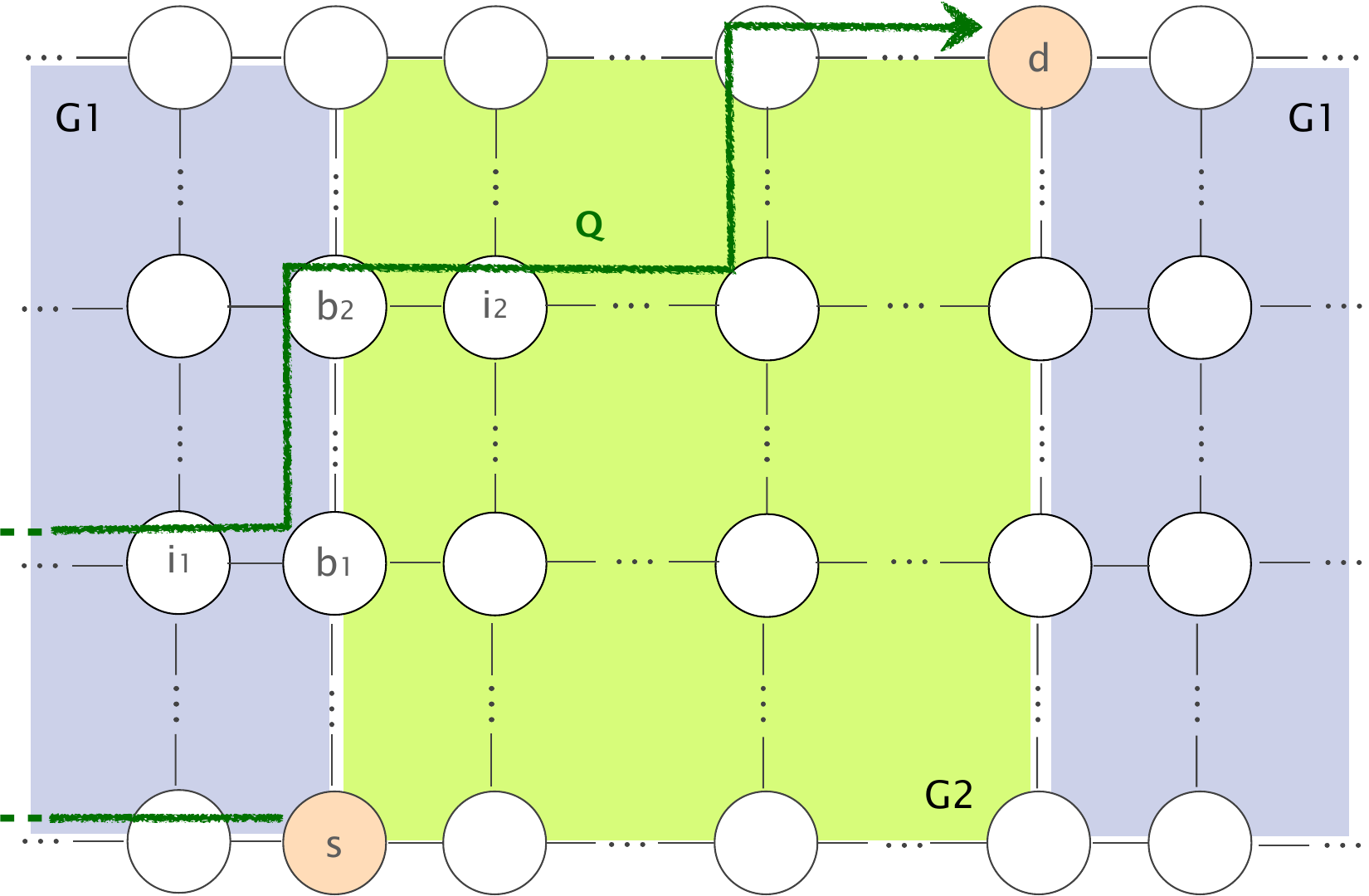}
	\caption{Visual support for the proof of Lemma~\ref{lem:ingrid-vs-outgrid-nonsinglepath}.}
	\label{fig:spacegrid-general-sp-within-one-grid}
\end{figure}

\Paste{theo:ingrid}
\begin{proof}
	Any pair of satellites either form at least one single-path grid, or all their grids are not single-path ones.
	The statement follows by noting that Lemma~\ref{lem:ingrid-vs-outgrid-singlepath} applies in the former case, while Lemma~\ref{lem:ingrid-vs-outgrid-nonsinglepath} holds in the latter one.
\end{proof}

\myitem{Additional shortest paths' features.}
Shortest paths also satisfy the constraints formalized by the following lemmas.

\begin{mylemma}~\label{lem:sp-no-step-down}
	Within any grid, the shortest path $P$ from $s$ to $d$ never includes two sequences $H_1$ to $H_2$ of cross-orbit links such that: (i) $P$ crosses $H_1$ before $H_2$, and;  (ii) $H_1$ is closer to the destination row than $H_2$.
\end{mylemma}
\begin{proof}
	Assume by contradiction that the shortest path $P$ from $s$ to $d$ traverses two sequences $H_1$ and $H_2$ of cross-orbit links such that $P$ crosses $H_1$ before $H_2$, and $H_1$ is closer to the destination row than $H_2$.

	We denote the first node in $H_1$ as $x$.
	Slightly abusing notation, we also denote the rows of the links belonging to $H_1$ and $H_2$ as $row(H_1)$ and $row(H_2)$ respectively.

	According to what we have assumed by contradiction, $row(H_2)$ is farther away from $row(d)$ than $row(H_1)$. Since $P$ terminates at $d$, $P$ must include at least one node $y$ in $row(H_1)$ after having crossed $H_2$.

	This implies that we can write $P$ as $(s \dots x \dots y \dots d)$  where both $x$ and $y$ are in $row(H_1)$, and $P$'s sub-path from $x$ to $y$ includes $H_1$ and $H_2$.
	Since $row(H_1) \neq row(H_2)$, $P$'s sub-path from $x$ and $y$ is not a sequence of cross-orbit links.
	Assumption~\ref{ass:ingrid-vs-outgrid-simple} then ensures that $P$'s sub-path cannot be the shortest path from $x$ to $y$. In turn, this implies that we can build a path $P_1$ shorter than $P$. An even shorter path can be built by removing loops in $P_1$, if any,  as proved by Lemma~\ref{lem:shortcut-loops}.

	The existence of a path shorter than $P$ contradicts the hypothesis that $P$ is the shortest path from $s$ to $d$, hence yielding the statement.
\end{proof}

\begin{mylemma}~\label{lem:sp-no-step-aside}
	Within any grid, the shortest path $P$ from $s$ to $d$ never crosses two sequences $V_1$ to $V_2$ of intra-orbit links such that: (i) $P$ crosses $V_1$ before $V_2$, and;  (ii) $V_1$ is closer to the destination column than $V_2$.
\end{mylemma}
\begin{proof}
	The statement is proved by an argument symmetrical to the proof of Lemma~\ref{lem:sp-no-step-down}. Indeed, if a path $P$ crosses $V_1$ before $V_2$, and $V_1$ is closer to $col(d)$ than $V_2$, there must exist two nodes $x$ and $y$ such that (i) $x$ and $y$ are connected by a sequence of intra-orbit links, and (ii) the sub-path of $P$ from $x$ to $y$ includes a mix of cross-orbit and intra-orbit links. Hence, Assumption~\ref{ass:ingrid-vs-outgrid-simple} implies that $P$ cannot be the shortest path from $s$ to $d$. 
\end{proof}

%% file: theory-samedir-descending.tex
\subsection*{Type-A grids}

Consider a generic type-A grid, as the one visualized in Figure~\ref{fig:spacegrid-base}.
We start by showing a couple of properties, formalized in the following lemmas.

\begin{mylemma}\label{lem:optimal-preds-succs-descending}
	Consider any type-A grid and any column $j$ in it.
	At any time, if the cross-orbit link $l$ in $row(s)$ and column $j$ is one of the shortest links in $j$, every link in the sub-path of the source row from $s$ and $l$
	is the shortest one in its column.
\end{mylemma}
\begin{proof}
	Let $l_s$ be a source-row link such that $l_s$ is the shortest link in its column $j$ at a time $t$. Additionally, let $l_d$ be the destination-row link in column $j$.

	By Property~\ref{prop:hlink-length}, $l_s$ must be more distant to the equator than any other cross-orbit link in column $j$, including $l_d$.
	Formally, $\delta_t(l_{s}) \geq \delta_t(l_{d})$.
	This is possible only if $l_s$ is not between the equator and the destination row.

	Consider any link $h_s$ in the sub-path of $row(s)$ from $s$ and $l$.
	Let $k$ be the column of $h_s$, and let $h_d$ be the destination-row link in column $k$.
	Notation is displayed in Figure~\ref{fig:spacegrid-equator-descending-columns}.

	We have the following cases.
	\begin{itemize}
		\item $h_d$ is between the equator and the source row. In this case, we have $\delta_t(h_d) < \delta_t(h_s)$.
		\item $h_d$ is not between the equator and the source row. By definition, $row(h_d) = row(l_d)$, and $col(l_d) = j$ is closer to the destination column than $col(h_d) = k$. We then have $\delta_t(h_d) < \delta_t(l_d)$ by Property~\ref{prop:descending-grids-hlink-destcol}.
		Additionally, $l_s$ is not between the equator and the destination row by hypothesis, as mentioned above. 
		Since $row(h_s) = row(l_s)$ and $col(h_s) = k$ is closer to the source column than $col(l_s) = j$ by definition, we additionally have $\delta_t(l_s) < \delta_t(h_s)$ by Property~\ref{prop:descending-grids-hlink-sourcecol}.
		Considering all the above observations together, we conclude that $\delta_t(h_d) < \delta_t(l_d) \leq \delta_t(l_s) < \delta_t(h_s)$.
	\end{itemize}

	In both cases, $h_d$ is closer to the equator than $h_s$.
	By Property~\ref{prop:hlink-length}, $h_{d}$ is therefore longer than $h_{s}$.
	Lemma~\ref{lem:optimal-in-column} then ensures that $h_{s}$ is the shortest cross-orbit link in column $k$, which directly proves the statement.
\end{proof}

\begin{figure}[t]
	\centering
	\includegraphics[width=\columnwidth]{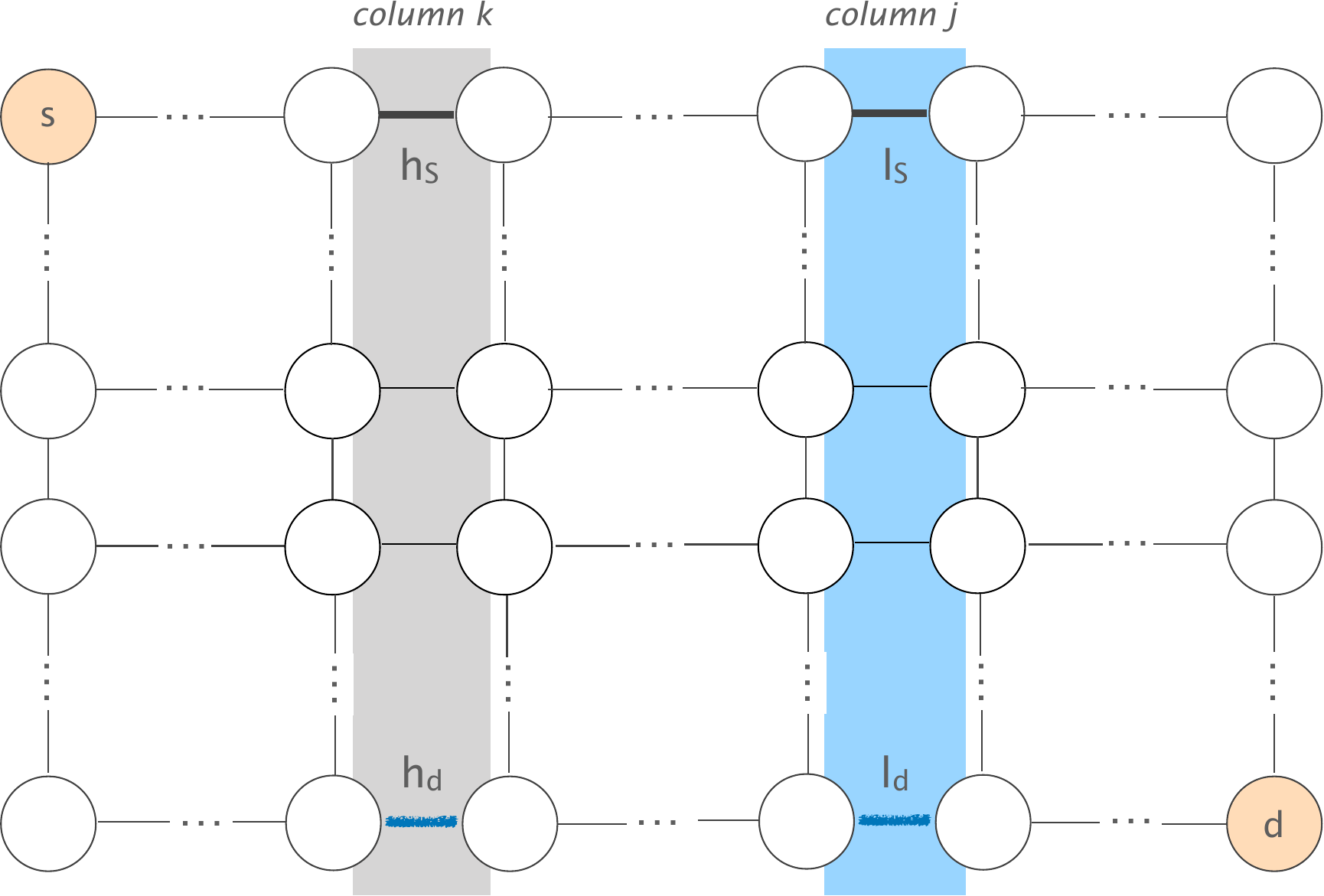}
	\caption{Visual support for Lemma~\ref{lem:optimal-preds-succs-descending}.}
	\label{fig:spacegrid-equator-descending-columns}
\end{figure}

\begin{mylemma}\label{lem:optimal-succs-descending}
	Consider any type-A grid and any column $m$ in it.
	At any time, if the cross-orbit link $l$ in $row(d)$ and column $m$ is one of the shortest links in $m$, every link in the sub-path of the destination row from $l$ to $d$ is the shortest one in its column.
\end{mylemma}
\begin{proof}
	An argument symmetric to the one in the proof of Lemma~\ref{lem:optimal-preds-succs-descending} yields the statement.	
\end{proof}

We now define the shape of the shortest path in type-A grids when they move in the same direction.

Let $\mathcal{F}_t$ be the set of cross-orbit links such that at time $t$, every link  in $\mathcal{F}_t$ is the farthest away from the equator among the cross-orbit links in its column
\footnote{Note that at specific times, the source-row and destination-row links of one column in the grid may have the same length, and hence both be the shortest in the column.When this happens, $s$ and $d$ have two shortest paths of the same length.	Since we are interested in any shortest path, we define $\mathcal{F}_t$ to include only one of the two shortest links in the column.}.
Lemma~\ref{lem:optimal-in-column} guarantees that $\mathcal{F}_t$ includes only links in the source row or in the destination row, when grids move in the same direction.

The shortest path $P$ from $s$ to $d$ is the path that includes all the links in $\mathcal{F}_t$ and the sequence of intra-orbit links connecting links in $\mathcal{F}_t$ that belong to adjacent columns and different rows.
Figure~\ref{fig:spacegrid-equator-descending-sp} visualizes such a path $P$ along with sub-paths formally defined below.

The following lemma shows that $P$ is a non-loopy path from $s$ to $d$.
We use this lemma to prove that $P$ is indeed the shortest path from $s$ to $d$ in Theorem~\ref{theo:sp-descending-crossing}.

\begin{mylemma}\label{lem:sp-exists-descending-crossing}
	Links in $\mathcal{F}_t$ form a sequence of source-row cross-orbit links $H_s = (s \dots x)$ and a sequence of destination-row cross-orbit links $H_d = (z \dots d)$ such that there exists a sequence of intra-orbit links $V_p = (x \dots z)$.
\end{mylemma}
\begin{proof}
	By definition, $\mathcal{F}_t$ includes one cross-orbit link per column in the grid, and each links in $\mathcal{F}_t$ is either in the source row or in the destination row by Lemma~\ref{lem:optimal-in-column}.

	If all links in $\mathcal{F}_t$ are destination-row links, then $x=s$, $H_s = (s)$, $z$ is the grid's corner in $row(d)$ other than $d$, $H_d$ spans the entire destination row, and there exists a sequence of intra-orbit links from $s$ to $z$ by definition of grid. This observation immediately yields the statement.

	Similarly, if $\mathcal{F}_t$ includes only source-row links, the statement follows by noting that $z = d$, $H_d = (d)$, $x$ is the grid's corner in $row(s)$ other than $s$, $H_s$ spans the entire source row, and there indeed exists a sequence of intra-orbit links from $d$ to $x$, by definition of grid.

	We thus focus now on the case where $\mathcal{F}(t)$ includes both source-row and destination-row links.

	Let $h_s \in \mathcal{F}(t)$ be the source-row link in a column $j$ such that there exists no source-row link in $\mathcal{F}_t$ belonging to a column closer than $j$ to $col(d)$.
	By Lemma~\ref{lem:optimal-preds-succs-descending}, 
	$\mathcal{F}_t$ must include all source-row links in columns closer than $j$ to $col(s)$. The union of those links forms a path $H_s = (s \dots x)$, where $x$ is an extreme of $h_s$.

	Consider now the column $k$ which is next to $j$ and closer to $d$ than $j$, with possibly $k=col(d)$.
	By definition of $h_s$ and by Lemma~\ref{lem:optimal-in-column}, a destination-row link $h_d \in {\mathcal{F}_t}$ must be the shortest in column $k$. Hence, $h_d \in {\mathcal{F}_t}$ by definition of $\mathcal{F}_t$.

	Lemma~\ref{lem:optimal-succs-descending} then ensures that all the destination-row links belonging to a column closer than $k$ to $col(d)$ must be the shortest ones in their respective columns, and hence they all are in $\mathcal{F}_t$.
	That is, the destination-row links in $\mathcal{F}_t$ form a path $H_d = (z \dots d)$ where $z$ is an extreme of $h_d$. 

	Finally, there must exist a sequence of intra-orbit links from $x$ to $z$ because the last link of $H_s$ belongs to a column (i.e., $j$) that is next to the column (i.e., $k$) of the first link in $H_d$, by definition of $j$ and $k$.
\end{proof}

\begin{figure}[t]
	\centering
	\includegraphics[width=.75\columnwidth]{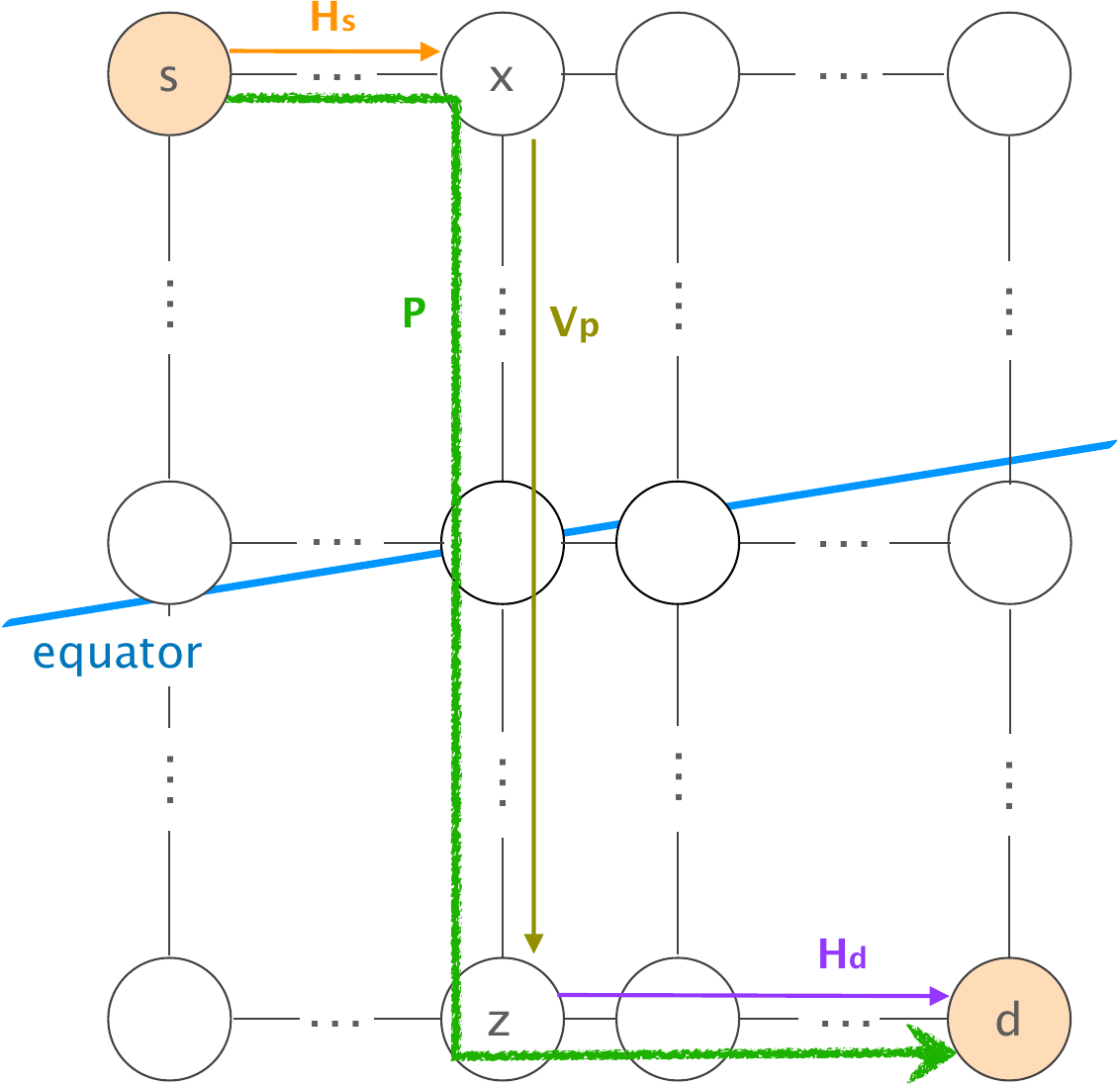}
	\caption{General form of the shortest path in type-A grids moving in the same direction.}
	\label{fig:spacegrid-equator-descending-sp}
\end{figure}

\Paste{theo:samedir-descending-sp}
\begin{proof}
	Consider the path $P$ as defined in this section.
	Lemma~\ref{lem:sp-exists-descending-crossing} ensures that $P$ exists and is a non-loopy path from $s$ to $d$.
	By construction, $P$ includes the shortest cross-orbit link in each column,	and a minimal number of intra-orbit links.
	Hence, no path in the grid can be shorter than $P$.
\end{proof}

%% file: theory-samedir-ascending.tex
\subsection*{Type-B grids}
We now switch to type-B grids.
As a reference, Figure~\ref{fig:spacegrid-ascending-crossing-equator-paths} displays one such grid, highlighting the paths $T$ and $B$ defined earlier in this section.

\begin{figure}[t]
	\centering
	\includegraphics[width=.7\columnwidth]{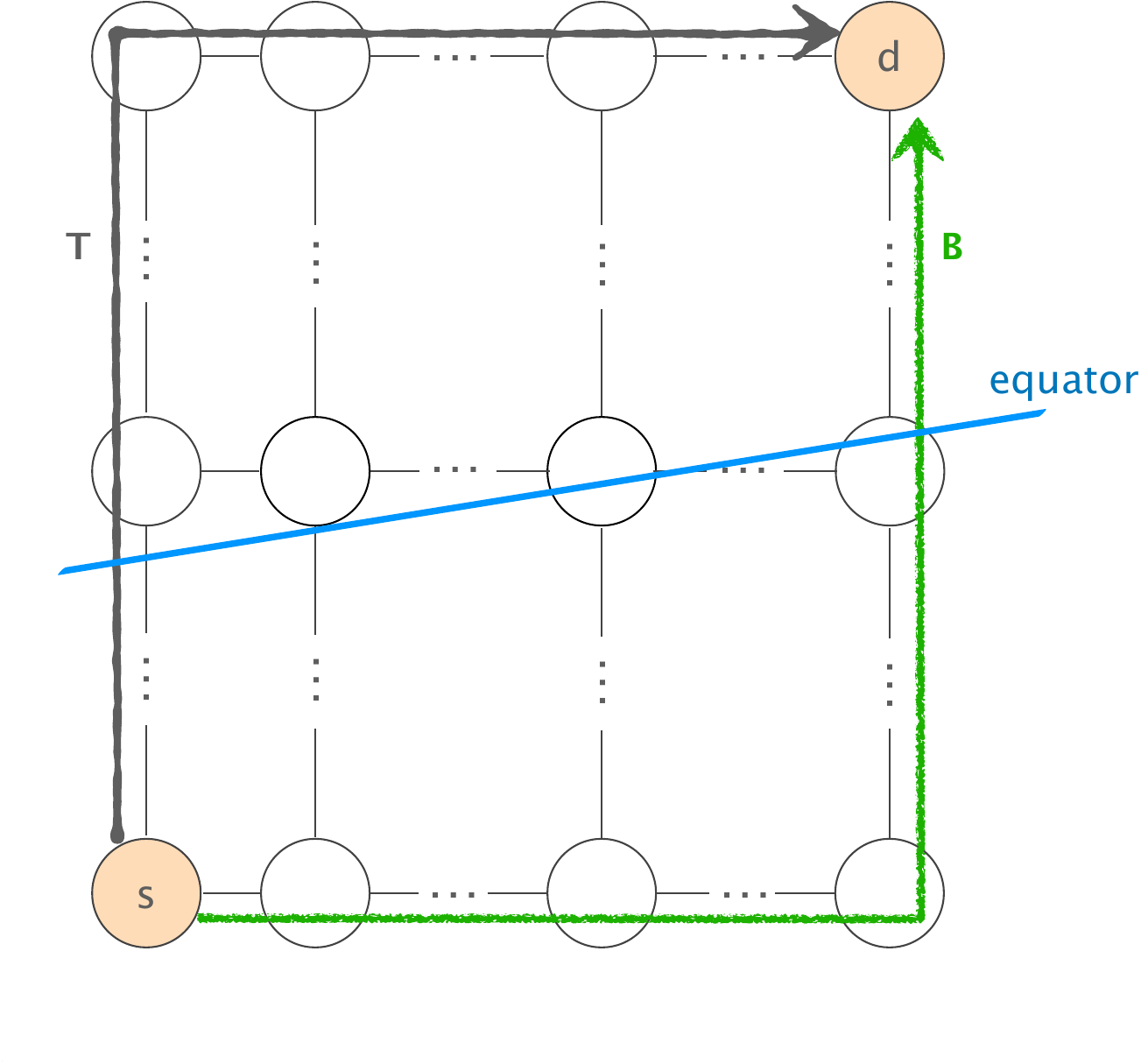}
	\caption{Example type-B grid, along with the two notable paths $T$ and $B$.}
	\label{fig:spacegrid-ascending-crossing-equator-paths}
\end{figure}

We start by showing that in any type-B grid moving in the same direction, the shortest path does not cross intra-orbit links between internal nodes -- or, equivalently, it never crosses more than one sequence of cross-orbit links, as stated by Lemma~\ref{lem:ascending-grid-crossing-no-middle-vlink}.
To this end, we first prove the following lemma.

\begin{mylemma}~\label{lem:ascending-grid-crossing-no-step-up}
	Consider any type-B grid at any time $t$. The shortest path from $s$ to $d$ cannot be the concatenation of a non-empty sequence of cross-orbit links in the source row, a non-empty sequence of intra-orbit links and a non-empty sequence of cross-orbit links in the destination row.
\end{mylemma}
\begin{proof}
	Let $Q$ be any path formed by exactly three sub-paths: a sequence of cross-orbit links in the source row, a sequence of cross-orbit links in the destination row, and a sequence of intra-orbit links connecting them.
	We refer the reader to Figure~\ref{fig:spacegrid-ascending-crossing-equator-no-zigzag} for the notation used in this proof.

	Consistently with Figure~\ref{fig:spacegrid-ascending-crossing-equator-no-zigzag}, we assume that the source row is the southmost row, and the destination row is the northmost one. This assumption is without loss of generality, because a symmetrical argument holds whenever $row(d)$ is the southmost row and $row(s)$ is the northmost one.

	By definition of the notable paths $B$ and $T$, the sequences of cross-orbit links in $Q$ must be sub-paths of $B$ and $T$.
	It must therefore be possible to write $Q$ as the concatenation of $B_1$, $V_z$ and $T_2$, where $B_1 = (s \dots u\ v)$ is a sub-path of $B$ starting from the source and including only cross-orbit links; $T_2 = (z\ w \dots d)$ is a sub-path of $T$ including only of cross-orbit links and terminating at the destination; and $V_z = (v \dots z)$ is a sequence of intra-orbit links connecting $B_1$ and $T_2$.

	According to this notation, $B$ must then be the concatenation of $B_1$ with another sequence $B_2$ of cross-orbit links in the source row, plus a minimal set of intra-orbit links.
	Similarly, $T$ must consist of a minimal sequence of intra-orbit links plus two sequences $T_1$ and $T_2$ of cross-orbit links.

	We now assume by contradiction that $Q = B_1\ V_z\ T_2$ is the shortest path from $s$ to $d$ at a given time $t$.
	The following relations must be true.

	\begin{itemize}
	\item For $Q$ to be the shortest path, we must have $len_t(u,v) \leq len_t(p,z)$, where $p$ is the destination-row node connected to $u$ with a sequence of intra-orbit links.
	Indeed, if $(u,v)$ is longer than $(p,z)$, we can build a path shorter than $Q$ by replacing its sub-path $(u,v) V_z$ with a sequence of intra-orbit links $(u \dots p)$ followed by $(p,z)$.

	\item For a similar reason as the one just above, $len_t(z,w) \leq len_t(v,x)$, where $x$ is the source-row node connected to $w$ with a sequence of intra-orbit links.

	\item 
	For $len_t(z,w) \leq len_t(v,x)$ to hold, $(z,w)$ cannot be between the equator and $row(s)$, otherwise $(z,w)$ would be closer to the equator than $(v,x)$.
	Property~\ref{prop:ascending-grids-hlink-sourcecol} then implies that $\delta_t(z,w) < \delta_t(p,z)$. By Property~\ref{prop:hlink-length}, we then have $len_t(p,z) < len_t(z,w)$.
	\end{itemize}

	Combining the above constraints on the considered links, we have: $len_t(u,v) \leq len_t(p,z) < len_t(z,w) \leq len_t(v,x)$.
	That is, $len_t(u,v) < len_t(v,x)$.

	If $len_t(u,v) < len_t(v,x)$, $(u,v)$ must be between the equator and $row(d)$ at time $t$.
	Indeed, by Property~\ref{prop:ascending-grids-hlink-destcol}, if $(u,v)$ is not between the equator and the destination row, then $(u,v)$ must be closer to the equator than $(v,x)$, given that the two links are in the same row and $(v,x)$ is closer to $col(d)$ than $(u,v)$.

	In turn, if $(u,v)$ is between the equator and $row(s)$, $(u,v)$ must be closer to the equator than the source-row link $(p,z)$ in the same column, and hence $len_t(u,v) > len_t(p,z)$. 

	This however leads to a contradiction (see first point in the above list), and hence yields the statement.
\end{proof}

\begin{figure}[t]
	\centering
	\includegraphics[width=.7\columnwidth]{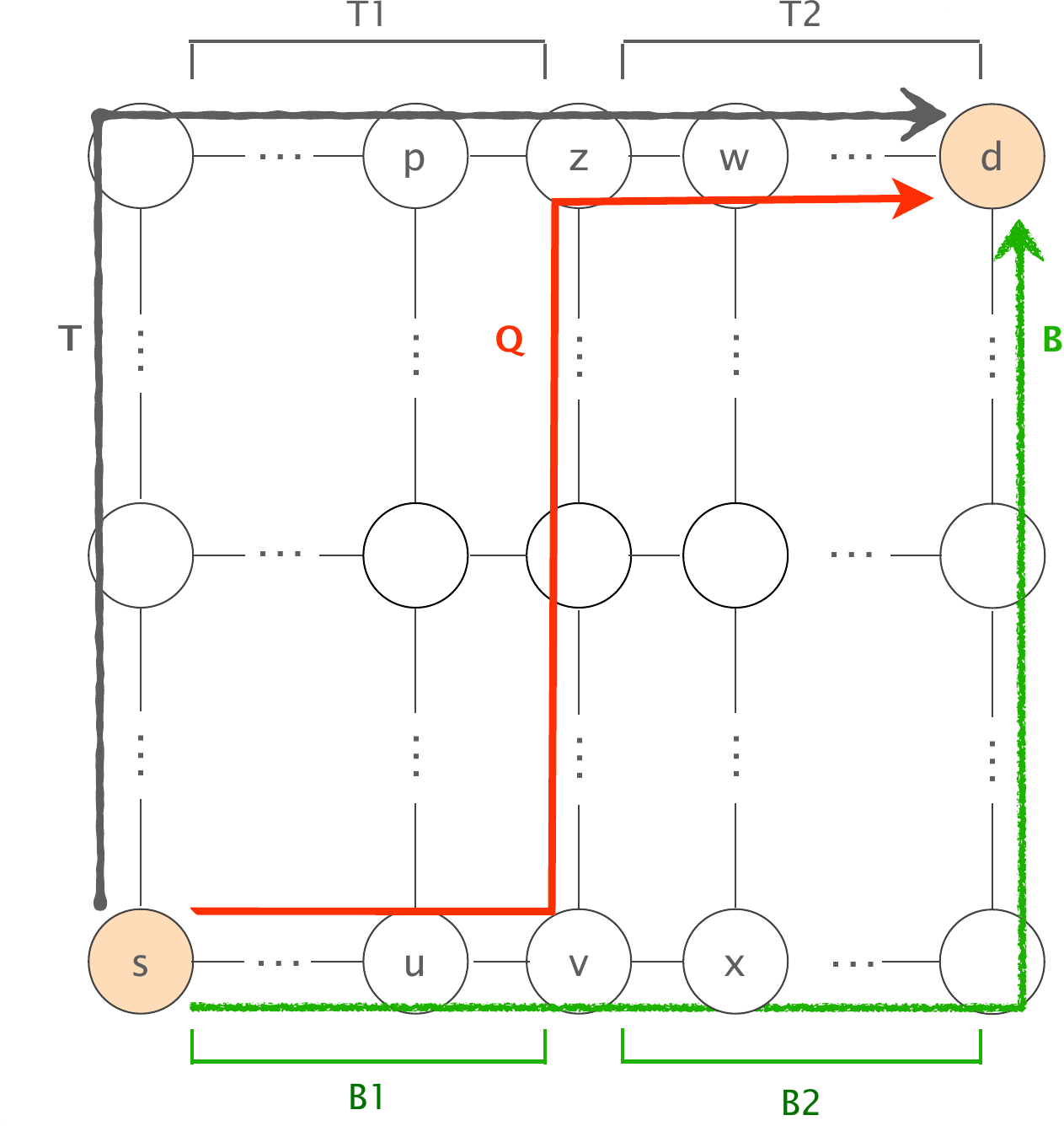}
	\caption{Visual support for the proof of Lemma~\ref{lem:ascending-grid-crossing-no-step-up}.}
	\label{fig:spacegrid-ascending-crossing-equator-no-zigzag}
\end{figure}

\begin{mylemma}\label{lem:ascending-grid-crossing-no-middle-vlink}
	In any type-B grid moving in the same direction, the shortest path from $s$ to $d$ always crosses exactly one sequence of cross-orbit links.
\end{mylemma}
\begin{proof}
	Assume by contradiction that the shortest path from $s$ to $d$ is a path $Q$ including at least two sequences of cross-orbit links.
	Let $H_1$ and $H_2$ respectively be the first and second sequence of cross-orbit links crossed by $Q$. 
	Also, let $l_1=(x,u)$ be the first link in $H_1$, and $l_2=(v,y)$ be the last link in $H_2$: $Q$ can then be written as $(s \dots x\ u \dots v\ y \dots d)$, where possibly $x=s$ and $y=d$.

	Lemma~\ref{lem:sp-no-step-down} guarantees that $Q$ cannot be the shortest path if $H_1$ is closer to the destination row than $H_2$.
	Hence, $H_1$ must be closer to the source row than $H_2$.

	Additionally, $col(l_1)$ must be closer to $col(s)$ than $col(l_2)$.
	Indeed, if $col(l_1)$ is not closer to $col(s)$ than $col(l_2)$, the sequence of intra-orbit links following $H_1$ in $Q$ would be closer to $col(d)$ than the sequence of intra-orbit links just before $H_2$: in this case, $Q$ would not be the shortest path by Lemma~\ref{lem:sp-no-step-aside}.

	Combining the observations above, we conclude that the sub-path $Q_{xy}$ of $Q$ from $x$ to $y$ must be the concatenation of $H_1$, a sequence of intra-orbit links, and $H_2$.

	Consider now the $x$-$y$ grid.
	Note that $G_{xy}$ is contained in the $s$-$d$ grid, $row(l_1)$ is closer than $row(l_2)$ to $row(s)$ and $col(l_1)$ is closer than $col(l_2)$ to $col(s)$.
	For this reasons, the $x$-$y$ grid is also a type-B grid moving in the same direction.

	In this case, however, Lemma~\ref{lem:ascending-grid-crossing-no-step-up} states that $Q_{xy}$ cannot be the shortest path from $x$ to $y$.
	We can therefore build an $s$-$d$ path $Q'$ shorter than $Q$, by replacing $Q_{xy}$ with the shortest path from $x$ to $y$.
	If $Q'$ contains loops, Lemma~\ref{lem:shortcut-loops} ensures that we can build an even shorter non-loopy path.
	The existence of such a path contradicts the hypothesis that $Q$ is the shortest path from $s$ to $d$, yielding the statement.
\end{proof}

The following lemma provides additional constraints on the possible shortest paths in type-B grids moving in the same direction.

\begin{mylemma}\label{lem:ascending-grid-crossing-restrict-paths}
	In any type-B grid $G$ moving in the same direction, the shortest path $P$ from $s$ to $d$ cannot have the first cross-orbit link closer to the equator than $(s,a)$ with $a \in G$.
	Also, $P$'s last cross-orbit link cannot be closer to the equator than $(b,d)$ with $b \in G$.
\end{mylemma}
\begin{proof}
	Let $Q$ be the shortest path from $s$ to $d$.

	Assume by contradiction that the first cross-orbit link $h_i$ that $Q$ crosses is closer to the equator than link $(s,a)$, with $a$ being the neighbor of $s$ in the source row.
	Let $q_1$ and $q_2$ be the two extremes of $h_i$, meaning that $h_i = (q_1,q_2)$.
	By definition of $h_i$, it must be possible to write $Q$ as the concatenation of a sequence of intra-orbit links $V_i = (s \dots q_1)$, the link  $h_i = (q_1,q_2)$, and a remaining sub-path $R$ from $q_2$ to $d$.

	Consider now the path $Q'$ from $s$ to $d$, which starts with the link $(s,a)$, continues with a sequence $V'$ of intra-orbit links from $a$ to $q_2$, and ends with $R$.
	Note that $V'$ must exist because $h_i$ and $(s,a)$ are both in $col(s)$, by definition of $h_i$.

	Let's now compare the length of $Q$ and $Q'$.
	By hypothesis, $h_i$ is closer to the equator than $(s,a)$; hence, $h_i$ is longer than $(s,a)$ by Property~\ref{prop:hlink-length}.
	By construction, $V_i$ and $V'$ have the same number of intra-orbit links, and hence they have the same length by Property~\ref{prop:vlink-length}.

	Since both $Q$ and $Q'$ terminate with the same sub-path $R$ from $q_2$ to $d$, $Q'$ must be shorter than $Q$.
	If $Q'$ includes loops, Lemma~\ref{lem:shortcut-loops} ensures that we can build a non-loopy path even shorter than $Q'$.
	The existence of a path shorter than $Q$ contradicts the hypothesis that $Q$ is the shortest path.

	A similar argument holds if we assume by contradiction that $Q$'s last cross-orbit link $h_f=(q_3,q_4)$ is closer to the equator than link $(b,d)$, where $b$ is the neighbor of $d$ in the destination row.
	In this case, we can indeed build a path $Q''$ shorter than $Q$ by replacing $Q$'s ending sub-path $(q_3\ q_4 \dots d)$ with a sequence of intra-orbit links from $q_3$ followed by $(b,d)$.
	As for the previous case, $Q''$ is shorter than $Q$ because $(b,d)$ is farther away from the equator than $h_f$, and the number of fixed-length intra-orbit links crossed by $Q''$ is the same as in $Q$.
	If $Q''$ includes loops, we can build an even shorter path by Lemma~\ref{lem:shortcut-loops}.
	The existence of such a path then contradicts again the hypothesis that $Q$ is the shortest path from $s$ to $d$.

	The above two arguments yield the statement.
\end{proof}

\begin{figure}[t]
	\centering
	\includegraphics[width=.9\columnwidth]{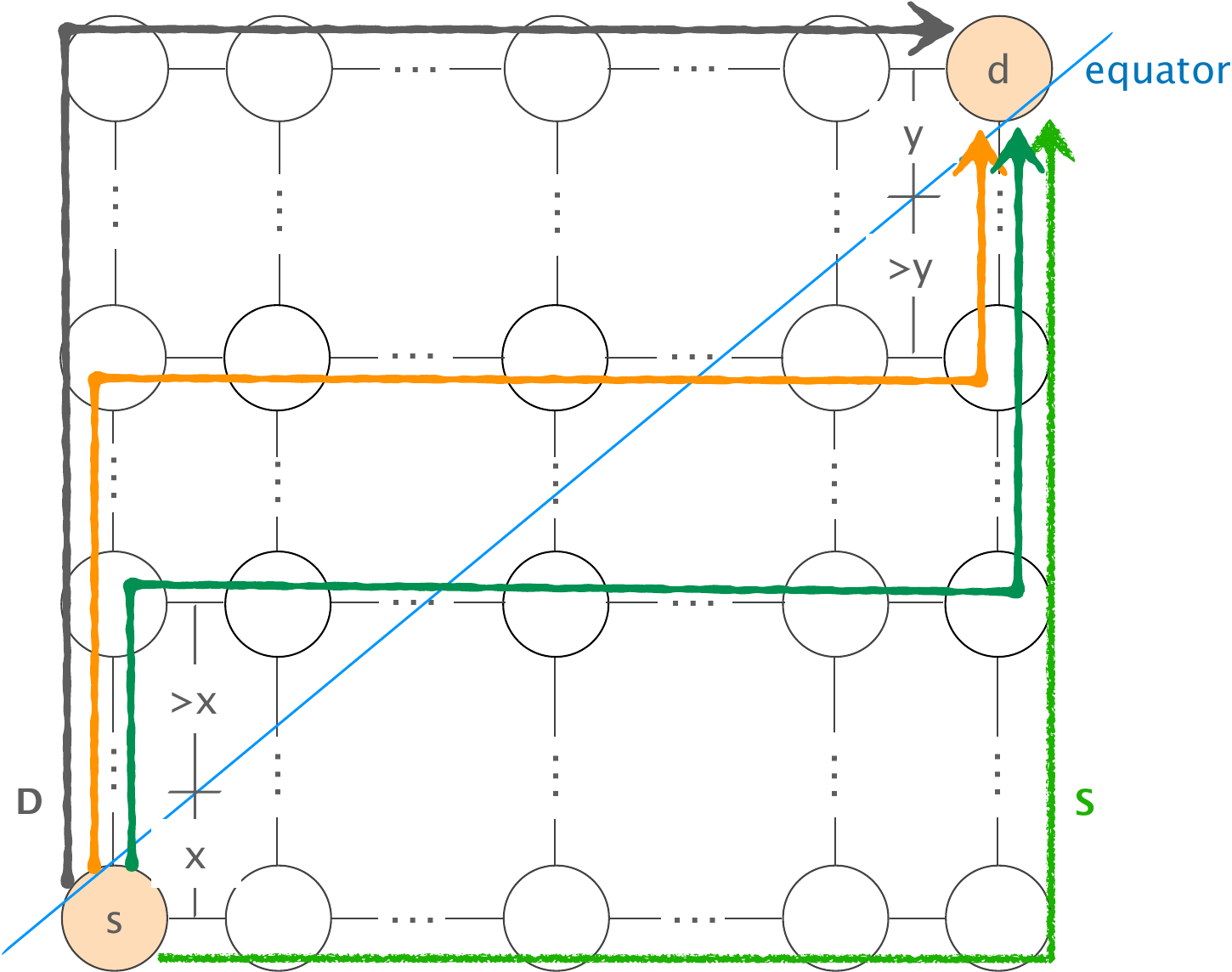}
	\caption{Candidate shortest paths in type-B grids moving in the same direction.}
	\label{fig:spacegrid-sp-ascending-crossing}
\end{figure}

The last two lemmas directly lead to Theorem~\ref{theo:sp-ascending-crossing}.
Figure~\ref{fig:spacegrid-sp-ascending-crossing} displays the final set of candidate shortest paths.

\Paste{theo:samedir-ascending-sp}
\begin{proof}
	The first condition is proved by Lemma~\ref{lem:ascending-grid-crossing-no-middle-vlink}, while Lemma~\ref{lem:ascending-grid-crossing-restrict-paths} ensures the final two conditions.
\end{proof}

%% file: theory-pole-descending.tex
\subsection*{Type-A grids}

\begin{figure}[t]
	\centering
	\includegraphics[width=.8\columnwidth]{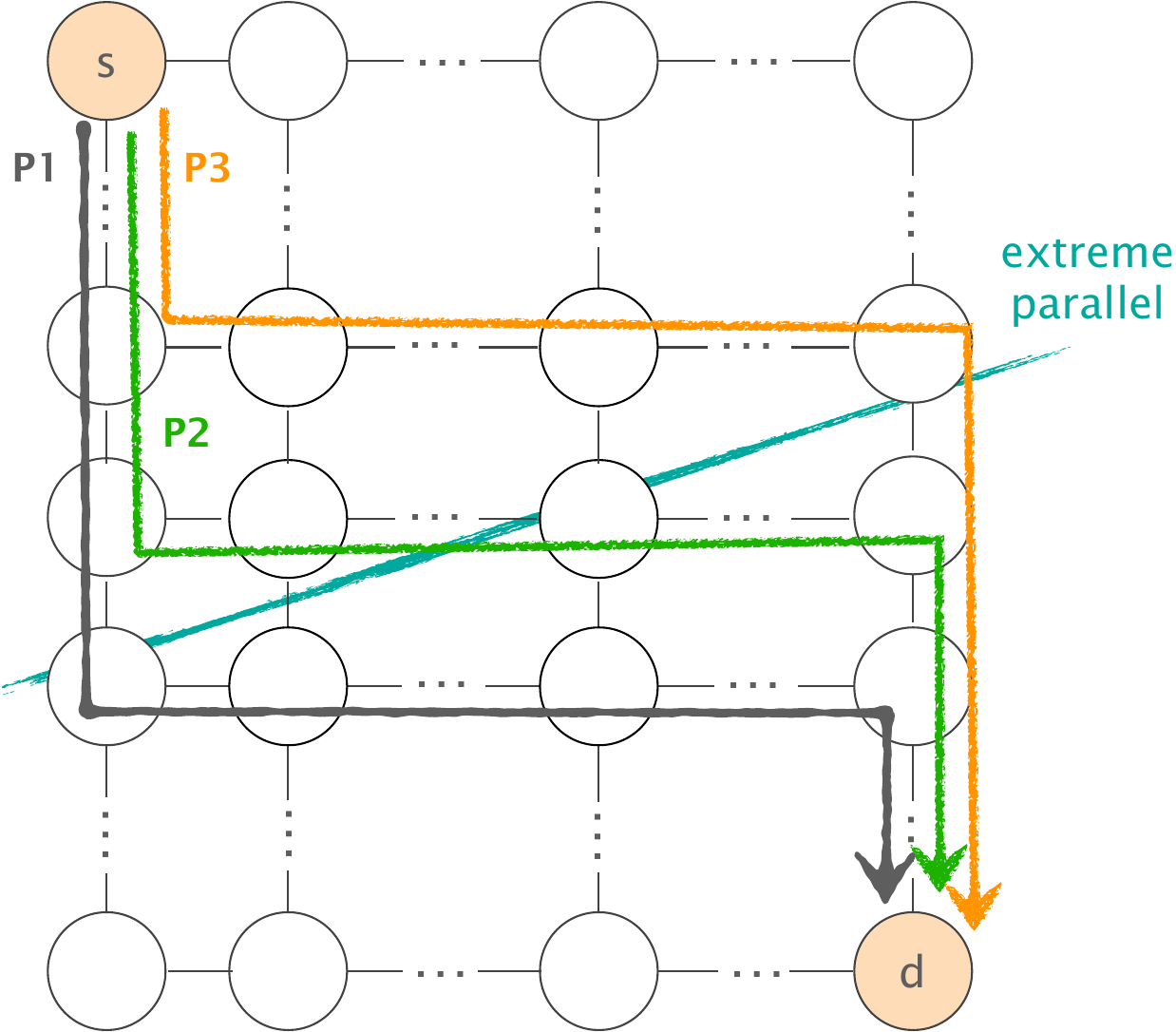}
	\caption{Candidate shortest paths when a type-A grid close to a single extreme parallel.}
	\label{fig:spacegrid-descending-crossing-pole}
\end{figure}

Throughout this section, we denote the set of candidate shortest paths in a type-A grid close to a single extreme parallel at a time $t$ as $\mathcal{S}_t$.
Figure~\ref{fig:spacegrid-descending-crossing-pole} visualizes $\mathcal{S}_t$ in an example grid.

Paths in $\mathcal{S}_t$ have two properties.
First, every path in $\mathcal{S}_t$ can be written as the concatenation of a possibly empty sequence of intra-orbit links, an entire row, and another possibly empty sequence of intra-orbit links.
Second, if we define $a$ and $b$ as the rows including the cross-orbit links closest to the extreme parallel among those in $col(s)$ and $col(d)$, then every row between $a$ and $b$, including both $a$ and $b$, is crossed by one path in $\mathcal{S}_t$.

An immediate consequence of the above definition is that every cross-orbit link which is the closest to the extreme parallel in its column is included in one path in $\mathcal{S}_t$.

We now prove that the shortest path from $s$ to $d$ in the grid at time $t$ is an element of the set $\mathcal{S}_t$.

\begin{mylemma}\label{lem:no-zigzag-descending-grid-wrapping}
	For any type-A grid moving in different directions and close to a single extreme parallel, the shortest path from $s$ to $d$ includes only cross-orbit links in the same row.
\end{mylemma}
\begin{proof}
	Assume by contradiction that at a given time $t$, the shortest path from $s$ to $d$ is a path $Z$ which includes at least two cross-orbit links belonging to two different rows. 

	By hypothesis, $Z$ must then include at least one sub-path $Z_j$ concatenating:
	a cross-orbit link $h_1 = (x_1,y_1)$;
	a cross-orbit link $h_2 = (x_2,y_2)$, such that $row(h1) \neq row(h_2)$; and
	a non-empty sequence $V_z$ of intra-orbit links connecting $y_1$ with $x_2$.
	Figure~\ref{fig:spacegrid-descending-pole-nozigzag} visualizes $Z_j$.

	For $Z$ to be the shortest path, $row(h_1)$ must be closer to the source row than $row(h_2)$ by Lemma~\ref{lem:sp-no-step-down}; also, $col(h_1)$ must be closer to the source column than $col(h_2)$ by Lemma~\ref{lem:sp-no-step-aside}.

	Let $a_1$ be the cross-orbit link in $row(h_1)$ and $col(h_2)$; also, let $a_2$ be the cross-orbit link in $row(h_2)$ and $col(h_1)$.
	The following properties hold.

	\begin{itemize}
		\item 
		We must have that $len_t(h_1) \leq len_t(a_2)$.
		Otherwise, if $h_1$ is longer than $a_2$, $Z_j$ is not the shortest path from $x_1$ to $y_2$.
		Consider indeed the path $Q_j$ where the sequence of intra-orbit links from $x_1$ to an extreme of $a_2$ is concatenated with $a_2$ and $h_2$: $Q_j$ is shorter than $Z_j$ because it includes the same number of fixed-length intra-orbit links as $Z_j$, the same link $h_2$, and $a_2$ instead of $h_1$. This would imply that $Z$ is not the shortest path because we can obtain a shorter path by replacing $Z_j$ with $Q_j$, yielding a contradiction.

		\item For a reason similar to the above one, we must have $len_t(h_2) \leq len_t(a_1)$.

		\item 
		For $len_t(h_1) \leq len_t(a_2)$ to hold, $a_2$ cannot be between the extreme parallel  and $h_1$; otherwise, $a_2$ would be closer to the extreme parallel than $h_1$, and hence shorter than $h_1$, by Property~\ref{prop:hlink-length-pole}.
		If $a_2$ is not between the extreme parallel  and $h_1$, it is also not between the extreme parallel and the source row. In this case, Property~\ref{prop:descending-grids-pole-hlink-destcol} then states that $len_t(a_2) < len_t(h_2)$.
	\end{itemize}

	Combining those observations on the considered link lengths, we have:
	$len_t(h_1) \leq len_t(a_2) < len_t(h_2) \leq len_t(a_1)$.
	That is, $len_t(h_1) < len_t(a_1)$.

	If $len_t(h_1) < len_t(a_1)$, then $a_1$ must be between the extreme parallel and the destination row.
	Indeed, Property~\ref{prop:descending-grids-pole-hlink-sourcecol} implies that if $a_1$ is not between the extreme parallel and the destination row, then $a_1$ is shorter than $h_1$.

	Since $row(h_2)$ is closer to $row(d)$ than $row(h_1) = row(a_1)$, then $a_1$ must also be between the extreme parallel and $h_2$. 
	This however implies that $len_t(a_1) < len_t(h_2)$ by Property~\ref{prop:hlink-length-pole}, which contradicts one of the observations above, and hence yields the statement.
\end{proof}

\begin{figure}[t]
	\centering
	\includegraphics[width=.45\columnwidth]{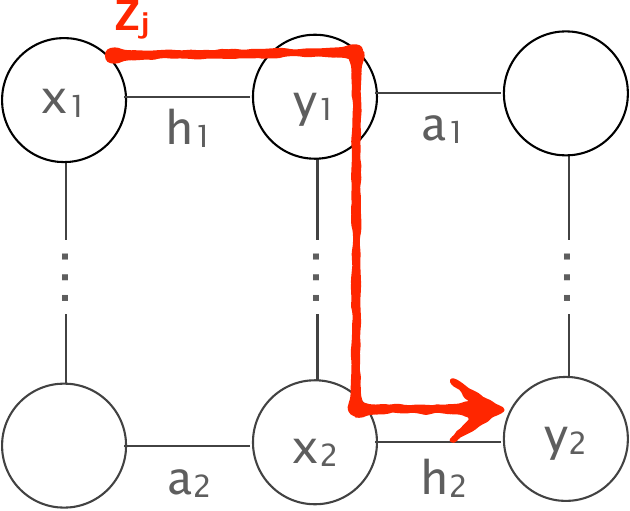}
	\caption{Visual support for the proof of Lemma~\ref{lem:no-zigzag-descending-grid-wrapping}.}
	\label{fig:spacegrid-descending-pole-nozigzag}
\end{figure}

\begin{mylemma}\label{lem:no-outside-rows-descending-grid-wrapping}
	For any type-A grid moving in different directions and close to a single extreme parallel, the shortest path from $s$ to $d$	never includes cross-orbit links in a row other than the ones traversed by paths in $\mathcal{S}_t$.
\end{mylemma}
\begin{proof}
	Assume by contradiction that at a given time $t$, the shortest path from $s$ to $d$ is a path $O$ which includes a sequence of cross-orbit links belonging to a row $o$ not crossed by any path in $\mathcal{S}_t$. By Lemma~\ref{lem:no-zigzag-descending-grid-wrapping}, $O$ must include all the cross-orbit links in row $o$.

	Remember that the paths in $\mathcal{S}_t$ cover all the rows between $a$ and $b$ (included): rows $a$ and $b$ cross the closest cross-orbit link to the extreme parallel in $col(s)$ and $col(d)$ respectively. 
	In other words, row $o$ cannot include any link crossed by the extreme parallel.

	This implies that only two possibilities exist for row $o$.
	The first possibility is that all the links in $o$ are between $row(s)$ and the extreme parallel: for example, any row between the source row and the one crossed by P3 in Figure~\ref{fig:spacegrid-descending-crossing-pole} would belong to this case.
	The second possibility is that all the links in $o$ are between $row(d)$ and the extreme parallel: for example, any row between the destination row and the one crossed by P1 in Figure~\ref{fig:spacegrid-descending-crossing-pole} would fall in this second category.

	\begin{figure}[t]
		\centering
		\includegraphics[width=\columnwidth]{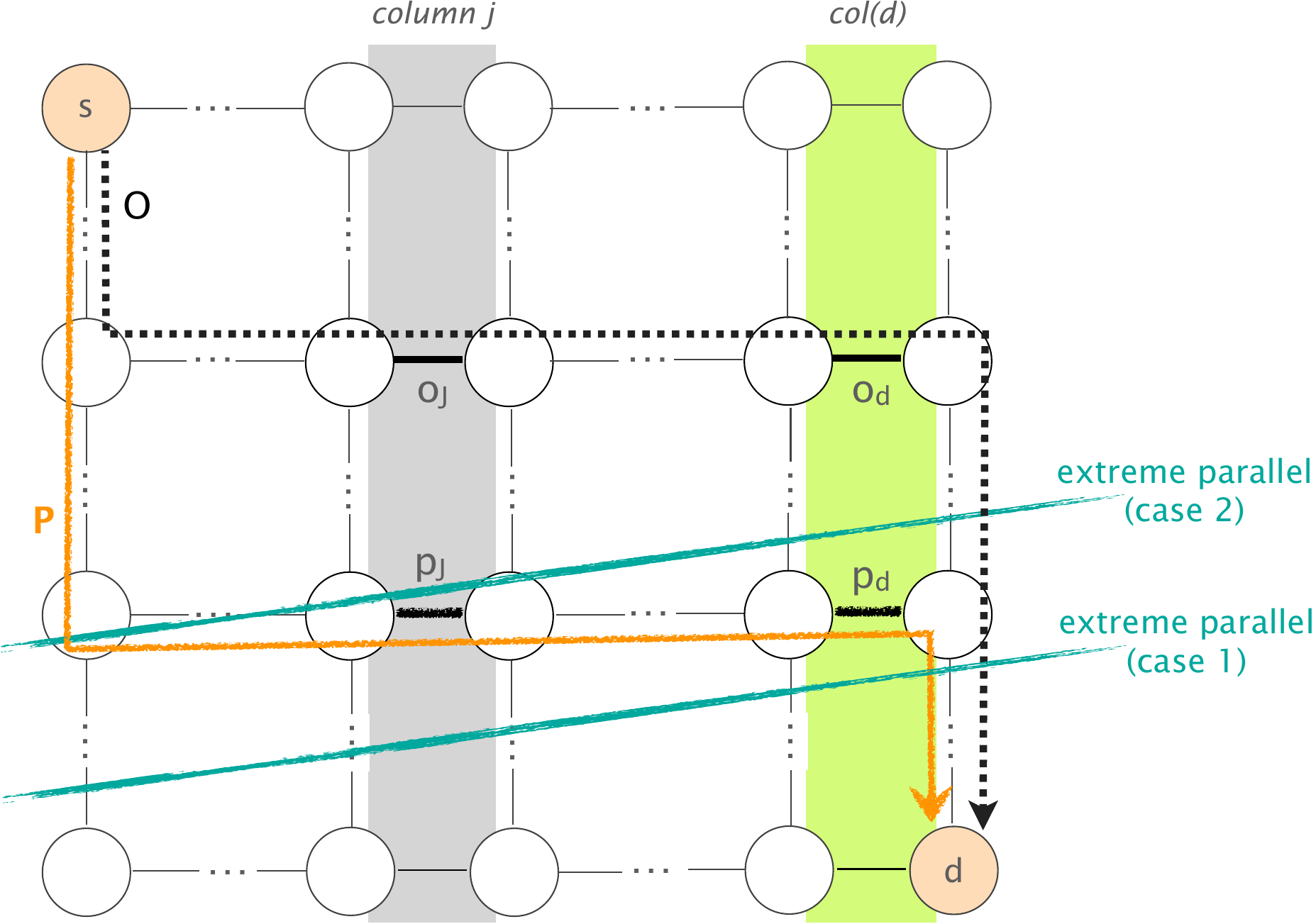}
		\caption{Visual support for the proof of Lemma~\ref{lem:no-outside-rows-descending-grid-wrapping}.}
		\label{fig:spacegrid-descending-pole-nooutrow}
	\end{figure}

	Assume for now that all the links in $o$ are between $row(s)$ and the extreme parallel, with possibly $o = row(s)$.
	Refer to Figure~\ref{fig:spacegrid-descending-pole-nooutrow} for a visualization of the following notation.

	Let $o_d$ be the cross-orbit link in $col(d)$ traversed by $O$.
	By definition, $\mathcal{S}_t$ contains a path $P$ including a cross-orbit link $p_d$ in $col(d)$ which is closer to the extreme parallel than $o_d$.
	By Property~\ref{prop:hlink-length-pole}, we then have $len_t(p_d) < len_t(o_d)$.

	Consider now any other column $j \neq col(d)$.
	Let $o_j$ be the link in column $j$ included in $O$, and $p_j$ be the link in column $j$ included in $P$.

	Note already that since $len_t(p_d) < len_t(o_d)$, $o_d$ cannot be between the extreme parallel and $p_d$. 
	Also, $row(p_d)$ is closer to $row(d)$ than $row(o_d)$ by assumption on $o$.

	Those two conditions imply that $o_d$ is not between the extreme parallel and the destination row.
	Since $row(o_d) = row(o_j)$ and $col(o_j)$ is closer to the source column than $col(o_d)$, Property~\ref{prop:descending-grids-pole-hlink-sourcecol} then asserts that $o_d$ is closer to the extreme parallel than $o_j$, and hence $len_t(o_d) < len_t(o_j)$.

	We then have the following two cases.
	\begin{itemize}
		\item if $p_j$ is not between the extreme parallel and the destination row (i.e., see case 1 in Figure~\ref{fig:spacegrid-descending-pole-nooutrow}), then the $p_j$ must be between the extreme parallel and the source row. Since $row(o_j)$ is closer to $row(s)$ than $row(p_j)$ by assumption on $o$, $p_j$ must be closer to the extreme parallel than $o_j$, hence $len_t(p_j) < len_t(o_j)$.

		\item if $p_j$ is between the extreme parallel and the destination row (i.e., see case 2 in Figure~\ref{fig:spacegrid-descending-pole-nooutrow}), then $p_j$ cannot be between the extreme parallel and the source row. Property~\ref{prop:descending-grids-pole-hlink-destcol} then implies that $p_j$ is closer to the extreme parallel than $p_d$, and hence $len_t(p_j) < len_t(p_d)$.
		Since we have already shown that $len_t(p_d) < len_t(o_d)$ and $len_t(o_d) < len_t(o_j)$, 
		we have that $len_t(p_j) < len_t(o_j)$.
	\end{itemize}

	In both cases, $p_j$ is shorter than $o_j$.
	This means that for every column, $P$ traverses a cross-orbit link shorter than the corresponding one included in $O$.

	Since all paths in $\mathcal{S}_t$ include a minimal set of fixed-size intra-orbit links, we conclude that $P$ is shorter than $O$, contradicting the hypothesis that $O$ is the shortest path.

	The statement follows by noting that a symmetric argument holds for paths not in $\mathcal{S}_t$ that cross rows between the extreme parallel and the destination row.
\end{proof}

\Paste{theo:diffdir-descending}
\begin{proof}
	The statement claims that the shortest path is an element of $\mathcal{S}_t$.
	We prove it by considering any path $Q$ which is not part of $\mathcal{S}_t$.
	For $Q$, one of the following cases must hold.
	\begin{itemize}
		\item $Q$ crosses links in two or more rows. Lemma~\ref{lem:no-zigzag-descending-grid-wrapping} then ensures that $Q$ cannot be the shortest path.
		\item $Q$ crosses a row that is not traversed by any path in $\mathcal{S}_t$. Lemma~\ref{lem:no-outside-rows-descending-grid-wrapping} proves that $Q$ cannot be the shortest path.
	\end{itemize}

	In both cases, $Q$ cannot be the shortest path from $s$ to $d$, yielding the statement.
\end{proof}

%% file: theory-pole-ascending.tex
\subsection*{Type-B grids}

Consider again the set $\mathcal{F}_t$ defined in \S\ref{app:samedir}.
As a reminder, $\mathcal{F}_t$ is the set of cross-orbit links that are the farthest one from the equator in their respective columns.
Note that for grids moving in different directions and close to a single extreme parallel, $\mathcal{F}_t$ can also be defined as the set of cross-orbit links closest to the extreme parallel in their respective columns.
Contrary to grids moving in the same direction, $\mathcal{F}_t$ is not guaranteed to include only source-row and destination-row links -- it may actually contain only middle links.

We now prove that at any time $t$, any type-B grid moving in different direction and close to a single extreme parallel admits a path $P$ from $s$ to $d$ crossing all and only the links in $\mathcal{F}_t$ plus a minimum number of intra-orbit links.
We then prove that $P$ is the shortest path in the grid at $t$.
An illustration of such a path $P$ is displayed in Figure~\ref{fig:spacegrid-ascending-crossing-pole-spshape}.

\begin{figure}[t]
	\centering
	\includegraphics[width=.8\columnwidth]{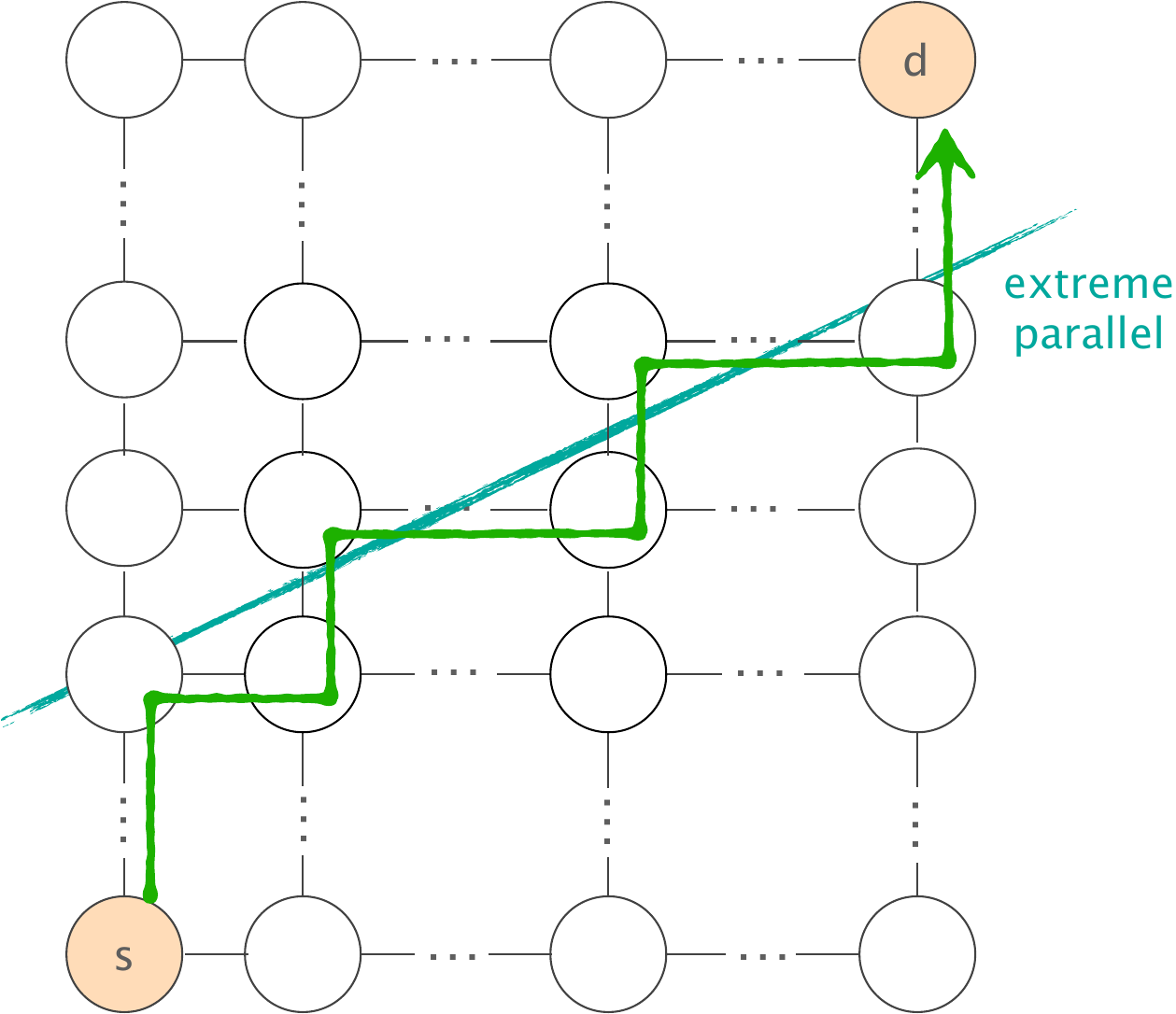}
	\caption{Visualization of the shortest path in a type-B grid close to a single extreme parallel.}
	\label{fig:spacegrid-ascending-crossing-pole-spshape}
\end{figure}

\begin{mylemma}\label{lem:ascending-pole-no-step-down}
	Consider any time $t$ when an ascending grid moves in different directions and is close to a single extreme parallel. For any two links $h_1$ and $h_2$ in $\mathcal{F}_t$, if $col(h_1)$ is closer to the source column than $col(h_2)$, then $row(h_1)$ is closer to the source row than $row(h_2)$.
\end{mylemma}
\begin{proof}
	Assume by contradiction that $row(h_2)$ is closer to the source row than $row(h_1)$.
	Let $a_1$ be the cross-orbit link in $row(h_1)$ and $col(h_2)$, and let $a_2$ be the link in $row(h_2)$ and $col(h_1)$.
	We refer to Figure~\ref{fig:spacegrid-ascending-crossing-pole-lemma-support} for a visualization of the notation used in this proof.

	One of the following cases must hold.
	\begin{itemize}
		\item $a_1$ is between $h_2$ and the extreme parallel.
		In this case, then $a_1$ is closer to the extreme parallel than $h_2$, which contradicts the hypothesis that $h_2 \in \mathcal{F}_t$.
		\item $a_2$ is between $h_1$ and the extreme parallel.
		In this case, $a_2$ is closer to the extreme parallel than $h_1$, which contradicts the hypothesis that $h_1 \in \mathcal{F}_t$.
		\item $a_1$ is not between $h_2$ and the extreme parallel, and $a_2$ is not between $h_1$ and the extreme parallel.\\
		Since we assumed that $row(h_2)$ is closer to $row(s)$ than $row(h_1)$ and $row(a_1) = row(h_1)$, then $a_1$ cannot be between the extreme parallel and the source row. Property~\ref{prop:ascending-grids-pole-hlink-sourcecol} implies that $a_1$ is closer to the extreme parallel than $h_1$, and hence $len_t(a_1) < len_t(h_1)$.
		Similarly, $a_2$ cannot be between the extreme parallel and the destination row, which implies $len_t(a_2) < len_t(h_2)$, by Property~\ref{prop:ascending-grids-pole-hlink-destcol}.\\
		Also, since $h_1,h_2 \in \mathcal{F}_t$ by hypothesis, $h_1$ must be closer to the extreme parallel than $a_2$, and $h_2$ must be closer to the extreme parallel than $a_1$. We must therefore have $len_t(h_1) \leq len_t(a_2)$ and $len_t(h_2) \leq len_t(a_1)$.\\
		This set of inequalities generates a contradiction, as we have $len(a_1) < len(h_1) \leq len(a_2) < len(h_2) \leq len(a_1)$, that is, $len(a_1) < len(a_1)$.
	\end{itemize}

	In all the cases, we end up with a contradiction, which yields the statement.
\end{proof}

\begin{figure}[t]
	\centering
	\includegraphics[width=.95\columnwidth]{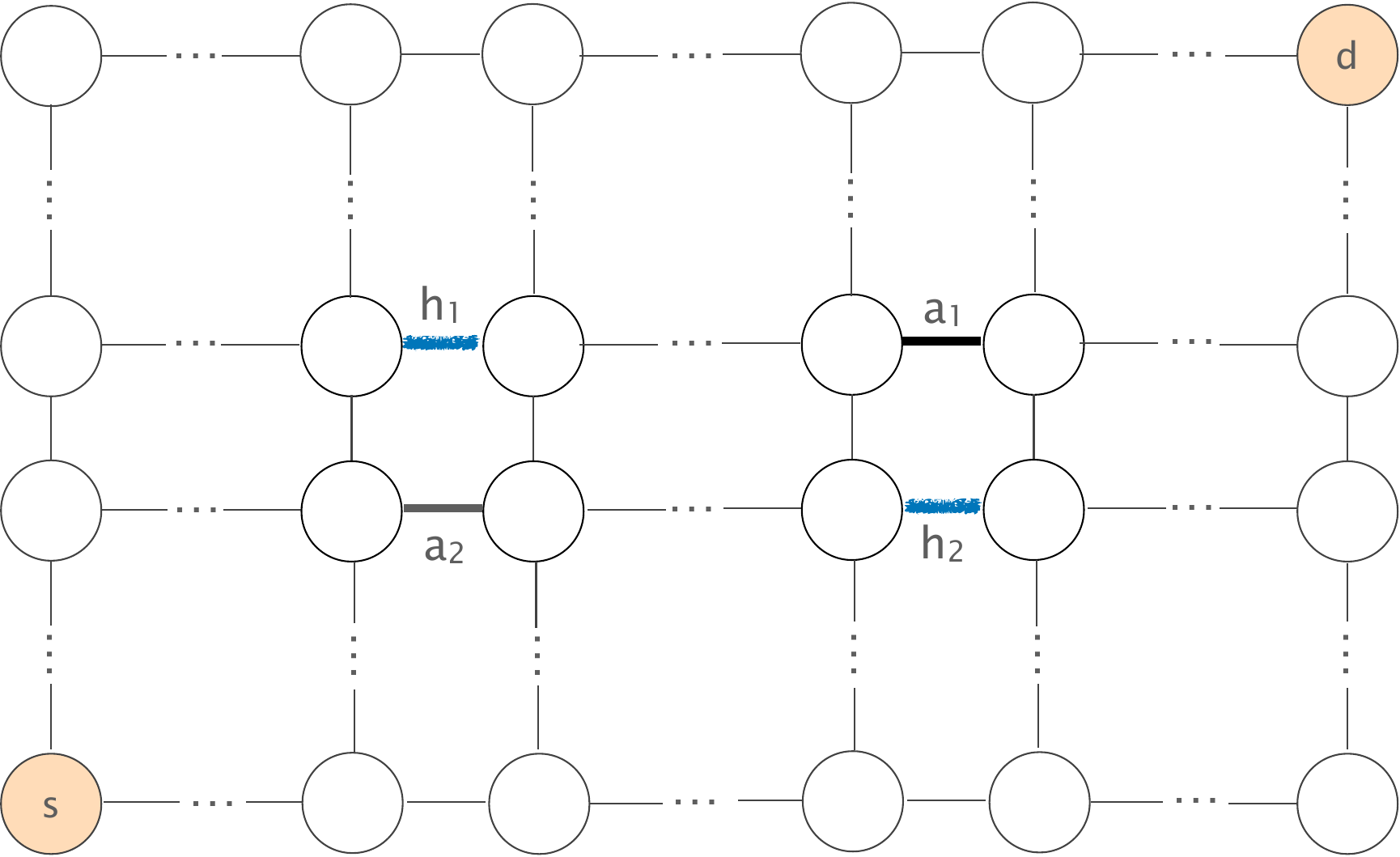}
	\caption{Visual support for the proof of Lemma~\ref{lem:ascending-pole-no-step-down}.}
	\label{fig:spacegrid-ascending-crossing-pole-lemma-support}
\end{figure}

\Paste{theo:diffdir-ascending}
\begin{proof}
	Let $P$ be the $s$-$d$ path crossing all and only the links in $\mathcal{F}_t$ plus minimal sequences of intra-orbit links connecting pairs of links in $\mathcal{F}_t$ belonging to adjacent columns. 

	Such a path $P$ is guaranteed to exist. Indeed, $\mathcal{F}_t$ contains exactly one cross-orbit link per column.
	Also, for any two links $l_1, l_2 \in \mathcal{F}_t$ such that $col(l_1)$ and $col(l_2)$ are adjacent, we can always connect $l_1$ and $l_2$ with a single non-loopy sequence of intra-orbit links belonging to the grid.

	We now show that $P$ has a minimal number of intra-orbit links.
	Let $[h_1,h_2,\dots,h_n]$ be the sequence of links in $\mathcal{F}_t$ ordered from the one in $col(s)$ to the one in $col(d)$.
	Lemma~\ref{lem:ascending-pole-no-step-down} states that for each $i = 1,\dots,n-1$ either $row(h_{i+1}) = row(h_{i})$ or $h_{i+1}$ is closer to the destination row than $h_i$. This means that $P$ always traverses intra-orbit links from a node closer to the source row to a node closer to the destination row. Hence, $P$ contains a number of intra-orbit links equal to the number of rows in the grid minus one, which is indeed minimal for any path from $s$ to $d$.

	The statement then follows by noting that $P$ crosses a set of cross-orbit links, which are minimal in number (as they are one per column) and are the shortest in their respective columns, by definition of $\mathcal{F}_t$.
\end{proof}

%% file: theory-frr.tex
\section{Fast Rerouting}
\label{app:frr}

We now prove effectiveness and efficiency of our fast rerouting scheme in the case of single-link and single-node failures.
We refer to Figure~\ref{alg:sat-logic} for a precise definition of the forwarding operations performed by every satellite in our scheme.

\Paste{theo:frr}
\begin{proof}
	Consider any packet $p$ sent by the source GS to a satellite $s$ with a tag list ${\rm T}_p$ 
	which encodes a path $P = (s \dots d)$.
	By hypothesis, $P$ crosses a failed link $(x,y)$.
	We therefore can write $P$ as $(s \dots x\ y \dots d)$, where possibly $x=s$ or $y=d$.

	Assume for now that $(x,y)$ is a cross-orbit link.
	By definition of satellites' forwarding operations in our scheme (see Figure~\ref{alg:sat-logic}), we have the following cases.

	\begin{itemize}
		\item the tag list ${\rm T}_p$ contains at least two unconsumed tags when $x$ decides the next-hop for $p$.

		Since $(x,y)$ is a cross-orbit link, the current tag $\tau_x$ in ${\rm T}_p$ must map to a subpath $P_x$ of $P$ including only cross-orbit links, and the next tag $\tau_{x+1}$ must map to a sequence $P_t$ of intra-orbit links.

		According to lines 10 and 11 in Figure~\ref{alg:sat-logic}, $x$ then forwards $p$ over an intra-orbit link $(x,u)$, without modifying the list of tags but decrementing the next unconsumed tag by one.
		Since $u$ is selected according the direction of $\tau_{x+1}$, $u$ must be topologically closer than $x$ to $d$ by Lemma~\ref{lem:sp-no-step-down}: the same lemma guarantees that $u$ is not the previous hop of $x$ in $P$.

		Since the tag list is unmodified, $u$ then forwards $p$ to its neighbor $v$ in the same orbit as $y$.
		Note that $(p,v)$ is working, given that $(x,y)$ is the only failed link by hypothesis. 
		Since $x$ does not modify the list of tags, $p$ is forwarded over cross-orbit links until it reaches a satellite $t$ in $P_t$, with $t=v$ if $y$ in $P_t$.
		From that point on, the packet is forwarded according to the original path, given that the list of tags is unmodified and the next unconsumed tag is decremented by one.

		Overall, $p$ is rerouted over a path $Q$ which is equal to $P$ except that $Q$'s sub-path from $x$ to $t$ consists of an intra-orbit link $(x,u)$ followed by a sequence of cross-orbit links $(u \dots t)$ instead of a sequence of cross-orbit links followed by one intra-orbit link. 

		This implies that the packet is delivered to $d$, the hop stretch is zero in this case, and the delay increase is equal to the difference between sequences of cross-orbit links in adjacent chains.

		\item the tag list ${\rm T}_p$ contains only one unconsumed tag when $x$ decides the next-hop for $p$. 

		Since $(x,y)$ is a cross-orbit link, the current tag $\tau_x$ in ${\rm T}_p$ must map to a subpath $P_x$ of $P$ including only cross-orbit links. Therefore, the last consumed tag must have encoded a sequence of intra-orbit links.

		According to line 13 in Figure~\ref{alg:sat-logic}, $x$ sends $p$ over an intra-orbit link $(x,u)$.
		Once again, Lemma~\ref{lem:sp-no-step-down} guarantees that $u$ is not the previous hop of $x$ in $P$.

		Additionally, line 14 in Figure~\ref{alg:sat-logic} ensures that a tag is appended to ${\rm T}_p$ so that $p$ is eventually forwarded to a node in the row of $x$, which is also the row of $d$.

		Overall, $p$ is rerouted over a path $Q$ which is equal to $P$ from $s$ to $x$, and replaces the last sequence $(x \dots d)$ of cross-orbit links with an intra-orbit link $(x,u)$ followed by a sequence of cross-orbit links terminating at a node $t$ in the same orbit of $d$ plus an intra-orbit link $(t,d)$. 

		The packet is therefore delivered to $d$, with a stretch of two hops, and a delay increase equal to the length of two intra-orbit links.
	\end{itemize}

	If $(x,y)$ is an intra-orbit link, a symmetric argument guarantees that the rerouted packet is delivered to $d$ either with
	\begin{inparaenum}[(i)]
		\item no hop stretch and no delay increase, given that all intra-orbit links always have the same length, or
		\item a stretch of two hops and a delay increase equal to two cross-orbit links.
	\end{inparaenum}

	The statement follows by noting that the path $Q$ followed by the rerouted packet does not include $y$ in any of the above cases: this implies that the rerouting path $Q$ followed by $p$ would not include any failed link or node even if the entire node $y$ is failed.
\end{proof}